\documentclass{amsart}

\usepackage[british]{babel}
\usepackage{enumitem}
\usepackage[dvipsnames]{xcolor}
\usepackage{mlmodern} 
\usepackage[T1]{fontenc}

\usepackage[leqno]{amsmath}
\usepackage{amssymb}
\usepackage{amsthm}
\usepackage{mathtools}
\usepackage{stmaryrd}
\usepackage[mathcal]{euscript}
\usepackage{tikz-cd}
\usetikzlibrary{decorations.pathmorphing,decorations.pathreplacing,decorations.markings,patterns,patterns.meta}

\usepackage{caption}
\usepackage{subcaption}
\usepackage{graphicx}

\definecolor{LinkColor}{HTML}{697B8C}
\definecolor{CiteColor}{HTML}{8C6B69}
\definecolor{URLColor}{HTML}{898C69}
\usepackage[colorlinks=true,linkcolor=LinkColor,citecolor=CiteColor,urlcolor=URLColor]{hyperref}
\usepackage[capitalise,nameinlink,noabbrev]{cleveref}

\makeatletter
\newcommand{\dummylabel}[2]{\def\@currentlabel{#2}\label{#1}}
\makeatother

\usepackage[backend=bibtex,style=alphabetic,defernumbers=true,doi=false,isbn=false,url=false,useprefix=true,sorting =anyt,backref=true,giveninits=true,maxnames=99,maxalphanames=99]{biblatex}
\bibliography{bibliography.bib}

\theoremstyle{plain}
\newtheorem{theorem}{Theorem}[section]   
\newtheorem*{theorem*}{Theorem}
\newtheorem{proposition}[theorem]{Proposition}
\newtheorem{corollary}[theorem]{Corollary}
\newtheorem{lemma}[theorem]{Lemma}

\theoremstyle{definition}
\newtheorem{definition}[theorem]{Definition}
\newtheorem{example}[theorem]{Example}
\newtheorem{construction}[theorem]{Construction}
\theoremstyle{remark}
\newtheorem{remarknumbered}[theorem]{Remark}
\newtheorem*{remark}{Remark}

\crefname{construction}{Construction}{Constructions}

\usepackage{todonotes}

\newcommand{\Opens}{\mathop{\mathcal{O}\mspace{-4mu}}}

\newcommand{\Powerset}{\mathop{\mathcal{P}\mspace{-1mu}}}

\DeclareMathOperator{\im}{im}
\DeclareMathOperator{\id}{id}

\newcommand{\sqleq}{\sqsubseteq} 
\newcommand{\sqgeq}{\sqsupseteq}

\newcommand{\op}{{\mathrm{op}}}

\newcommand{\refines}{\Subset}

\renewcommand{\leq}{\leqslant}
\renewcommand{\geq}{\geqslant}
\renewcommand{\epsilon}{\varepsilon}

\newcommand{\Leq}{\mathrel{\trianglelefteqslant}}
\newcommand{\LeqU}{\mathrel{\trianglelefteqslant_\mathrm{U}}}

\newcommand{\up}{\mathop{\uparrow}\!}
\newcommand{\down}{\mathop{\downarrow}\!}

\newcommand{\tridown}{\text{\rotatebox[origin=c]{90}{$\triangleleft$}}}

\usepackage{trimclip}
\makeatletter
\newcommand{\tworightarrow}{\mathrel{\text{\two@rightarrow}}}
\newcommand{\two@rightarrow}{%
	\sbox0{$\m@th\rightarrow$}%
	\smash{\rlap{\kern0.1\wd0\clipbox{{.3\width} {-\height} 0pt {-\height}}{$\m@th\rightarrow$}}}%
	$\m@th\rightarrow$%
}
\makeatletter
\newcommand{\twoleftarrow}{\mathrel{\text{\two@leftarrow}}}
\newcommand{\two@leftarrow}{%
	\sbox0{$\m@th\leftarrow$}%
	\smash{\rlap{\kern0.1\wd0\clipbox{{-.1\width} {-\height} 3pt {-\height}}{$\m@th\leftarrow$}}}%
	$\m@th\leftarrow$%
}

\newcommand{\preUp}{\scalebox{1.2}[1.05]{\rotatebox[origin=c]{90}{$\rightarrowtriangle$}}}
\newcommand{\preDown}{\scalebox{1.2}[1.05]{\rotatebox[origin=c]{-90}{$\rightarrowtriangle$}}}

\newcommand{\Up}{\text{\preUp}}
\newcommand{\Down}{\text{\preDown}}
\newcommand{\Upsub}[1]{\Up_{\mspace{-3.5mu}#1}}
\newcommand{\Downsub}[1]{\Down_{\mspace{-1.5mu}#1}}

\newcommand{\Cov}{\mathrm{Cov}}
\newcommand{\CovLeq}{\Cov_{\Leq}}

\newcommand{\Covcaus}{\mathrm{Caus}}
\newcommand{\Covchron}{\mathrm{Chron}}
\newcommand{\cov}{\triangleleft}

\newcommand{\sqleqdown}{\mathrel{\sqsubset_\tridown}}
\newcommand{\Dcaus}{D_{\caus}}
\newcommand{\Dchron}{D_{\chron}}
\newcommand{\DCaus}{D_\Covcaus}
\newcommand{\DChron}{D_\Covchron}
\newcommand{\DLeq}{D_{\Leq}}

\newcommand{\cat}[1]{\mathbf{#1}}
\newcommand{\Set}{\cat{Set}}
\newcommand{\Top}{\cat{Top}}
\newcommand{\Loc}{\cat{Loc}}

\DeclareMathOperator{\loc}{Loc}
\DeclareMathOperator{\Kleisli}{Kl}

\newcommand{\caus}{\mathrel{\preccurlyeq}}

\DeclareFontFamily{U}{matha}{\hyphenchar\font45}
\DeclareFontShape{U}{matha}{m}{n}{
	<5> <6> <7> <8> <9> <10> gen * matha
	<10.95> matha10 <12> <14.4> <17.28> <20.74> <24.88> matha12
}{}
\DeclareSymbolFont{matha}{U}{matha}{m}{n}
\DeclareMathSymbol{\prechron}{3}{matha}{'316} 
\newcommand{\chron}{\mathrel{\prechron}}

\makeatletter
\newcommand{\newparallel}{\mathrel{\mathpalette\new@parallel\relax}}
\newcommand{\new@parallel}[2]{%
	\begingroup
	\sbox\z@{$#1T$}
	\resizebox{!}{\ht\z@}{\raisebox{\depth}{$\m@th#1/\mkern-5mu/$}}%
	\endgroup
}
\makeatother

\usepackage[most]{tcolorbox}
\definecolor{HighlightColor}{HTML}{DDFFCD}
\tcbset{on line, toprule=-2pt, bottomrule=-2pt,
	boxsep=2pt, left=0pt,right=0pt,top=0pt,bottom=0pt,
	colframe = white, colback=HighlightColor,  
	highlight math style={enhanced}
}
\newcommand{\highlighttikzcd}[1]{\tcbhighmath{#1}}

\title[Causal Coverage]{Causal Coverage in\\Ordered Locales and Spacetimes}
\author{Chris Heunen}
\email{chris.heunen@ed.ac.uk}
\address{School of Informatics, University of Edinburgh, Edinburgh EH8 9AB, United Kingdom}
\author{Nesta van der Schaaf}
\email{nesta.van-der-schaaf@inria.fr}
\address{QuaCS, INRIA de Saclay, Laboratoire Méthodes Formelles,
	91190 Gif-sur-Yvette, France}
\thanks{This paper extends Section~4.5 and Chapter~8 of the second-named author's PhD thesis~\cite{schaaf2024TowardsPointFreeSpacetimes}.\\
This work has been partially funded by the French National Research Agency (ANR) within the framework of ``Plan France 2030'', under the research projects EPIQ ANR-22-PETQ-0007, HQI-Acquisition ANR-22-PNCQ-0001 and HQI-R\&D ANR-22-PNCQ-0002.}

\begin{document}
	\begin{abstract}
		We develop relativistic causality theory in the setting of point-free topology by introducing a notion of \emph{causal coverage} in ordered locales, generalising their canonical coverage relation to incorporate causal structure. 
		This improves Christensen and Crane's construction of `causal sites'~\cite{christensen2005CausalSitesQuantum}. We connect to sheaf theory by showing that causal coverages can be interpreted as a generalised  Grothendieck topology, and the sheaf condition as a type of deterministic time evolution. To develop these notions, we introduce and study \emph{parallel ordered locales}.
		Causal coverage naturally induces a notion of \emph{domain of dependence}. Comparing the localic and curve-wise definitions in spacetimes, the localic domains strictly contain the classical ones.
	\end{abstract}
	\maketitle


\section{Introduction}
\label{section:introduction}
This paper develops notions from relativistic causality theory in the point-free setting of \emph{locales}: a notion of \emph{causal coverage} (in \cref{section:causal coverages ordered locales}), which gives rise to a generalised notion of \emph{domain of dependence} \cite{geroch1970DomainofDependence} (in \cref{section:domains of dependence}). We compare the point-free and classical notions in the setting of smooth spacetimes (in \cref{section:domains of dependence in spactime}). Before describing these notions intuitively, let us motivate them.

\subsubsection*{Point-free physics}

There is a philosophical case for a point-free foundation of the theory of spacetimes~\cite{forrest1996OntologyTopologyTheory,arntzenius2003QuantumMechanicsPointless,schaaf2024TowardsPointFreeSpacetimes}. If spacetime is modeled locally by the Euclidean space $\mathbb{R}^4$, then determining the precise coordinates of some point ${x\in\mathbb{R}^4}$ would require infinite measurement precision, which is unfeasible. As the mathematical foundations of  physical theories ideally only describe entities that are part of the theory's empirical content, a point-free foundation of spacetimes follows naturally. Quantum gravity research endorses this position. Isham~\cite[p59]{isham1990IntroductionGeneralTopology}:
\begin{quote}\small
	``the most important feature of space is not the points which it contains but rather the open subsets and the lattice relations between them. But then, since, physically speaking, a `point' is a most peculiar concept anyway, why not drop it altogether and deal directly with frames/locales?''
\end{quote}
Sorkin, originator of causal set theory~\cite{bombelli1987SpacetimeCausalSet}, writes similarly~\cite[p927]{sorkin1991FinitarySubstituteContinuous}:
\begin{quote}\small
	``From an `operational' perspective, an individual point of~[a topological space]~$S$ is a very ideal limit of what we can directly measure. A much better correlate of a single `position-determination' would probably be an open subset of $S$. Moreover, even for continuum physics, the individual points (or `events') of $S$ exist only as carriers for the topology, and thereby also for higher-level constructs such as the differentiable structure and the metric and `matter' fields: not the points per se, but only this kind of relation involving them has physical meaning.''
\end{quote}

In a quantum mechanical description of a particle in spacetime, the points are similarly immaterial: as individual points in spacetime have Lebesgue measure zero, the probability of finding a particle there will always be zero, irrespective of the wave function. It only makes sense to assign probabilities to \emph{regions}~\cite{forrest1996OntologyTopologyTheory}. On a technical level, this corresponds to the fact that the complete Boolean algebra of measurable subsets of the spacetime, modulo null sets, has no~atoms~\cite[\S 2]{arntzenius2003QuantumMechanicsPointless}. The problem persists in quantum field theory~\cite[\S I.5.2]{haag1996LocalQuantumPhysics}:
\begin{quotation}\small
	``There is another problem which has to be faced. The quantum field at a point cannot be an honest observable. Physically this appears evident because a measurement at a point would necessitate infinite energy.''
\end{quotation}
In fact, algebraic quantum field theory has several no-go theorems restricting spacetime points~\cite[\S 6]{halvorson2007AlgebraicQuantumField}. For instance, under translation covariance, the canonical position and momentum fields commute at any spacetime point.

\subsubsection*{Point-free mathematics}

Treating spacetimes as point-free spaces also opens up interesting connections to other fields of mathematics and computer science not typically associated with differential geometry or causality theory, such as lattice theory, logic, concurrency theory, and, as this article will show, sheaf theory.

\emph{Locales} are a point-free generalisation of topological spaces~\cite{johnstone1982StoneSpacesa}. Whereas a topological space consists of a topology $\Opens S$ on a specified set of points $S$, a locale $X$ is solely defined by a lattice $\Opens X$ of abstract `open regions'. To be precise,~$\Opens X$ is a \emph{frame:} it has finite meets ($\wedge$, intersections) and arbitrary joins ($\bigvee$, unions), making it a complete lattice, and finite meets distribute over joins. 
We denote the inclusion order on these abstract regions by $\sqleq$. Of course, any topological space $S$ becomes a locale by letting the abstract open regions to be the actual open subsets, meets be intersections, joins be unions, and $\sqleq$ be inclusion. But there are locales that do not arise this way, even locales with no points at all. The two notions are related by a well-known adjunction \begin{tikzcd}[cramped]
	\Top & \Loc
	\arrow[""{name=0, anchor=center, inner sep=0}, shift left=2, from=1-1, to=1-2]
	\arrow[""{name=1, anchor=center, inner sep=0}, shift left=2, from=1-2, to=1-1]
	\arrow["\dashv"{anchor=center, rotate=-90}, draw=none, from=0, to=1]
\end{tikzcd}
between the category of topological spaces and continuous functions and the category of locales and their morphisms~\cite[Chapter~IX]{maclane1994SheavesGeometryLogic}. Not just from a physical point of view, but also mathematically, and especially constructively, locales `behave better' than topological spaces~\cite{johnstone1983PointPointlessTopology}. For more technical background on locales, see~\cite{picado2012FramesLocalesTopology}.

\subsubsection*{Point-free causality}

Like topological spaces generalise to locales, spaces equipped with a preorder---thought of as modelling causality---generalise to \emph{ordered locales} (recalled in \cref{section:ordered locales}). 
An \emph{ordered topological space} $(S,\leq)$ is simply a topological space $S$ equipped with a preorder $\leq$ on the underlying set of points, that is, a transitive ($x\leq y\leq z$ implies $x\leq z$) and reflexive ($x\leq x$) relation. Motivating examples are the causality relations $\caus$ in smooth spacetimes $M$, where $x \caus y$ holds exactly when $x$ and $y$ are connected by a smooth future directed causal curve. The chronology relation $\chron$ is defined similarly but using timelike curves instead. For more detail, see~\cite{penrose1972TechniquesDifferentialTopology,minguzzi2019LorentzianCausalityTheory,landsman2021FoundationsGeneralRelativity}. 

\begin{figure}
	\centering
	\definecolor{cffffdc}{RGB}{255,255,220}
\definecolor{cefffff}{RGB}{239,255,255}
\definecolor{cededed}{RGB}{237,237,237}
\definecolor{c002297}{RGB}{0,34,151}
\definecolor{c139700}{RGB}{19,151,0}
\definecolor{cffebdc}{RGB}{255,235,220}
\definecolor{ce0ffdc}{RGB}{224,255,220}
\definecolor{c001764}{RGB}{0,23,100}
\definecolor{c640000}{RGB}{100,0,0}
\definecolor{c644700}{RGB}{100,71,0}
\definecolor{c0d6400}{RGB}{13,100,0}
\definecolor{grey}{HTML}{F6F6F6}
\definecolor{yellow}{HTML}{FFFFEE}
\definecolor{blue}{HTML}{F7FFFF}
\definecolor{CiteColor}{HTML}{8C6B69}

\def \globalscale {.9000000}
\begin{tikzpicture}[y=1pt, x=1pt, yscale=\globalscale,xscale=\globalscale, every node/.append style={scale=\globalscale}, inner sep=0pt, outer sep=0pt]
	\begin{scope}[blend group=multiply]
		\path[fill=yellow,line cap=butt,line join=miter,line 
		width=1.0pt,miter limit=4.0] (188.824, 42.656) -- (223.264, 102.006) -- 
		(123.982, 102.006) -- (156.461, 45.751) -- cycle;

		\path[fill=blue,line cap=butt,line join=miter,line 
		width=1.0pt,miter limit=4.0] (156.461, 91.458) -- (123.982, 35.203) -- 
		(223.264, 35.203) -- (190.953, 91.4) -- cycle;
	\end{scope}
	
	\begin{scope}[blend group=multiply]
		\path[fill=grey,line cap=butt,line join=miter,line 
		width=1.0pt,miter limit=4.0] (21.548, 42.235) -- (56.057, 102.006) -- 
		(-12.961, 102.006) -- cycle;

		\path[fill=grey,line cap=butt,line join=miter,line 
		width=1.0pt,miter limit=4.0] (25.064, 94.974) -- (-9.445, 35.203) -- (59.573, 
		35.203) -- cycle;
	\end{scope}
	
	\begin{scope}[decoration={
			markings,
			mark=between positions .5 and .95 step 30pt with {\arrow[line width = .75pt]{>}}}
		]
		\path[postaction={decorate},draw=c002297,line cap=butt,line join=miter,line width=1.0pt,miter 
		limit=4.0] (21.548, 42.235).. controls (21.548, 42.235) and (19.66, 51.593) ..
		(23.176, 62.141).. controls (26.692, 72.688) and (26.603, 78.293) .. (26.603,
		83.469).. controls (26.603, 88.644) and (25.064, 94.974) .. (25.064, 94.974);
	\end{scope}

	\path[fill=black,line width=1.0pt,dash pattern=on 2.0pt off 1.0pt] (21.548, 
	42.235) circle (3.516pt);

	\path[fill=black,line width=1.0pt,dash pattern=on 2.0pt off 1.0pt] (25.064, 
	94.974) circle (3.516pt);

	\path[draw=black,fill=cefffff,line cap=butt,line join=miter,line 
	width=1.0pt,miter limit=4.0] (155.099, 88.743).. controls (157.407, 94.133) 
	and (165.07, 94.32) .. (170.524, 94.974).. controls (177.25, 95.781) and 
	(187.042, 98.191) .. (190.637, 91.984).. controls (192.266, 89.172) and 
	(190.396, 84.268) .. (187.793, 82.512).. controls (182.036, 78.628) and 
	(174.413, 85.107) .. (167.646, 85.628).. controls (163.819, 85.922) and 
	(159.273, 83.239) .. (156.134, 85.628).. controls (155.3, 86.263) and (154.67,
	87.742) .. (155.099, 88.743) -- cycle;

	\path[draw=black,fill=cffffdc,line cap=butt,line join=miter,line 
	width=1.0pt,miter limit=4.0] (169.726, 55.563).. controls (174.772, 56.359) 
	and (180.859, 60.479) .. (184.939, 57.406).. controls (189.384, 54.06) and 
	(191.578, 44.805) .. (187.589, 40.929).. controls (182.431, 35.918) and 
	(173.555, 44.29) .. (166.396, 44.966).. controls (163.45, 45.244) and 
	(159.784, 43.062) .. (157.519, 44.966).. controls (155.461, 46.697) and 
	(154.608, 50.698) .. (156.133, 52.914).. controls (158.75, 56.716) and 
	(165.166, 54.843) .. (169.726, 55.563) -- cycle;

	\path[draw=black,line cap=butt,line join=miter,line width=0.5pt,miter 
	limit=4.0,dash pattern=on 0.5pt off 2.0pt] (21.548, 42.235) -- (56.057, 
	102.006);

	\path[draw=black,line cap=butt,line join=miter,line width=0.5pt,miter 
	limit=4.0,dash pattern=on 0.5pt off 2.0pt] (21.548, 42.235) -- (-12.961, 
	102.006);

	\path[draw=black,line cap=butt,line join=miter,line width=0.5pt,miter 
	limit=4.0,dash pattern=on 0.5pt off 2.0pt] (-9.445, 35.203) -- (25.064, 
	94.974) -- (59.573, 35.203);

	\path[draw=black,line cap=butt,line join=miter,line width=0.5pt,miter 
	limit=4.0,dash pattern=on 0.5pt off 2.0pt] (156.461, 45.751) -- (123.982, 
	102.006);

	\path[draw=black,line cap=butt,line join=miter,line width=0.5pt,miter 
	limit=4.0,dash pattern=on 0.5pt off 2.0pt] (188.824, 42.656) -- (223.264, 
	102.308);

	\path[draw=black,line cap=butt,line join=miter,line width=0.5pt,miter 
	limit=4.0,dash pattern=on 0.5pt off 2.0pt] (156.461, 91.458) -- (123.982, 
	35.203);

	\path[draw=black,line cap=butt,line join=miter,line width=0.5pt,miter 
	limit=4.0,dash pattern=on 0.5pt off 2.0pt] (190.637, 91.984) -- (223.264, 
	35.472);

	\path[draw=CiteColor,line cap=round,line join=miter,line width=1.001pt,miter 
	limit=4.0] (67.255, 68.162).. controls (67.255, 68.162) and (71.953, 68.523) 
	.. (74.287, 68.307).. controls (76.671, 68.087) and (78.945, 66.538) .. 
	(81.319, 66.847).. controls (83.918, 67.185) and (85.73, 70.363) .. (88.351, 
	70.363).. controls (90.971, 70.363) and (92.762, 66.847) .. (95.383, 66.847)..
	controls (98.003, 66.847) and (99.794, 70.363) .. (102.415, 70.363).. 
	controls (105.035, 70.363) and (106.896, 67.449) .. (109.447, 66.847).. 
	controls (110.587, 66.578) and (112.963, 66.847) .. (112.963, 66.847);

	\path[draw=CiteColor,line cap=round,line join=miter,line width=1.001pt,miter 
	limit=4.0] (67.255, 64.646).. controls (67.255, 64.646) and (71.953, 65.007) 
	.. (74.287, 64.791).. controls (76.671, 64.571) and (78.945, 63.022) .. 
	(81.319, 63.331).. controls (83.918, 63.669) and (85.73, 66.847) .. (88.351, 
	66.847).. controls (90.971, 66.847) and (92.762, 63.331) .. (95.383, 63.331)..
	controls (98.003, 63.331) and (99.794, 66.847) .. (102.415, 66.847).. 
	controls (105.035, 66.847) and (106.896, 63.933) .. (109.447, 63.331).. 
	controls (110.587, 63.062) and (112.963, 63.331) .. (112.963, 63.331);

	\path[draw=CiteColor,line cap=round,line join=round,line width=1.001pt,miter 
	limit=4.0] (111.841, 69.6).. controls (111.841, 69.6) and (111.841, 65.942) ..
	(117.075, 64.996).. controls (111.841, 64.113) and (111.841, 60.455) .. 
	(111.841, 60.455);

		%
		%
		%
		%
		%
		%
		%
		%
		%
		%
		%
		%
		%
		%
		%
		%
		%
		%
		%
		%
		%
		%
		%
		%
		%
		%
		%
	
\end{tikzpicture}
	\caption{Transition from order on points to order on regions.}
	\label{figure:intuition}
\end{figure}
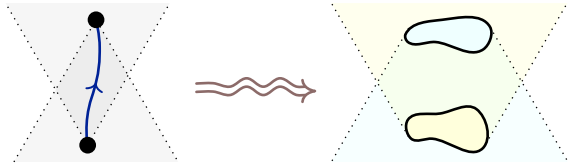

Locales have no primitive points, so a preorder needs to be based elsewhere. An \emph{ordered locale} $(X,\Leq)$ is a locale $X$ equipped with a preorder $\Leq$ on its frame of opens~$\Opens X$ that satisfies one additional axiom (\cref{definition:ordered locales}). The intuition is portrayed in \cref{figure:intuition}. The adjunction between topological spaces and locales extends to categories of ordered topological spaces and ordered locales~\cite{heunenSchaaf2024OrderedLocales}, justifying regarding ordered locales as a point-free generalisation of ordered topological spaces. Ordered locales in turn form a mathematical setting to study point-free generalisations of spacetimes by Malament's theorem, stating roughly that a smooth spacetime $M$ is determined by the ordered space~$(M,\caus)$~\cite{malament1977ClassContinuousTimelike}. 
For more details, see~\cite{schaaf2024TowardsPointFreeSpacetimes}.

\subsubsection*{Causal cover}

This paper studies how a natural notion of causal cover arises in an ordered locale. 
A similar idea was first proposed in \cite[Definition~2.12]{christensen2005CausalSitesQuantum} in the setting of so-called \emph{causal sites}, but has not been developed further. Our considerations here are deeply inspired by their definition.

The intuition is this. Think of an ordered space $(S,\leq)$ as an arena in which information can `flow', subject to the restraints of the order~$\leq$. For example, in spacetime physical information can only propagate along causal curves. Here, we might imagine information flowing along continuous monotone curves $\gamma\colon [0,1]\to S$. Note that we are completely agnostic about what information the space `contains', and merely study the possible ways of propagation.

To what extent is the information of a given region $U$ determined by some other region $A$ in its past? We say $A$ \emph{covers $U$ (from below)} if every continuous monotone path landing in $U$ has an extension into the past that intersects $A$.
Intuitively: all information that flows into $U$ has to come from~$A$. \Cref{section:causal coverages ordered locales} formulates a localic version by replacing curves by sequences in the localic causal order~$\Leq$, regarded as approximating ideal, infinitesimally thin curves.

\subsubsection*{Causal coverage}

A \emph{coverage} is a relation $\cov$ capturing when a region is `covered' by a collection of regions. The usual notion of open cover gives a canonical cover on the open sets of a topological space or locale: $U\cov (U_i)_{i\in I}$ exactly when $U\subseteq \bigcup_{i\in I}U_i$.
Coverages are fundamental in lattice theory and topology, for example on meet-semilattices \cite{johnstone1982StoneSpacesa,ball2014ExtendingSemilatticesFrames}, in formal topology \cite{sambin1989IntuitionisticFormalSpaces,sambin2003PointsFormalTopology}, or as Grothendieck topologies in category theory and sheaf theory \cite{maclane1994SheavesGeometryLogic}.
A causal notion of coverage can therefore transfer insights and tools from these fields of mathematics to study causality theory, see \Cref{section:conclusions}.

Coverages satisfy three fundamental axioms (see \cref{section:causal coverages as grothendieck topologies} for more detail):
any region covers itself;
covering is preserved under intersection;
a family is covering if all its members are covered.
One of the main results of this paper is that the definition of causal coverage in ordered locales satisfies analogous properties (see \cref{lemma:properties of coverage from locale}). In particular, our definition thus improves on causal coverages in causal sites, which need not satisfy the second fundamental axiom~\cite[Remark~2.13]{christensen2005CausalSitesQuantum}. \Cref{theorem:causal grothendieck topology} shows that the notion of causal coverage canonically defines a generalised Grothendieck topology, connecting to sheaf theory. 

\subsubsection*{Domain of dependence}

Moreover, causal coverages induces a notion of \emph{domain of dependence} (see \cref{section:domains of dependence}). Domains of dependence were introduced to characterise globally hyperbolic spacetimes~\cite{geroch1970DomainofDependence}, and form a fundamental tool in modern mathematical relativity theory. In the point-free setting, we define the domain of dependence of a region $U$ as the largest region causally covered by~$U$. We compare this localic definition with the traditional curve-wise definition (in \cref{section:domains of dependence in spactime}). While the localic domain of dependence always contains the curve-wise version, they are generally distinct. We end by speculating on possible implications to the problem of hole-freeness~\cite{krasnikov2009EvenMinkowskiSpace,manchak2009SpacetimeHolefree} in \Cref{section:conclusions}.


\section{Ordered locales}
\label{section:ordered locales}
Ordered locales were introduced in \cite{heunenSchaaf2024OrderedLocales} to serve as a point-free analogue of preordered topological spaces. As mentioned, by \emph{ordered space} we simply mean a pair $(S,\leq)$, where $S$ is a topological space, and $\leq$ is a preorder defined on the underlying set of~$S$. In an ordered space we introduce the familiar \emph{upset} and \emph{downset} of any subset~$A\subseteq S$:
	\[
	\up A:=\{x\in P:\exists a\in A: a\leq x\},
	\quad\text{and}\quad
	\down A:=\{y\in P:\exists a\in A: y\leq a \}.
	\]
In the current context of relativistic causality, we shall also call these the \emph{future (cone)} and \emph{past (cone)} of $A$, respectively. It is elementary to show that the future and past cones define so-called join-preserving monads on the powerset $\Powerset(S)$ of $S$.

\begin{lemma}\label{lemma:properties of cones}
	If $(S,\leq)$ is a preordered set, then for any subsets $A,B\subseteq S$:
	\begin{enumerate}[label={(\alph*)},font=\normalfont]
		\item $\up A\subseteq \up B$ and $\down A\subseteq \down B$ whenever $A\subseteq B$;
		\item $A\subseteq \up A$ and $A\subseteq \down A$;
		\item $\up\up A\subseteq \up A$ and $\down\down A\subseteq \down A$.
	\end{enumerate}
	Further, for any family $(A_i)_{i\in I}$ of subsets of $S$:
	\begin{enumerate}[label={(\alph*)},font=\normalfont]
		\setcounter{enumi}{3}
		\item $\bigcup_{i\in I}\up A_i = \up\left(\bigcup_{i\in I}A_i\right)$ and $\bigcup_{i\in I}\down A_i = \down\left(\bigcup_{i\in I}A_i\right)$.
	\end{enumerate}
\end{lemma}
An operation with properties (a)--(c) is called a \emph{monad} or \emph{closure operator}. Note that~(b) and~(c) together imply the equalities $\up\up A = \up A$ and $\down\down A= \down A$.

Generally the topological and order structure of an ordered space are not required to be compatible in any way. However, for the purposes of our discussion we need the following condition.
\begin{definition}\label{definition:open cones}
	We say an ordered space $(S,\leq)$ has \emph{open cones} if $\up U$ and $\down U$ are open whenever $U\subseteq S$ is open.
\end{definition}
The main result \cite[Theorem~6.3]{heunenSchaaf2024OrderedLocales} shows there is an adjunction between certain categories of ordered spaces with open cones and ordered locales. Thus, ordered locales can be viewed as point-free analogues of ordered spaces with open cones. Here is the main definition.

\begin{definition}
	\label{definition:ordered locales}
	An \emph{ordered locale} $(X,\Leq)$ is a locale $X$ equipped with a preorder~$\Leq$ on its frame of opens $\Opens X$, satisfying the following axiom:
		\[\tag{$\vee$}\label{axiom:V}
		\forall i\in I: U_i\Leq V_i 
		\qquad \text{implies} \qquad 
		\bigvee_{i\in I}U_i\Leq \bigvee_{i\in I}V_i.
		\]
\end{definition}

The intuition is that $\Leq$ describes a causal order on arbitrary open regions of the space, and we want to develop causality theory in this setting without having to resort to causal relations between points. The first elementary components of this theory are the localic analogues of the future and past cones in an ordered space. 

\begin{definition}\label{definition:localic cones}
	The \emph{localic cones} of an ordered locale $(X,\Leq)$ are defined as
	\[
	\Up U:=\bigvee \left\{ V\in\Opens X: U\Leq V\right\}
	\qquad\text{and}\qquad
	\Down U:=\bigvee \left\{ W\in\Opens X: W\Leq U\right\}.
	\]
\end{definition}

They enjoy properties similar to those of \cref{lemma:properties of cones}.

\begin{lemma}\label{lemma:properties of localic cones}
	For a locale $X$ and a preorder $\Leq$ on its frame of opens:
	\begin{enumerate}[label={(\alph*)},font=\normalfont]
		\item if $U\Leq V$ then $U\sqleq \Down V$ and $V \sqleq \Up U$;
		\item $U\sqleq \Up U$ and $U \sqleq \Down U$.
	\end{enumerate}
	If axiom~\eqref{axiom:V} is satisfied, then furthermore:
	\begin{enumerate}[label={(\alph*)},font=\normalfont]\setcounter{enumi}{2}
		\item $U \Leq \Up U$ and $\Down U \Leq U$; 
		\item $\Up \Up U = \Up U$ and $\Down \Down U = \Down U$;
		\item if $U\sqleq V$ then $\Up U\sqleq \Up V$ and $\Down U\sqleq \Down V$. 
	\end{enumerate}
\end{lemma}
\begin{proof}
	See \cite[Lemma~4.4]{schaaf2024TowardsPointFreeSpacetimes}.
\end{proof}

In particular, $\Up$ and $\Down$ are monads on $\Opens X$.

\begin{example}\label{example:egli-milner order}
	If $(S,\leq)$ is an ordered space, then the locale $\loc(S)$, defined by $\Opens\loc(S):= \Opens S$, becomes an ordered locale when equipped with the \emph{Egli-Milner order}:
		\[
			U\Leq V
			\qquad
			\text{if and only if}\qquad
			U\subseteq \down V\text{~and~}V\subseteq \up U.
		\]
	From now, we always understand $\loc(S)$ to be an ordered locale in this sense.
\end{example}

\begin{lemma}\label{lemma:localic cones in a space}
	The localic cones in ordered space are calculated as:
	\[
	\Up U = (\up U)^\circ
	\qquad\text{and}\qquad
	\Down U = (\down U)^\circ.
	\]
\end{lemma}
\begin{proof}
	See \cite[Lemma~4.13]{schaaf2024TowardsPointFreeSpacetimes}.
\end{proof}

\begin{corollary}\label{corollary:localic cones in space with open cones}
	In an ordered space with open cones: $\Up U =\up U$ and $\Down U = \down U$.
\end{corollary}

An important motivating class of examples is that of smooth spacetimes. A \emph{smooth spacetime} is a smooth connected Lorentzian manifold $M$ with time orientation. For more details we refer to \cite{penrose1972TechniquesDifferentialTopology,minguzzi2019LorentzianCausalityTheory,landsman2021FoundationsGeneralRelativity}. The Lorentzian structure induces the \emph{chronology relation} $\chron$ and the \emph{causality relation} $\caus$ on the set of points~\cite[\S 1.11]{minguzzi2019LorentzianCausalityTheory}. The chronology describes information flow going strictly slower than the speed of light, while the causality relation describes information flow going at most at the speed of light. Generally, $\chron$ is transitive, and $\caus$ is a preorder. Hence~$(M,\caus)$ is an ordered space. The future and past cones induced by~$\chron$ are denoted~$I^\pm$, respectively, and those induced by $\caus$ are denoted by $J^\pm$. It turns out that the topology, chronology, and causality of a spacetime are intertwined. A collection of results from the literature relevant to the discussion on ordered spaces is presented in \cite[Chapter~3]{schaaf2024TowardsPointFreeSpacetimes}, the conclusion of which is the following.

\begin{theorem}\label{theorem:spacetimes have OC}
	$(M,\caus)$ has open cones for any smooth spacetime~$M$.
\end{theorem}

\begin{corollary}\label{example:localic cones in spacetime}
	For any open region $U$ in a smooth spacetime $M$ we have
	\[
	\Up U = I^+(U) = J^+(U)
	\qquad\text{and}\qquad
	\Down U = I^-(U)=J^-(U).
	\]
\end{corollary}

\begin{remark}
	It can be seen that the open cone condition for spacetimes is roughly equivalent to the so-called \emph{push-up principle} in causality theory~\cite[\S 3.1.2]{schaaf2024TowardsPointFreeSpacetimes}. There exist lower regularity, non-smooth spacetimes that do not satisfy the push-up principle~\cite{grant2020FutureNotAlways}, and therefore will not satisfy the open cone condition. We will not consider such spacetimes here.
\end{remark}

\section{Parallel ordered locales}
\label{section:parallel ordered locales}
Just as the chronological and causal cones $I^\pm$ and $J^\pm$ are fundamental tools in mathematical relativity theory, the localic cones provide one of the main tools to study notions of causality inside an ordered locale. In fact, if $X$ is a locale and $u,d\colon \Opens X\to \Opens X$ are monads on its frame of opens, then we obtain the structure of an ordered locale $\Leq$ via
\[
U\Leq V
\qquad\text{if and only if}\qquad
U\sqleq d(V)\text{~and~}V\sqleq u(U).
\]
Here we think of $u$ and $d$ as abstract future and past localic cone operations. It can be shown that every ordered locale essentially arises in this way \cite[\S 4.2.2]{schaaf2024TowardsPointFreeSpacetimes}. It is therefore not much of a restriction to assume the following axiom:
\[\tag{\textsc{c}-$\Leq$}\label{axiom:cones give order}
U\Leq V \qquad \text{if and only if} \qquad U\sqleq \Down V \text{ and } V\sqleq \Up U.
\]

However, the cones $u$ and $d$ share \emph{a priori} no relation with one another. This seems to miss some of the intuition of the future and past cones in ordered spaces and spacetimes, which are related at least in the following sense:
	\[
		x\leq y
		\quad\text{if and only if}\quad
		x\in \down\{y\}
		\quad\text{if and only if}\quad
		y\in\up\{x\}.
	\]
So, to be able to develop a well-behaved theory of causality in ordered locales, we introduce some additional axioms. First: we ask the localic cones $\Up$ and~$\Down$ to be \emph{join-preserving}, which is elementary in the case of $\up$ and~$\down$ (\cref{lemma:properties of cones}). Second: we introduce a compatibility condition between the localic cones that we call \emph{parallel orderedness}, introduced in~\cite[\S 4.5]{schaaf2024TowardsPointFreeSpacetimes}. Intuitively, this condition ensures that the localic cones run in parallel, or, that the speed of propagation into the future is the same as the speed of propagation into the past (see \cref{example:minkowski different speeds of light}).

\begin{definition}\label{definition:localic cones preserve joins}
	We say the localic cones of an ordered locale $(X,\Leq)$ \emph{preserve joins} if the following law holds:
	\[\tag{\textsc{c}-$\vee$}\label{axiom:LV}
	\Up\bigvee_{i\in I} U_i = \bigvee_{i\in I}\Up U_i
	\qquad\text{and}\qquad
	\Down \bigvee_{i\in I} U_i = \bigvee_{i\in I}\Down U_i.
	\]
\end{definition}

\begin{proposition}\label{proposition:space with OC has join-preserving cones}
	If $(S,\leq)$ has open cones, then $\loc(S)$ satisfies~\eqref{axiom:LV}.
\end{proposition}
\begin{proof}
	This follows simply from the fact that the localic cones equal the point-wise cones (\cref{corollary:localic cones in space with open cones}), and the latter preserve joins by \cref{lemma:properties of cones}(d).
\end{proof}

\begin{example}\label{example:upper order}
	However, even for ordered locales based on spaces $S$, axiom~\eqref{axiom:LV} can fail for quite trivial reasons. Consider for instance the \emph{upper order}, which is defined by $U\LeqU V$ if and only if $V\subseteq\up U$. This defines an ordered locale whose past localic cone is constant: $\Down U = S$, and so fails to preserve empty joins whenever $S$ is non-empty. But~\eqref{axiom:LV} can fail in less trivial ways, even for finite joins, as the next example shows.
\end{example}

\begin{example}\label{example:LV can fail in spaces}
	Consider the ordered space whose underlying space is the disjoint union ${(\{0\}\times [0,\infty))\cup (\{1\}\times \mathbb{R})}$, and whose order is generated by setting ${(0,x)\leq (1,y)}$ if and only if $x=y$. This space does not have open cones, since the future of $\{0\}\times [0,\infty)$ is not open in $\{1\}\times\mathbb{R}$. We construct two open subsets of $S$ such that $\Up(U\cup V)\neq \Up U \cup \Up V$. Namely, take:
	\[
	U:= \{0\}\times [0,\infty)
	\qquad\text{and}\qquad
	V:= \{1\}\times (-\infty,0).
	\]
	Then $\up V = V$, and $\up U = U \cup \left(\{1\}\times [0,\infty)\right)$. We therefore find that
	\[
	\Up(U\cup V) = \left(\up U \cup \up V\right)^\circ 
	=
	S,
	\]
	while
	\[
	\Up U \cup \Up V = \left(\up U\right)^\circ \cup V
	=
	\left(\{0\}\times [0,\infty)\right)\cup\left(\{1\}\times \left(\mathbb{R}\setminus \{0\}\right)\right).
	\]
	In other words, $\Up U \cup \Up V$ misses precisely the point $(1,0)$. See \cref{figure:LV can fail in spaces}. We can see that~\eqref{axiom:LV} fails essentially due to the interior operator not preserving joins.
\end{example}
\begin{figure}[t]\centering
	\definecolor{cffebdc}{RGB}{255,235,220}
\definecolor{cefffff}{RGB}{239,255,255}
\definecolor{ce0ffdc}{RGB}{224,255,220}
\definecolor{c640000}{RGB}{100,0,0}
\definecolor{c001764}{RGB}{0,23,100}
\definecolor{c0d6400}{RGB}{13,100,0}
\definecolor{cffffdc}{RGB}{255,255,220}
\definecolor{c644700}{RGB}{100,71,0}

\def \globalscale {1.000000}
\begin{tikzpicture}[y=1pt, x=1pt, yscale=\globalscale,xscale=\globalscale, every node/.append style={scale=\globalscale}, inner sep=0pt, outer sep=0pt]
  \path[fill=cffebdc,line cap=butt,line join=miter,line width=1.0pt,miter 
  limit=4.0] (59.585, 87.892) -- (128.247, 87.883) -- (128.247, 76.553) -- 
  (59.585, 76.561) -- (57.894, 82.217) -- cycle;

  \path[fill=cefffff,line cap=butt,line join=miter,line width=1.0pt,miter 
  limit=4.0] (55.176, 87.887) -- (-13.486, 87.883) -- (-13.486, 76.553) -- 
  (55.176, 76.557) -- (57.38, 82.222) -- cycle;

  \path[fill=ce0ffdc,line width=1.0pt] (57.38, 59.545) -- (128.247, 59.537) -- 
  (128.247, 48.206) -- (57.38, 48.215) -- cycle;

  \path[draw=black,line cap=butt,line join=miter,line width=0.5pt,miter 
  limit=4.0,dash pattern=on 0.5pt off 2.0pt] (57.38, 82.222) -- (57.38, 53.876);

  \path[draw=c640000,line cap=butt,line join=miter,line width=2.0pt,miter 
  limit=4.0] (57.894, 82.217) -- (114.073, 82.222);

  \path[draw=c001764,line cap=butt,line join=miter,line width=2.0pt,miter 
  limit=4.0,dash pattern=on 4.0pt off 2.0pt] (0.892, 82.217) -- (-13.486, 
  82.217);

  \path[draw=c640000,line cap=butt,line join=miter,line width=2.0pt,miter 
  limit=4.0,dash pattern=on 4.0pt off 2.0pt] (114.073, 82.222) -- (128.247, 
  82.222);

  \path[draw=c0d6400,line cap=butt,line join=miter,line width=2.0pt,miter 
  limit=4.0,dash pattern=on 4.0pt off 2.0pt] (114.073, 53.876) -- (128.247, 
  53.876);

  \path[draw=c0d6400,line cap=butt,line join=miter,line width=2.0pt,miter 
  limit=4.0] (57.38, 53.876) -- (114.073, 53.876);

  \path[draw=c0d6400,line cap=butt,line join=miter,line width=2.0pt,miter 
  limit=4.0] (57.38, 59.545) -- (57.38, 48.206);

  \path[draw=c001764,line cap=butt,line join=miter,line width=2.0pt,miter 
  limit=4.0] (56.915, 82.217) -- (0.892, 82.222);

  \node[text=c0d6400,line cap=butt,line join=miter,line width=1.0pt,miter 
  limit=4.0,anchor=south west] (text6) at (81.61, 60.744){$U$};

  \node[text=black,line cap=butt,line join=miter,line width=1.0pt,miter 
  limit=4.0,anchor=south west] (text7) at (131.613, 48.916){$[0,\infty)$};

  \node[text=black,line cap=butt,line join=miter,line width=1.0pt,miter 
  limit=4.0,anchor=south west] (text7-5) at (135.835, 78.067){$\mathbb{R}$};

  \path[draw=c640000,line width=2.0pt] (60.163, 87.887).. controls (58.721, 
  86.349) and (57.91, 84.325) .. (57.894, 82.217).. controls (57.911, 80.113) 
  and (58.717, 78.093) .. (60.154, 76.557);

  \path[draw=c001764,line cap=butt,line join=miter,line width=2.0pt,miter 
  limit=4.0] (54.646, 87.887).. controls (56.088, 86.349) and (56.898, 84.325) 
  .. (56.915, 82.217).. controls (56.898, 80.113) and (56.091, 78.093) .. 
  (54.655, 76.557);

  \node[text=c640000,line cap=butt,line join=miter,line width=1.0pt,miter 
  limit=4.0,anchor=south west] (text11) at (74.641, 89.149){$(\up U)^\circ$};

  \node[text=c001764,line cap=butt,line join=miter,line width=1.0pt,miter 
  limit=4.0,anchor=south west] (text12) at (23.435, 89.392){$V$};

%
%
%
%
%
%
%
%
%
%
%
%
%
%
%
%
%
%
%
%
%
%
%
%
%
%
%

\end{tikzpicture}
	\caption{Illustration of the space in \cref{example:LV can fail in spaces}.}
	\label{figure:LV can fail in spaces}
\end{figure}
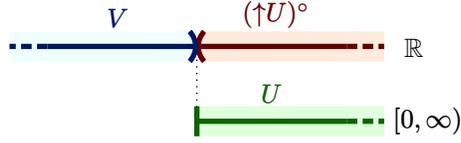

Next, we describe the class of parallel ordered locales.

\begin{definition}\label{definition:parallel ordered locales}
	If $(X,\Leq)$ is an ordered locale, we say $\Leq$ \emph{respects meets} if the following two laws hold:
	\begin{center}
		\dummylabel{axiom:wedge}{$\wedge^\pm$}%
		\begin{minipage}{.3\linewidth}
			\[\tag{$\wedge^+$}\label{axiom:wedge+}
			\begin{tikzcd}[every label/.append style = {font = \normalsize},column sep=0.25cm, row sep=0.2cm]
				{U} & \highlighttikzcd{\exists U'} \\
				{V} & {V'}
				\arrow["\sqleq"{anchor=center, rotate=-90}, draw=none, from=2-1, to=1-1]
				\arrow["\sqleq"{anchor=center, rotate=-90}, draw=none, from=2-2, to=1-2]
				\arrow["{}"{description}, "\Leq"{anchor=center}, draw=none, from=1-1, to=1-2]
				\arrow["{}"{description}, "\Leq"{anchor=center}, draw=none, from=2-1, to=2-2]
			\end{tikzcd}
			\]
		\end{minipage}%
		\hfil
		\begin{minipage}{.3\linewidth}
			\[\tag{$\wedge^-$}\label{axiom:wedge-}
			\begin{tikzcd}[every label/.append style = {font = \normalsize},column sep=0.25cm, row sep=0.2cm]
				\highlighttikzcd{\exists U} & {U'} \\
				{V} & {V'.}
				\arrow["\sqleq"{anchor=center, rotate=-90}, draw=none, from=2-1, to=1-1]
				\arrow["\sqleq"{anchor=center, rotate=-90}, draw=none, from=2-2, to=1-2]
				\arrow["{}"{description}, "\Leq"{anchor=center}, draw=none, from=1-1, to=1-2]
				\arrow["{}"{description}, "\Leq"{anchor=center}, draw=none, from=2-1, to=2-2]
			\end{tikzcd}
			\]
		\end{minipage}
	\end{center}
	We say $\Leq$ \emph{preserves bottom} if the bottom element of $\Opens X$ is only related to itself:
	\[\tag{$\varnothing$}\label{axiom:empty}
	U\Leq \varnothing \Leq V \quad\text{implies}\quad U=\varnothing=V.
	\]
	Lastly, we say $(X,\Leq)$ is \emph{parallel ordered} if all three axioms~\eqref{axiom:wedge} and~\eqref{axiom:empty} hold.
\end{definition}

\begin{remark}
	What is meant by `respects meets' is explained further in \cref{lemma:wedge iff frobenius,lemma:wedge iff strong wedge}. The intended reading of the diagram in axiom~\eqref{axiom:wedge+} is that: if $U\sqleq V$ and $V\Leq V'$, \emph{then} there exists $U'$ such that $U\Leq U'$ and $U'\sqleq V'$. Axiom~\eqref{axiom:wedge-} is parsed analogously, as are similar subsequent diagrams. The visual intuition behind these axioms is portrayed in \cref{figure:axiom wedge+}(a).
\end{remark}

\begin{figure}[b]\centering
	\begin{subfigure}[b]{0.4\textwidth}
		\definecolor{cededed}{RGB}{237,237,237}
\definecolor{cffffdc}{RGB}{255,255,220}
\definecolor{cffebdc}{RGB}{255,235,220}
\definecolor{c640000}{RGB}{100,0,0}
\definecolor{ccd0000}{RGB}{205,0,0}
\definecolor{cefffff}{RGB}{239,255,255}
\definecolor{ce0ffdc}{RGB}{224,255,220}
\definecolor{c001764}{RGB}{0,23,100}
\definecolor{c644700}{RGB}{100,71,0}
\definecolor{c0d6400}{RGB}{13,100,0}
\definecolor{c644700}{RGB}{100,71,0}

\def \globalscale {1.000000}
\begin{tikzpicture}[y=.75pt, x=.8pt, yscale=\globalscale,xscale=\globalscale, every node/.append style={scale=\globalscale}, inner sep=0pt, outer sep=0pt]
	\begin{scope}[blend group=multiply]
		
		\path[draw=black,fill=cededed,line cap=butt,line join=miter,line 
		width=1.0pt,miter limit=4.0] (29.034, 31.199).. controls (31.157, 20.98) and 
		(49.78, 28.563) .. (60.215, 28.364).. controls (78.189, 28.022) and (113.109, 
		18.462) .. (114.073, 31.199).. controls (115.254, 46.803) and (81.28, 51.255) 
		.. (63.285, 49.929).. controls (50.307, 48.973) and (26.386, 43.939) .. 
		(29.034, 31.199) -- cycle;

		\path[draw=black,fill=cededed,line cap=butt,line join=miter,line 
		width=1.0pt,miter limit=4.0] (23.365, 90.726).. controls (24.771, 75.249) and 
		(53.388, 81.021) .. (68.928, 80.846).. controls (83.753, 80.679) and (111.973,
		74.479) .. (112.598, 89.293).. controls (113.405, 108.403) and (76.473, 
		106.04) .. (57.38, 104.899).. controls (45.119, 104.167) and (22.253, 102.959)
		.. (23.365, 90.726) -- cycle;

		\path[fill=cffffdc,line cap=butt,line join=miter,line width=1.0pt,miter 
		limit=4.0] (88.562, 36.868) -- (114.355, 107.734) -- (37.257, 107.734) -- 
		(63.05, 36.868) -- cycle;

		\path[fill=cffebdc,line cap=butt,line join=miter,line width=1.0pt,miter 
		limit=4.0] (64.539, 94.611) -- (89.552, 94.286) -- (115.609, 22.695) -- 
		(38.364, 22.695) -- cycle;
		
	\end{scope}
	
	\node[text=black,line cap=butt,line join=miter,line width=1.0pt,miter 
	limit=4.0,anchor=south west] (text15) at (33.575, 29.898){$V$};

	\node[text=black,line cap=butt,line join=miter,line width=1.0pt,miter 
	limit=4.0,anchor=south west] (text15-2) at (29, 86.65){$V'$};

	\path[draw=black,fill=cffffdc,line cap=butt,line join=miter,line 
	width=1.0pt,miter limit=4.0] (65.875, 44.442).. controls (63.12, 41.755) and 
	(62.277, 35.997) .. (64.643, 32.963).. controls (67.247, 29.623) and (73.123, 
	30.142) .. (77.223, 31.199).. controls (81.798, 32.377) and (88.986, 34.997) 
	.. (88.562, 39.703).. controls (88.182, 43.911) and (81.372, 44.57) .. 
	(77.223, 45.372).. controls (73.497, 46.092) and (68.592, 47.092) .. (65.875, 
	44.442) -- cycle;

	\path[draw=c640000,fill=cffebdc,line cap=butt,line join=miter,line 
	width=1.001pt,miter limit=4.0,dash pattern=on 8.005pt off 1.001pt] (66.865, 
	99.025).. controls (64.11, 96.338) and (63.267, 90.581) .. (65.633, 87.546).. 
	controls (68.237, 84.207) and (74.113, 84.726) .. (78.213, 85.782).. controls 
	(82.788, 86.96) and (89.976, 89.581) .. (89.552, 94.286).. controls (89.172, 
	98.495) and (82.362, 99.153) .. (78.213, 99.955).. controls (74.487, 100.676) 
	and (69.582, 101.675) .. (66.865, 99.025) -- cycle;

	\path[draw=black,line cap=butt,line join=miter,line width=0.5pt,miter 
	limit=4.0,dash pattern=on 0.5pt off 2.0pt] (63.05, 36.868) -- (37.257, 
	107.734);

	\path[draw=black,line cap=butt,line join=miter,line width=0.5pt,miter 
	limit=4.0,dash pattern=on 0.5pt off 2.0pt] (88.562, 36.868) -- (114.355, 
	107.734);

	\path[draw=black,line cap=butt,line join=miter,line width=0.5pt,miter 
	limit=4.0,dash pattern=on 0.5pt off 2.0pt] (64.539, 94.611) -- (38.364, 
	22.695);

	\path[draw=black,line cap=butt,line join=miter,line width=0.5pt,miter 
	limit=4.0,dash pattern=on 0.5pt off 2.0pt] (89.552, 94.286) -- (115.609, 
	22.695);

	\node[text=c644700,line cap=butt,line join=miter,line width=1.0pt,miter 
	limit=4.0,anchor=south west] (text19) at (69.832, 33){$U$};

	\node[text=ccd0000,line cap=butt,line join=miter,line width=1.0pt,miter 
	limit=4.0,anchor=south west] (text20) at (66.84, 87){$\exists U'$};

	\path[draw=black,line cap=butt,line join=miter,line width=0.5pt,miter 
	limit=4.0,dash pattern=on 0.5pt off 2.0pt] (24.391, 95.564) -- (-2.132, 
	22.695);

	\path[draw=black,line cap=butt,line join=miter,line width=0.5pt,miter 
	limit=4.0,dash pattern=on 0.5pt off 2.0pt] (111.573, 94.841) -- (137.642, 
	22.695);

	\path[draw=black,line cap=butt,line join=miter,line width=0.5pt,miter 
	limit=4.0,dash pattern=on 0.5pt off 2.0pt] (29.034, 31.199) -- (1.177, 
	107.734);

	\path[draw=black,line cap=butt,line join=miter,line width=0.5pt,miter 
	limit=4.0,dash pattern=on 0.5pt off 2.0pt] (114.073, 31.199) -- (141.93, 
	107.734);

		%
		%
		%
		%
		%
		%
		%
		%
		%
		%
		%
		%
		%
		%
		%
		%
		%
		%
		%
		%
		%
		%
		%
		%
		%
		%
		%
		%
		%
		%
	
\end{tikzpicture}
		\caption{Axiom~\eqref{axiom:wedge+}.}
	\end{subfigure}\hfil
	\begin{subfigure}[b]{0.4\textwidth}
		\definecolor{cededed}{RGB}{237,237,237}
\definecolor{cffffdc}{RGB}{255,255,220}
\definecolor{c644700}{RGB}{100,71,0}
\definecolor{c640000}{RGB}{100,0,0}
\definecolor{cffebdc}{RGB}{255,235,220}
\definecolor{ccd0000}{RGB}{205,0,0}
\definecolor{cefffff}{RGB}{239,255,255}
\definecolor{ce0ffdc}{RGB}{224,255,220}
\definecolor{c001764}{RGB}{0,23,100}
\definecolor{c0d6400}{RGB}{13,100,0}

\def \globalscale {1.000000}
\begin{tikzpicture}[y=.8pt, x=.8pt, yscale=\globalscale,xscale=\globalscale, every node/.append style={scale=\globalscale}, inner sep=0pt, outer sep=0pt]
	\begin{scope}[blend group = multiply]
		\path[draw=black,fill=cededed,line cap=butt,line join=miter,line 
		width=1.0pt,miter limit=4.0] (111.974, 40.152).. controls (111.411, 22.201) 
		and (76.78, 30.242) .. (58.846, 31.199).. controls (47.814, 31.787) and 
		(26.511, 29.112) .. (26.935, 40.152).. controls (27.462, 53.878) and (53.836, 
		45.742) .. (67.35, 48.206).. controls (82.892, 51.041) and (112.858, 68.311) 
		.. (111.974, 40.152) -- cycle;

		\path[fill=cffffdc,line cap=butt,line join=miter,line width=1.0pt,miter 
		limit=4.0] (71.819, 37.478) -- (104.445, 39.703) -- (126.111, 99.23) -- 
		(49.343, 99.23) -- cycle;
	\end{scope}

	\path[draw=black,fill=cededed,line width=1.0pt] (58.581, 73.848) -- (52.27, 
	91.186).. controls (42.943, 92.482) and (33.855, 93.037) .. (27.722, 87.07).. 
	controls (24.909, 84.333) and (24.218, 78.314) .. (26.771, 75.334).. controls 
	(31.579, 69.724) and (40.826, 72.718) .. (48.876, 73.718).. controls (51.734, 
	74.073) and (55.049, 74.056) .. (58.581, 73.848) -- cycle;

	\node[text=black,line cap=butt,line join=miter,line width=1.0pt,miter 
	limit=4.0,anchor=south west] (text15-2) at (29.668, 76.66){$V$};

	\path[draw=black,fill=cffffdc,line cap=butt,line join=miter,line 
	width=1.0pt,miter limit=4.0] (80.178, 45.587).. controls (76.855, 43.581) and 
	(71.554, 43.888) .. (71.819, 37.478).. controls (72.057, 31.73) and (82.708, 
	32.466) .. (88.447, 32.857).. controls (94.248, 33.251) and (104.756, 33.954) 
	.. (104.467, 39.762).. controls (104.21, 44.915) and (99.093, 44.157) .. 
	(96.023, 45.719).. controls (93.646, 46.929) and (91.114, 48.274) .. (88.447, 
	48.289).. controls (85.547, 48.305) and (82.661, 47.085) .. (80.178, 45.587) 
	-- cycle;

	\node[text=black,line cap=butt,line join=miter,line width=1.0pt,miter 
	limit=4.0,anchor=south west] (text19) at (32.899, 36){$U$};

	\node[text=c644700,line cap=butt,line join=miter,line width=1.0pt,miter 
	limit=4.0,anchor=south west] (text15) at (82.011, 35){$W$};

	\path[draw=black,line cap=butt,line join=miter,line width=0.5pt,miter 
	limit=4.0,dash pattern=on 0.5pt off 2.0pt] (24.632, 80.927) -- (4.469, 25.529);

	\path[draw=black,line cap=butt,line join=miter,line width=0.5pt,miter 
	limit=4.0,dash pattern=on 0.5pt off 2.0pt] (102.735, 87.892) -- (125.412, 
	25.587);

	\path[draw=black,line cap=butt,line join=miter,line width=0.5pt,miter 
	limit=4.0,dash pattern=on 0.5pt off 2.0pt] (27.193, 37.89) -- (4.867, 99.23);

	\path[draw=black,line cap=butt,line join=miter,line width=0.5pt,miter 
	limit=4.0,dash pattern=on 0.5pt off 2.0pt] (110.768, 35.176) -- (134.041, 
	99.316);

	\path[draw=black,line cap=butt,line join=miter,line width=0.5pt,miter 
	limit=4.0,dash pattern=on 0.5pt off 2.0pt] (104.445, 39.703) -- (126.111, 
	99.23);

	\path[draw=black,line cap=butt,line join=miter,line width=0.5pt,miter 
	limit=4.0,dash pattern=on 0.5pt off 2.0pt] (71.819, 37.478) -- (49.343, 99.23);

	\path[draw=c640000,fill=cffebdc,line width=1.001pt,dash pattern=on 8.005pt off
	1.001pt] (58.581, 73.848).. controls (71.129, 73.11) and (86.57, 69.81) .. 
	(95.355, 70.866).. controls (107.51, 72.326) and (106.02, 82.651) .. (102.735,
	87.892).. controls (96.839, 97.302) and (80.555, 89.361) .. (69.451, 89.3).. 
	controls (64.078, 89.27) and (58.128, 90.373) .. (52.27, 91.186) -- cycle;

	\node[text=ccd0000,line cap=butt,line join=miter,line width=1.0pt,miter 
	limit=4.0,anchor=south west] (text20) at (58.078, 73.232){$V\wedge \Up W$};

		%
		%
		%
		%
		%
		%
		%
		%
		%
		%
		%
		%
		%
		%
		%
		%
		%
		%
		%
		%
		%
		%
		%
		%
		%
		%
		%
		%
		%
		%
	
\end{tikzpicture}
		\caption{Strong~\eqref{axiom:wedge+}.}
	\end{subfigure}
	\caption{Illustration of~\eqref{axiom:wedge+} and \cref{lemma:wedge iff strong wedge}.}
	\label{figure:axiom wedge+}
\end{figure}
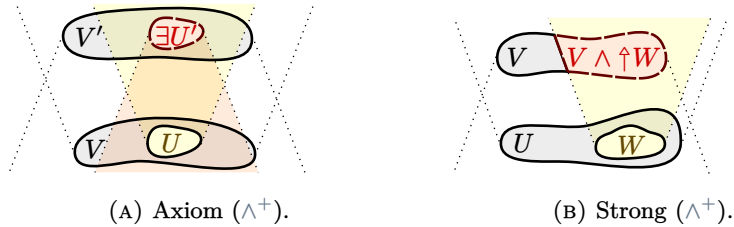

\begin{lemma}\label{lemma:wedge implies cones determine order}
	In any ordered locale satisfying~\eqref{axiom:wedge}, axiom~\eqref{axiom:cones give order} holds.
\end{lemma}
\begin{proof}
	The `only if' direction in~\eqref{axiom:cones give order} holds in any ordered locale by \cref{lemma:properties of localic cones}(a). For the other direction, suppose that $U\sqleq \Down V$ and $V\sqleq \Up U$. Using \cref{lemma:properties of localic cones}(c) we hence get $U\Leq \Up U\sqgeq V$ and $U\sqleq \Down V\Leq V$. Applying axioms~\eqref{axiom:wedge-} and~\eqref{axiom:wedge+}, respectively, we get:
	\[
	\begin{tikzcd}[every label/.append style = {font = \normalsize},column sep=0.25cm, row sep=0.2cm]
		\highlighttikzcd{\exists U'} & {V} \\
		{U} & {\Up U}
		\arrow["\sqleq"{anchor=center, rotate=-90}, draw=none, from=2-1, to=1-1]
		\arrow["\sqleq"{anchor=center, rotate=-90}, draw=none, from=2-2, to=1-2]
		\arrow["{}"{description}, "\Leq"{anchor=center}, draw=none, from=1-1, to=1-2]
		\arrow["{}"{description}, "\Leq"{anchor=center}, draw=none, from=2-1, to=2-2]
	\end{tikzcd}
	\qquad\text{and}\qquad
	\begin{tikzcd}[every label/.append style = {font = \normalsize},column sep=0.25cm, row sep=0.2cm]
		{U} & \highlighttikzcd{\exists V'} \\
		{\Down V} & {V.}
		\arrow["\sqleq"{anchor=center, rotate=-90}, draw=none, from=2-1, to=1-1]
		\arrow["\sqleq"{anchor=center, rotate=-90}, draw=none, from=2-2, to=1-2]
		\arrow["{}"{description}, "\Leq"{anchor=center}, draw=none, from=1-1, to=1-2]
		\arrow["{}"{description}, "\Leq"{anchor=center}, draw=none, from=2-1, to=2-2]
	\end{tikzcd}
	\]
	Applying~\eqref{axiom:V} to the top two rows then gives $U\vee U'\Leq V\vee V'$, which reduces to $U\Leq V$, as desired.
\end{proof}

To explain the terminology `parallel ordered', we need to show how the axioms~\eqref{axiom:wedge} can be reformulated in terms of statements about the localic cones.

\begin{lemma}\label{lemma:wedge iff frobenius}
	An ordered locale $(X,\Leq)$ satisfies~\eqref{axiom:wedge} if and only if~\eqref{axiom:cones give order} holds and the localic cones satisfy the following two conditions:
	\begin{center}
		\dummylabel{axiom:frobenius}{\textsc{f}$^\pm$}
		\begin{minipage}{.4\linewidth}
			\[\tag{\textsc{f}$^-$}\label{axiom:frobenius+}
			\Up U\wedge V\sqleq \Up\left(U\wedge \Down V\right)
			\]
		\end{minipage}%
		\hfil
		\begin{minipage}{.4\linewidth}
			\[\tag{\textsc{f}$^+$}\label{axiom:frobenius-}
			\Down U\wedge V\sqleq \Down\left(U\wedge \Up V\right).
			\]
		\end{minipage}
	\end{center}
\end{lemma}
\begin{proof}
	We already know by \cref{lemma:wedge implies cones determine order} that~\eqref{axiom:wedge} imply~\eqref{axiom:cones give order}. To show that~\eqref{axiom:wedge+} implies~\eqref{axiom:frobenius-}, note that by \cref{lemma:properties of localic cones}(c) we get $\Down U\wedge V\sqleq \Down U \Leq U$, so~\eqref{axiom:wedge+} says there exists an open region $W$ such that $\Down U\wedge V\Leq W\sqleq U$. In turn, \cref{lemma:properties of localic cones}(a) and (e) give $W\sqleq \Up(\Down U \wedge V)\sqleq \Up\Down U\wedge \Up V$, so we find $W\sqleq U\wedge \Up V$. Applying \cref{lemma:properties of localic cones}(a) to $\Down U\wedge V\Leq W$ then gives the desired inclusion:
	\[
	\Down U\wedge V\sqleq \Down W \sqleq \Down \left(U\wedge \Up V\right).
	\]

	Conversely, we need to prove that~\eqref{axiom:frobenius-} implies~\eqref{axiom:wedge+} under~\eqref{axiom:cones give order}. Given $U\sqleq V\Leq V'$, define $U':= \Up U \wedge V'\sqleq V'$. Since $U\sqleq V\sqleq \Down V'$ by \cref{lemma:properties of localic cones}(a), we find using~\eqref{axiom:frobenius-} that
	\[
	U=U\wedge \Down V' \sqleq \Down \left(V'\wedge \Up U\right) = \Down U'.
	\]
	That $U'\sqleq \Up U$ follows by construction, so~\eqref{axiom:cones give order} gives $U\Leq U'$, as desired, hence showing that~\eqref{axiom:wedge+} holds. That~\eqref{axiom:wedge-} follows from~\eqref{axiom:frobenius+} is proved similarly.
\end{proof}

\begin{remark}
	These new conditions are visualised in \cref{figure:axiom:frobenius-}. Note the resemblance between these laws to the so-called `Frobenius reciprocity condition' of open locale maps \cite[\S IX.7]{maclane1994SheavesGeometryLogic}.
\end{remark}

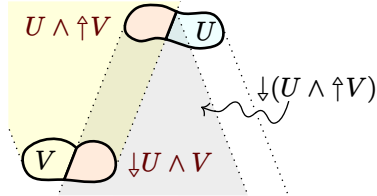
\begin{figure}[b]\centering
	\definecolor{cffffdc}{RGB}{255,255,220}
\definecolor{cededed}{RGB}{237,237,237}
\definecolor{cefffff}{RGB}{239,255,255}
\definecolor{c640000}{RGB}{100,0,0}
\definecolor{cffebdc}{RGB}{255,235,220}
\definecolor{ce0ffdc}{RGB}{224,255,220}
\definecolor{c001764}{RGB}{0,23,100}
\definecolor{c644700}{RGB}{100,71,0}
\definecolor{c0d6400}{RGB}{13,100,0}

\def \globalscale {1.000}
\begin{tikzpicture}[y=.7pt, x=.75pt, yscale=\globalscale,xscale=\globalscale, every node/.append style={scale=\globalscale}, inner sep=0pt, outer sep=0pt]
	\begin{scope}[blend group=multiply]
		
		\path[fill=cffffdc,line cap=butt,line join=miter,line width=1.0pt,miter 
		limit=4.0] (1.386, 32.097) -- (-6.541, 53.876) -- (-6.715, 121.907) -- (80.24,
		121.907) -- (47.479, 31.897) -- cycle;

		\path[fill=cededed,line cap=butt,line join=miter,line width=1.0pt,miter 
		limit=4.0] (53.33, 112.595) -- (78.176, 116.238) -- (114.29, 16.94) -- 
		(18.545, 17.025) -- cycle;

		\path[draw=black,fill=cefffff,line width=1.0pt] (78.242, 116.42) -- (72.784, 
		101.423).. controls (78.955, 99.856) and (84.855, 96.058) .. (91.163, 96.57)..
		controls (94.736, 96.86) and (99.451, 97.575) .. (101.061, 100.778).. 
		controls (103.1, 104.833) and (101.305, 111.053) .. (97.863, 114.013).. 
		controls (93.566, 117.709) and (86.446, 114.213) .. (80.924, 115.488).. 
		controls (80.003, 115.701) and (79.115, 116.034) .. (78.242, 116.42) -- cycle;
	\end{scope}

	\path[draw=black,fill=cffffdc,line width=1.0pt] (21.62, 25.475) -- (27.831, 
	42.537).. controls (20.81, 50.41) and (8.166, 48.216) .. (2.493, 41.956).. 
	controls (-0.558, 38.589) and (0.456, 31.717) .. (3.521, 28.364).. controls 
	(7.399, 24.122) and (14.784, 25.686) .. (20.53, 25.529).. controls (20.885, 
	25.52) and (21.254, 25.496) .. (21.62, 25.475) -- cycle;

	\path[draw=black,line cap=butt,line join=miter,line width=0.5pt,miter 
	limit=4.0,dash pattern=on 0.5pt off 2.0pt] (53.33, 112.595) -- (18.545, 
	17.025);

	\path[draw=black,line cap=butt,line join=miter,line width=0.5pt,miter 
	limit=4.0,dash pattern=on 0.5pt off 2.0pt] (101.253, 108.706) -- (134.622, 
	17.025);

	\node[text=black,line cap=butt,line join=miter,line width=1.0pt,miter 
	limit=4.0,anchor=south west] (text8) at (88, 101){$U$};

	\node[text=black,line cap=butt,line join=miter,line width=1.0pt,miter 
	limit=4.0,anchor=south west] (text9) at (7, 31.272){$V$};

	\node[text=c640000,line cap=butt,line join=miter,line width=1.0pt,miter 
	limit=4.0,anchor=south west] (text10) at (2, 100){$U\wedge \Up V$};

	\node[text=c640000,line cap=butt,line join=miter,line width=1.0pt,miter 
	limit=4.0,anchor=south west] (text12) at (52.917, 26){$\Down U\wedge V$};

	\node[text=black,line cap=butt,line join=miter,line width=1.0pt,miter 
	limit=4.0,anchor=south west] (text14) at (118, 67){$\Down (U\wedge \Up 
		V)$};

	\path[draw=black,fill=cffebdc,line cap=butt,line join=miter,line 
	width=1.0pt,miter limit=4.0] (21.62, 25.475).. controls (30.063, 24.992) and 
	(41.144, 21.785) .. (46.042, 28.364).. controls (49.187, 32.59) and (47.399, 
	40.722) .. (43.207, 43.913).. controls (39.112, 47.03) and (32.148, 47.576) ..
	(27.831, 42.537) -- cycle;

	\path[draw=black,fill=cffebdc,line width=1.0pt] (78.242, 116.42).. controls 
	(75.52, 117.623) and (72.951, 119.355) .. (70.021, 119.679).. controls 
	(65.817, 120.144) and (60.82, 120.407) .. (57.469, 117.826).. controls 
	(54.075, 115.212) and (51.396, 110.177) .. (52.406, 106.014).. controls 
	(53.075, 103.257) and (56.236, 101.355) .. (58.969, 100.598).. controls 
	(62.136, 99.72) and (65.435, 102.213) .. (68.719, 102.064).. controls (70.091,
	102.002) and (71.444, 101.763) .. (72.784, 101.423) -- cycle;

	\path[draw=black,line cap=butt,line join=miter,line width=0.5pt,miter 
	limit=4.0,dash pattern=on 0.5pt off 2.0pt] (78.176, 116.238) -- (114.287, 
	17.025);

	\path[draw=black,line cap=butt,line join=miter,line width=0.5pt,miter 
	limit=4.0,dash pattern=on 0.5pt off 2.0pt] (47.479, 31.897) -- (80.24, 
	121.907);

	\path[->,draw=black,line cap=butt,line join=miter,line width=0.5pt,miter 
	limit=4.0] (134.857, 67.994).. controls (134.857, 67.994) and (132.198, 
	56.551) .. (127.559, 56.048).. controls (124.428, 55.708) and (124.117, 
	62.302) .. (121.018, 62.866).. controls (118.494, 63.325) and (116.559, 
	58.907) .. (114.073, 59.545).. controls (110.931, 60.352) and (111.58, 66.793)
	.. (108.404, 67.455).. controls (105.172, 68.129) and (103.169, 62.837) .. 
	(99.9, 62.38).. controls (96.728, 61.936) and (90.615, 64.853) .. (90.615, 
	64.853);

	\path[draw=black,line cap=butt,line join=miter,line width=0.5pt,miter 
	limit=4.0,dash pattern=on 0.5pt off 2.0pt] (1.386, 32.097) -- (-6.541, 53.876);

		%
		%
		%
		%
		%
		%
		%
		%
		%
		%
		%
		%
		%
		%
		%
		%
		%
		%
		%
		%
		%
		%
		%
		%
		%
		%
		%
	
\end{tikzpicture}
	\caption{Illustration of~\eqref{axiom:frobenius-} in a Minkowski-like space.}
	\label{figure:axiom:frobenius-}
\end{figure}

The following is a very useful reformulation of axioms~\eqref{axiom:wedge}, and is the one we most often use in computations.

\begin{lemma}\label{lemma:wedge iff strong wedge}
	An ordered locale $(X,\Leq)$ satisfies~\eqref{axiom:wedge} if and only if the following two laws are satisfied:
	\[
\begin{tikzcd}[every label/.append style = {font = \normalsize},column sep=0.25cm, row sep=0.2cm]
	{W} & \highlighttikzcd{\Up W\wedge V} \\
	{U} & {V}
	\arrow["\sqleq"{anchor=center, rotate=-90}, draw=none, from=2-1, to=1-1]
	\arrow["\sqleq"{anchor=center, rotate=-90}, draw=none, from=2-2, to=1-2]
	\arrow["{}"{description}, "\Leq"{anchor=center}, draw=none, from=1-1, to=1-2]
	\arrow["{}"{description}, "\Leq"{anchor=center}, draw=none, from=2-1, to=2-2]
\end{tikzcd}
\qquad\text{and}\qquad
\begin{tikzcd}[every label/.append style = {font = \normalsize},column sep=0.25cm, row sep=0.2cm]
	\highlighttikzcd{\Down W\wedge U} & {W} \\
	{U} & {V.}
	\arrow["\sqleq"{anchor=center, rotate=-90}, draw=none, from=2-1, to=1-1]
	\arrow["\sqleq"{anchor=center, rotate=-90}, draw=none, from=2-2, to=1-2]
	\arrow["{}"{description}, "\Leq"{anchor=center}, draw=none, from=1-1, to=1-2]
	\arrow["{}"{description}, "\Leq"{anchor=center}, draw=none, from=2-1, to=2-2]
\end{tikzcd}
\]
\end{lemma}
\begin{proof}
	Clearly these conditions imply~\eqref{axiom:wedge}. To show the converse, note that in the second square we get from~\eqref{axiom:wedge-} that there exists $W'\in\Opens X$ such that ${U\sqgeq W'\Leq W}$. Hence, using \cref{lemma:properties of localic cones} and~\eqref{axiom:frobenius+} we get
	\[
	W =\Up W'\wedge W\sqleq \Up U \wedge W \sqleq \Up \left(U \wedge \Down W\right),
	\]
	so that by~\eqref{axiom:cones give order} we get $\Down W\wedge U\Leq W$. The proof for the first square is similar.
\end{proof}

\begin{remark}
	From this point on we use~\eqref{axiom:wedge} to refer interchangeably to the squares in either \cref{lemma:wedge iff strong wedge} or \cref{definition:parallel ordered locales}. The visual intuition of these stronger axioms is in \cref{figure:axiom wedge+}(b).
\end{remark}

\begin{lemma}\label{lemma:parallel order gives order between arbitrary opens}
	In a parallel ordered locale $(X,\Leq)$, for arbitrary $U,V\in\Opens X$:
	\[
	U\wedge \Down V \Leq V\wedge \Up\left( U \wedge \Down V \right)
	\qquad\text{and}\qquad
	U\wedge \Down \left(\Up U \wedge V\right) \Leq \Up U \wedge V.
	\]
\end{lemma}
\begin{proof}
	By \cref{lemma:wedge implies cones determine order,lemma:wedge iff frobenius} it suffices to give an argument using the localic cones utilising~\eqref{axiom:frobenius}. We prove $U\wedge \Down V \Leq V\wedge \Up\left( U \wedge \Down V \right)$. Trivially, the right hand term is contained in the future of the left hand term. For the other inclusion, use~\eqref{axiom:frobenius-} to calculate:
	\[
	U\wedge \Down V
	=
	\Down V \wedge (U\wedge \Down V)
	\sqleq
	\Down(V\wedge \Up(U\wedge\Down V)),
	\]
	which is the desired inclusion. The other case is proved dually using~\eqref{axiom:frobenius+}.
\end{proof}

Of course, generally, the intersection $U\wedge \Down V$ may well be empty, so that the lemma tells us simply that $\varnothing\Leq \varnothing$. While seemingly trivial, this is actually an important case, and captures the essence of what we mean by parallel orderedness. The visual idea is captured in \cref{figure:parallel vs nonparallel ordered}.

\begin{proposition}\label{proposition:cones are parallel}
	For any $U,V\in\Opens X$ in a parallel ordered locale $(X,\Leq)$ we have:
	\[
	U\wedge \Down V =\varnothing
	\qquad\text{if and only if}\qquad
	\Up U \wedge V =\varnothing.
	\]
\end{proposition}
\begin{proof}
	Assume $U\wedge\Down V =\varnothing$. Using~\eqref{axiom:empty} and~\eqref{axiom:frobenius+} immediately gives the inclusion $\Up U\wedge V \sqleq \Up (U\wedge\Down V)= \Up \varnothing =\varnothing$. The other implication follows using~\eqref{axiom:frobenius-}.
\end{proof}

\begin{figure}[t]\centering
	\begin{subfigure}[b]{0.3\textwidth}
		\definecolor{cefffff}{RGB}{239,255,255}
\definecolor{cffffdc}{RGB}{255,255,220}
\definecolor{c000208}{RGB}{0,2,8}
\definecolor{cffebdc}{RGB}{255,235,220}
\definecolor{ce0ffdc}{RGB}{224,255,220}
\definecolor{c001764}{RGB}{0,23,100}
\definecolor{c640000}{RGB}{100,0,0}
\definecolor{ccd0000}{RGB}{205,0,0}
\definecolor{c644700}{RGB}{100,71,0}
\definecolor{c0d6400}{RGB}{13,100,0}

\def \globalscale {1.000000}
\begin{tikzpicture}[y=.75pt, x=.75pt, yscale=\globalscale,xscale=\globalscale, every node/.append style={scale=\globalscale}, inner sep=0pt, outer sep=0pt]
	\begin{scope}[blend group = multiply]
	\path[fill=cefffff,line cap=butt,line join=miter,line width=1.0pt,miter 
	limit=4.0] (26.987, 22.695) -- (62.147, 83.594) -- (68.719, 68.049) -- 
	(92.819, 74.894) -- (108.404, 48.206) -- (108.404, 22.695) -- cycle;

	\path[fill=cffffdc,line cap=butt,line join=miter,line width=1.0pt,miter 
	limit=4.0] (21.742, 37.101) -- (35.134, 37.101) -- (49.337, 29.375) -- 
	(49.337, 42.537) -- (51.711, 42.537) -- (51.711, 93.561) -- (21.742, 93.561) 
	-- cycle;
	\end{scope}

	\path[draw=black,line cap=butt,line join=miter,line width=0.5pt,miter 
	limit=4.0,dash pattern=on 0.5pt off 2.0pt] (21.742, 37.101) -- (21.742, 
	93.561);

	\path[draw=black,fill=cefffff,line cap=butt,line join=miter,line 
	width=1.0pt,miter limit=4.0] (64.538, 72.794).. controls (63.816, 71.652) and 
	(61.512, 73.676) .. (60.683, 71.54).. controls (59.701, 69.01) and (66.844, 
	67.97) .. (70.146, 66.673).. controls (75.675, 64.499) and (81.769, 59.692) ..
	(87.305, 61.849).. controls (91.269, 63.394) and (95.18, 68.615) .. (94.016, 
	72.708).. controls (92.977, 76.361) and (87.711, 79.017) .. (83.82, 77.797).. 
	controls (79.342, 76.392) and (76.105, 69.736) .. (71.554, 70.884).. controls 
	(66.408, 72.181) and (68.042, 86.14) .. (63.05, 84.341).. controls (59.399, 
	83.025) and (66.613, 76.073) .. (64.538, 72.794) -- cycle;

	\path[draw=black,fill=cffffdc,line width=1.0pt] (33.829, 34.545) -- (40.89, 
	46.775).. controls (38.985, 48.868) and (37.95, 53.023) .. (35.467, 52.285).. 
	controls (32.949, 51.536) and (35.338, 46.725) .. (33.826, 44.577).. controls 
	(31.1, 40.703) and (20.905, 41.763) .. (21.741, 37.101).. controls (22.427, 
	33.279) and (29.445, 35.922) .. (33.163, 34.797).. controls (33.383, 34.731) 
	and (33.606, 34.634) .. (33.829, 34.545) -- cycle;

	\path[draw=black,line cap=butt,line join=miter,line width=0.5pt,miter 
	limit=4.0,dash pattern=on 0.5pt off 2.0pt] (49.337, 29.375) -- (49.337, 
	39.703);

	\path[draw=black,line cap=butt,line join=miter,line width=0.5pt,miter 
	limit=4.0,dash pattern=on 0.5pt off 2.0pt] (51.711, 42.537) -- (51.711, 
	93.561);

	\path[draw=black,line cap=butt,line join=miter,line width=0.5pt,miter 
	limit=4.0,dash pattern=on 0.5pt off 2.0pt] (62.147, 83.594) -- (26.987, 
	22.695);

	\path[draw=black,line cap=butt,line join=miter,line width=0.5pt,miter 
	limit=4.0,dash pattern=on 0.5pt off 2.0pt] (92.819, 74.894) -- (108.226, 
	48.206);

	\node[text=c000208,line cap=butt,line join=miter,line width=1.0pt,miter 
	limit=4.0,anchor=south west] (text13) at (34.282, 6.689){$U\wedge \Down V 
		\neq\varnothing$};

	\path[draw=black,fill=cffebdc,line cap=butt,line join=miter,line 
	width=1.0pt,miter limit=4.0] (33.829, 34.545).. controls (39.086, 32.45) and 
	(44.988, 24.573) .. (48.876, 28.364).. controls (50.992, 30.427) and (45.095, 
	33.766) .. (45.727, 36.653).. controls (46.324, 39.386) and (52.32, 39.807) ..
	(51.71, 42.537).. controls (50.944, 45.972) and (44.729, 44.164) .. (41.748, 
	46.034).. controls (41.44, 46.227) and (41.157, 46.482) .. (40.89, 46.775) -- 
	cycle;

	\node[text=c000208,line cap=butt,line join=miter,line width=1.0pt,miter 
	limit=4.0,anchor=south west] (text17) at (34.537, 97){$\Up U \wedge V 
		=\varnothing$};

		%
		%
		%
		%
		%
		%
		%
		%
		%
		%
		%
		%
		%
		%
		%
		%
		%
		%
		%
		%
		%
		%
		%
		%
		%
		%
		%
		%
		%
		%
	
\end{tikzpicture}
		\caption{Non-parallel cones.}
	\end{subfigure}\hfil
	\begin{subfigure}[b]{0.3\textwidth}
		\definecolor{cffffdc}{RGB}{255,255,220}
\definecolor{cefffff}{RGB}{239,255,255}
\definecolor{c000208}{RGB}{0,2,8}
\definecolor{cffebdc}{RGB}{255,235,220}
\definecolor{ce0ffdc}{RGB}{224,255,220}
\definecolor{c001764}{RGB}{0,23,100}
\definecolor{c640000}{RGB}{100,0,0}
\definecolor{ccd0000}{RGB}{205,0,0}
\definecolor{c644700}{RGB}{100,71,0}
\definecolor{c0d6400}{RGB}{13,100,0}

\def \globalscale {1.000000}
\begin{tikzpicture}[y=.75pt, x=.75pt, yscale=\globalscale,xscale=\globalscale, every node/.append style={scale=\globalscale}, inner sep=0pt, outer sep=0pt]
	\begin{scope}[blend group = multiply]
		\path[fill=cffffdc,line cap=butt,line join=miter,line width=1.0pt,miter 
		limit=4.0] (21.742, 36.868) -- (7.013, 62.38) -- (7.013, 93.686) -- (86.518, 
		93.561) -- (48.876, 28.364) -- (34.703, 39.703) -- cycle;

		\path[fill=cefffff,line cap=butt,line join=miter,line width=1.0pt,miter 
		limit=4.0] (26.987, 22.695) -- (62.147, 83.594) -- (68.719, 68.049) -- 
		(92.819, 74.894) -- (108.404, 48.206) -- (108.404, 22.695) -- cycle;
	\end{scope}
	
	\path[draw=black,fill=cefffff,line width=1.0pt] (73.393, 70.829) -- (70.83, 
	66.387).. controls (76.175, 64.11) and (81.996, 59.781) .. (87.305, 61.85).. 
	controls (91.269, 63.395) and (95.181, 68.615) .. (94.017, 72.707).. controls 
	(92.977, 76.361) and (87.711, 79.017) .. (83.82, 77.796).. controls (79.968, 
	76.588) and (77.034, 71.5) .. (73.393, 70.829) -- cycle;

	\path[draw=black,fill=cffffdc,line width=1.0pt] (33.829, 34.545) -- (40.89, 
	46.775).. controls (39.164, 48.672) and (38.152, 52.263) .. (36.129, 52.369)..
	controls (35.92, 52.381) and (35.7, 52.354) .. (35.467, 52.285).. controls 
	(32.949, 51.536) and (35.338, 46.725) .. (33.826, 44.577).. controls (31.1, 
	40.703) and (20.905, 41.763) .. (21.741, 37.101).. controls (22.427, 33.279) 
	and (29.445, 35.922) .. (33.163, 34.797).. controls (33.383, 34.731) and 
	(33.606, 34.634) .. (33.829, 34.545) -- cycle;

	\path[draw=black,line cap=butt,line join=miter,line width=0.5pt,miter 
	limit=4.0,dash pattern=on 0.5pt off 2.0pt] (62.147, 83.594) -- (26.987, 
	22.695);

	\path[draw=black,line cap=butt,line join=miter,line width=0.5pt,miter 
	limit=4.0,dash pattern=on 0.5pt off 2.0pt] (92.819, 74.894) -- (108.226, 
	48.206);

	\node[text=c000208,line cap=butt,line join=miter,line width=1.0pt,miter 
	limit=4.0,anchor=south west] (text13) at (34.282, 6.689){$U\wedge \Down V 
		\neq\varnothing$};

	\node[text=c000208,line cap=butt,line join=miter,line width=1.0pt,miter 
	limit=4.0,anchor=south west] (text17) at (34.537, 97){$\Up U \wedge V 
		\neq\varnothing$};

	\path[draw=black,line cap=butt,line join=miter,line width=0.5pt,miter 
	limit=4.0,dash pattern=on 0.5pt off 2.0pt] (21.742, 36.868) -- (7.013, 62.38);

	\path[draw=black,line cap=butt,line join=miter,line width=0.5pt,miter 
	limit=4.0,dash pattern=on 0.5pt off 2.0pt] (48.876, 28.364) -- (86.518, 
	93.561);

	\path[draw=black,fill=cffebdc,line width=1.0pt] (33.829, 34.545).. controls 
	(39.086, 32.45) and (44.988, 24.573) .. (48.876, 28.364).. controls (50.992, 
	30.427) and (45.095, 33.766) .. (45.727, 36.653).. controls (46.324, 39.386) 
	and (52.32, 39.807) .. (51.71, 42.537).. controls (50.944, 45.972) and 
	(44.729, 44.164) .. (41.748, 46.034).. controls (41.44, 46.227) and (41.157, 
	46.482) .. (40.89, 46.775) -- cycle;

	\path[draw=black,fill=cffebdc,line width=1.0pt] (73.393, 70.829).. controls 
	(72.801, 70.72) and (72.19, 70.723) .. (71.553, 70.884).. controls (66.408, 
	72.181) and (68.042, 86.14) .. (63.05, 84.341).. controls (59.399, 83.025) and
	(66.613, 76.074) .. (64.538, 72.794).. controls (63.816, 71.652) and (61.512,
	73.676) .. (60.683, 71.54).. controls (59.7, 69.01) and (66.844, 67.97) .. 
	(70.146, 66.672).. controls (70.372, 66.583) and (70.601, 66.484) .. (70.83, 
	66.387) -- cycle;

		%
		%
		%
		%
		%
		%
		%
		%
		%
		%
		%
		%
		%
		%
		%
		%
		%
		%
		%
		%
		%
		%
		%
		%
		%
		%
		%
		%
		%
		%
	
\end{tikzpicture}
		\caption{Parallel cones.}
	\end{subfigure}
	\caption{Intuition of non-parallel vs.~parallel cones.}
	\label{figure:parallel vs nonparallel ordered}
\end{figure}
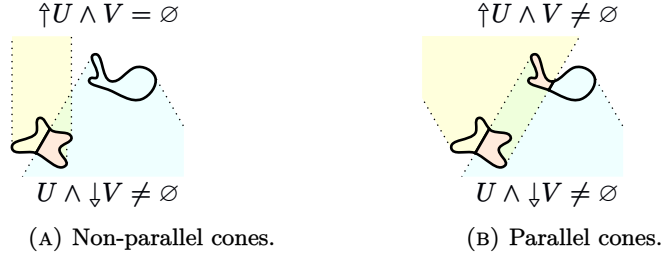

To close this section, we remark on parallel orderedness in spaces.
\begin{proposition}\label{proposition:space with OC is parallel ordered}
	If $(S,\leq)$ has open cones, then $\loc(S)$ is parallel ordered.
\end{proposition}
\begin{proof}
	From the characterisation of the localic cones in a space, it follows that $\Up \varnothing = (\up \varnothing)^\circ = \varnothing$, and similarly $\Down \varnothing = \varnothing$, so $\loc(S)$ satisfies~\eqref{axiom:empty} (even without open cones).
	
	In a space with open cones, axiom~\eqref{axiom:frobenius+} says that $\up U \cap V \subseteq \up\left(U\cap \down V\right)$. This is almost trivially true: if $y\in \up U\cap V$, then there exists $x\in U$ such that $x\leq y$. But then $x\in \down \{y\}\subseteq \down V$, so $y \in \up\left(U\cap \down V\right)$. The dual axiom~\eqref{axiom:frobenius-} follows similarly, and since ordered spaces satisfy~\eqref{axiom:cones give order}, the result follows by \cref{lemma:wedge iff frobenius}.
\end{proof}

\begin{example}\label{example:non-OC is not parallel}
	Locales induced by spaces without open cones are not always parallel ordered. Consider the space $\{\ast\}\sqcup \mathbb{R}$, where $\ast\leq 0$, which does not have open cones since $\up\{\ast\}$ is not open. We can then see that~\eqref{axiom:frobenius-} is violated, since it would imply a contradiction:
	\[
	\{\ast\} = \{\ast\}\cap \Down \mathbb{R} \subseteq \Down\left( \Up \{\ast\}\cap \mathbb{R}\right) = \Down \left(\{0\}^\circ \cap \mathbb{R}\right)= \Down \varnothing = \varnothing.
	\]
\end{example}

\begin{example}\label{example:minkowski different speeds of light}
	We construct a modified example of Minkowski space where the speed of light into the future is independent from the speed of light into the past. To do this, we consider the underlying space $\mathbb{R}^2$, and equip it with a family of preorders $\leq_\alpha$ defined by the Minkowski metric with speed of light $1/\alpha$. Explicitly:
	\[
	x\leq_\alpha y
	\qquad\text{if and only if}\qquad
	\parbox{.35\textwidth}{the straight line connecting $x$ and $y$ makes an angle with the horizontal axis of at least $\alpha$ radians.}
	\]
	For $\alpha\in (0,\pi/2)$ the pair $(\mathbb{R}^2,\leq_\alpha)$ comes from a spacetime, and so must have open cones (\cref{theorem:spacetimes have OC}). In particular we get monads $\Upsub{\alpha}$ and $\Downsub{\alpha}$ on $\Opens\mathbb{R}^2$.
	
	For any two numbers $\alpha,\beta\in (0,\pi/2)$ we therefore get an ordered locale $\loc(\mathbb{R}^2)$ whose order $\Leq$ is defined by the pair of cones $\Upsub{\alpha}$ and $\Downsub{\beta}$. See \cref{figure:minkowski diferent cones}(a). We claim that this locale is parallel ordered if and only if $\alpha = \beta$. In the latter case, where the slopes are indeed equal, the ordered locale $\loc(\mathbb{R}^2)$ is just the one coming from Minkowski space, which is parallel ordered by \cref{proposition:space with OC is parallel ordered}. Conversely, suppose for the sake of contradiction that $\alpha\neq \beta$. It is then straightforward to construct an example, such as in \cref{figure:minkowski diferent cones}(b), that contradicts the conclusion of \cref{proposition:cones are parallel}, showing that $\loc(\mathbb{R}^2)$ is not parallel ordered.
	
	An interesting non-spacetime example of a parallel ordered locale is \emph{vertical-$\mathbb{R}^2$}, induced by the order $\leq_{\pi/2}$, interpreted as Minkowski space where the speed of light is zero, see \cite[Chapter~7]{schaaf2024TowardsPointFreeSpacetimes}.
\end{example}

\begin{figure}[t]\centering
	\begin{subfigure}[b]{0.33\textwidth}\centering
		\definecolor{cffebdc}{RGB}{255,235,220}
\definecolor{ce0ffdc}{RGB}{224,255,220}
\definecolor{cffffdc}{RGB}{255,255,220}
\definecolor{c0d6400}{RGB}{13,100,0}
\definecolor{c640000}{RGB}{100,0,0}
\definecolor{cefffff}{RGB}{239,255,255}
\definecolor{c001764}{RGB}{0,23,100}
\definecolor{ccd0000}{RGB}{205,0,0}
\definecolor{c644700}{RGB}{100,71,0}

\def \globalscale {1.000000}
\begin{tikzpicture}[y=.8pt, x=.8pt, yscale=\globalscale,xscale=\globalscale, every node/.append style={scale=\globalscale}, inner sep=0pt, outer sep=0pt]
	
	\begin{scope}[blend group = multiply]
		\path[fill=cffebdc,line cap=butt,line join=miter,line width=1.0pt,miter 
		limit=4.0] (48.876, 79.388) -- (0.687, 61.848) -- (0.687, 39.703) -- (128.247,
		39.776) -- (128.247, 66.672) -- (82.121, 83.461) -- cycle;

		\path[fill=ce0ffdc,line cap=butt,line join=miter,line width=1.0pt,miter 
		limit=4.0] (31.869, 65.214) -- (94.023, 69.47) -- (101.27, 110.569) -- 
		(23.871, 110.569) -- cycle;
	\end{scope}

	\path[draw=black,fill=cffffdc,line cap=butt,line join=miter,line 
	width=1.0pt,miter limit=4.0] (31.869, 65.214).. controls (31.533, 61.824) and 
	(34.36, 57.939) .. (37.538, 56.71).. controls (44.805, 53.901) and (52.47, 
	63.23) .. (60.215, 62.38).. controls (64.415, 61.919) and (67.333, 56.51) .. 
	(71.554, 56.71).. controls (80.458, 57.133) and (93.624, 61.99) .. (94.231, 
	70.884).. controls (94.733, 78.246) and (84.376, 83.241) .. (77.223, 85.057)..
	controls (67.883, 87.428) and (57.673, 83.321) .. (48.876, 79.388).. controls
	(45.651, 77.945) and (32.113, 67.686) .. (31.869, 65.214) -- cycle;

	\path[draw=black,line cap=butt,line join=miter,line width=0.5pt,miter 
	limit=4.0,dash pattern=on 0.5pt off 2.0pt] (94.023, 69.47) -- (101.27, 
	110.569);

	\path[draw=black,line cap=butt,line join=miter,line width=0.5pt,miter 
	limit=4.0,dash pattern=on 0.5pt off 2.0pt] (31.869, 65.214) -- (23.871, 
	110.569);

	\path[draw=black,line cap=butt,line join=miter,line width=0.5pt,miter 
	limit=4.0,dash pattern=on 0.5pt off 2.0pt] (48.876, 79.388) -- (0.687, 61.848);

	\path[draw=black,line cap=butt,line join=miter,line width=0.5pt,miter 
	limit=4.0,dash pattern=on 0.5pt off 2.0pt] (82.121, 83.461) -- (128.247, 
	66.672);

	\path[draw=black,fill=cffffdc,line cap=butt,line join=miter,line 
	width=0.5pt,miter limit=4.0] (94.231, 69.47) -- (128.247, 69.47);

	\path[draw=black,line cap=butt,line join=miter,line width=0.5pt,miter 
	limit=4.0] (97.301, 87.868)arc(79.895:-0.05000000000001137:18.681) -- (94.023,
	69.477) -- cycle;

	\node[text=black,line cap=butt,line join=miter,line width=1.0pt,miter 
	limit=4.0,anchor=south west] (text45) at (55.875, 69.211){$U$};

	\node[text=c0d6400,line cap=butt,line join=miter,line width=1.0pt,miter 
	limit=4.0,anchor=south west] (text46) at (51.294, 92.662){$\Upsub{\alpha} U$};

	\node[text=c640000,line cap=butt,line join=miter,line width=1.0pt,miter 
	limit=4.0,anchor=south west] (text47) at (48.774, 41){$\Downsub{\beta}U$};

	\node[text=black,line cap=butt,line join=miter,line width=1.0pt,miter 
	limit=4.0,anchor=south west] (text48) at (108.094, 83.461){$\alpha$};

	\path[draw=black,line cap=butt,line join=miter,line width=0.5pt,miter 
	limit=4.0] (48.876, 79.388) -- (0.687, 79.388);

	\path[draw=black,line cap=butt,line join=miter,line width=0.5pt,miter 
	limit=4.0] (23.01, 69.992)arc(199.697:180.074:27.529 and 27.838) -- (48.928, 
	79.375) -- cycle;

	\node[text=black,line cap=butt,line join=miter,line width=1.0pt,miter 
	limit=4.0,anchor=south west] (text49) at (6.844, 67){$\beta$};

		%
		%
		%
		%
		%
		%
		%
		%
		%
		%
		%
		%
		%
		%
		%
		%
		%
		%
		%
		%
		%
		%
		%
		%
		%
		%
		%
		%
		%
		%
	
\end{tikzpicture}
		\caption{Definition of $\Upsub{\alpha}$ and $\Downsub{\beta}$.}
	\end{subfigure}\hfil
	\begin{subfigure}[b]{0.45\textwidth}\centering
		\definecolor{ce0ffdc}{RGB}{224,255,220}
\definecolor{cffebdc}{RGB}{255,235,220}
\definecolor{cffffdc}{RGB}{255,255,220}
\definecolor{c0d6400}{RGB}{13,100,0}
\definecolor{c640000}{RGB}{100,0,0}
\definecolor{cefffff}{RGB}{239,255,255}
\definecolor{c001764}{RGB}{0,23,100}
\definecolor{ccd0000}{RGB}{205,0,0}
\definecolor{c644700}{RGB}{100,71,0}

\def \globalscale {1.000000}
\begin{tikzpicture}[y=.8pt, x=.8pt, yscale=\globalscale,xscale=\globalscale, every node/.append style={scale=\globalscale}, inner sep=0pt, outer sep=0pt]
	
	\begin{scope}[blend group = multiply]
		\path[fill=ce0ffdc,line cap=butt,line join=miter,line width=1.0pt,miter 
		limit=4.0] (62.768, 60.453) -- (53.932, 110.569) -- (104.499, 110.569) -- 
		(95.502, 59.545) -- cycle;

		\path[fill=cffebdc,line cap=butt,line join=miter,line width=1.0pt,miter 
		limit=4.0] (24.355, 104.008) -- (-13.486, 90.235) -- (-13.486, 42.537) -- 
		(128.247, 42.537) -- (128.247, 72.428) -- (46.042, 102.348) -- cycle;
	\end{scope}

	\path[draw=black,fill=cffffdc,line cap=butt,line join=miter,line 
	width=0.5pt,miter limit=4.0] (95.842, 59.502) -- (128.247, 59.502);

	\path[draw=black,line cap=butt,line join=miter,line width=0.5pt,miter 
	limit=4.0] (98.184, 74.65)arc(80.426:-0.05000000000001137:15.355 and 15.353) 
	-- (95.63, 59.51) -- cycle;

	\node[text=c0d6400,line cap=butt,line join=miter,line width=1.0pt,miter 
	limit=4.0,anchor=south west] (text46) at (69.182, 95){$\Upsub{\alpha} V$};

	\node[text=c640000,line cap=butt,line join=miter,line width=1.0pt,miter 
	limit=4.0,anchor=south west] (text47) at (9.247, 54.261){$\Downsub{\beta}U$};

	\node[text=black,line cap=butt,line join=miter,line width=1.0pt,miter 
	limit=4.0,anchor=south west] (text48) at (115, 64.735){$\alpha$};

	\path[draw=black,line cap=butt,line join=miter,line width=0.5pt,miter 
	limit=4.0] (23.72, 103.757) -- (-13.486, 103.757);

	\path[draw=black,line cap=butt,line join=miter,line width=0.5pt,miter 
	limit=4.0] (0.687, 95.388)arc(199.697:180.074:24.519 and 24.793) -- (23.771, 
	103.744) -- cycle;

	\node[text=black,line cap=butt,line join=miter,line width=1.0pt,miter 
	limit=4.0,anchor=south west] (text49) at (-13.744, 91.5){$\beta$};

	\path[draw=black,fill=cffffdc,line cap=butt,line join=miter,line 
	width=1.0pt,miter limit=4.0] (34.703, 104.899).. controls (29.932, 105.569) 
	and (24.462, 104.849) .. (20.53, 102.065).. controls (17.211, 99.715) and 
	(12.62, 95.28) .. (14.299, 91.576).. controls (17.721, 84.027) and (30.659, 
	85.221) .. (38.747, 87.038).. controls (43.858, 88.186) and (51.879, 91.015) 
	.. (51.489, 96.239).. controls (51.02, 102.517) and (40.938, 104.024) .. 
	(34.703, 104.899) -- cycle;

	\path[draw=black,fill=cefffff,line cap=butt,line join=miter,line 
	width=1.0pt,miter limit=4.0] (81.187, 68.049).. controls (76.89, 67.366) and 
	(72.037, 70.816) .. (68.156, 68.847).. controls (65.191, 67.343) and (62.24, 
	63.735) .. (62.768, 60.453).. controls (63.561, 55.532) and (69.459, 51.781) 
	.. (74.388, 51.041).. controls (81.892, 49.915) and (93.835, 50.087) .. 
	(95.502, 59.545).. controls (96.752, 66.635) and (98.31, 71.297) .. (94.231, 
	73.718).. controls (90.154, 76.138) and (80.222, 72.69) .. (81.187, 68.049).. 
	controls (81.187, 68.049) and (81.187, 68.049) .. (81.187, 68.049) -- cycle;

	\path[draw=black,line cap=butt,line join=miter,line width=0.5pt,miter 
	limit=4.0,dash pattern=on 0.5pt off 2.0pt] (24.355, 104.008) -- (-13.486, 
	90.235);

	\path[draw=black,line cap=butt,line join=miter,line width=0.5pt,miter 
	limit=4.0,dash pattern=on 0.5pt off 2.0pt] (46.621, 102.137) -- (128.247, 
	72.428);

	\path[draw=black,line cap=butt,line join=miter,line width=0.5pt,miter 
	limit=4.0,dash pattern=on 0.5pt off 2.0pt] (62.768, 60.453) -- (53.932, 
	110.569);

	\path[draw=black,line cap=butt,line join=miter,line width=0.5pt,miter 
	limit=4.0,dash pattern=on 0.5pt off 2.0pt] (95.502, 59.545) -- (104.499, 
	110.569);

	\node[text=black,line cap=butt,line join=miter,line width=1.0pt,miter 
	limit=4.0,anchor=south west] (text57) at (20.53, 90.726){$U$};

	\node[text=black,line cap=butt,line join=miter,line width=1.0pt,miter 
	limit=4.0,anchor=south west] (text57-5) at (71.796, 54.936){$V$};

		%
		%
		%
		%
		%
		%
		%
		%
		%
		%
		%
		%
		%
		%
		%
		%
		%
		%
		%
		%
		%
		%
		%
		%
		%
		%
		%
		%
		%
		%
	
\end{tikzpicture}
		\caption{Failure of being parallel.}
		\label{figure:minkowski different cones not parallel}
	\end{subfigure}
	\caption{Illustration of the cones in \cref{example:minkowski different speeds of light}.}
	\label{figure:minkowski diferent cones}
\end{figure}
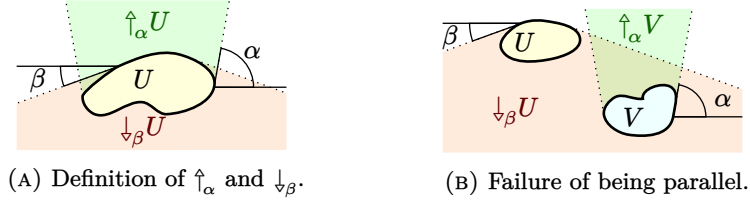

To summarise, we take the framework of ordered locales $(X,\Leq)$ satisfying axioms \eqref{axiom:LV} and~\eqref{axiom:wedge} (which together imply~\eqref{axiom:cones give order} and~\eqref{axiom:empty}) as our basis in which to develop notions of causality theory in a point-free setting. Note by \cref{theorem:spacetimes have OC,proposition:space with OC has join-preserving cones,proposition:space with OC is parallel ordered} that the ordered locale coming from any smooth spacetime satisfies these axioms.

\section{Paths in ordered locales}
\noindent\begin{minipage}{0.68\textwidth}
	The remainder of the paper is dedicated to defining and studying causal coverage in ordered locales. To start, we need the point-free analogue of monotone paths $\gamma\colon [0,1]\to S$. Instead of points, the idea is to take totally ordered families $(p_i)_{i\in I}$ of open regions that cover the image $\im(\gamma)\subseteq S$, such that $p_i\Leq p_j$ whenever $i < j$ (\cref{figure:covered path}). Specifically, we might think of a family $(p_t)_{t\in [0,1]}$ of open neighbourhoods $p_t\ni \gamma(t)$. However, since the unit interval is compact and $\gamma$ is continuous, the image $\im(\gamma)$ is compact, and so $(p_i)_{i\in I}$ will have a finite subcover. For this reason, and due to technical restrictions, we shall only consider paths to be defined in terms of \emph{finite} sequences, instead of arbitrary total orders.
\end{minipage}%
\hfill%
\begin{minipage}{.4\textwidth}\centering
	\definecolor{cffffdc}{RGB}{255,255,220}
\definecolor{cefffff}{RGB}{239,255,255}
\definecolor{c640000}{RGB}{100,0,0}

\def \globalscale {1.000000}
\begin{tikzpicture}[y=.75pt, x=.75pt, yscale=\globalscale,xscale=\globalscale, every node/.append style={scale=\globalscale}, inner sep=0pt, outer sep=0pt]
  \begin{scope}[shift={(-271.5, 73.66)}]
    \begin{scope}[shift={(-93.64, 0.7)}]
      \path[draw=black,fill=cffffdc,line cap=butt,line join=miter,line 
  width=1.0pt,miter limit=4.0] (425.36, -41.74).. controls (419.21, -43.19) and 
  (408.38, -53.38) .. (415.77, -60.2).. controls (423.94, -67.75) and (432.89, 
  -60.93) .. (436.7, -57.33).. controls (440.5, -53.73) and (434.79, -39.52) .. 
  (425.36, -41.74) -- cycle;

      \path[draw=black,fill=cefffff,line cap=butt,line join=miter,line 
  width=1.0pt,miter limit=4.0] (422.52, 53.22).. controls (415.94, 46.93) and 
  (419.69, 41.88) .. (425.36, 39.05).. controls (431.03, 36.21) and (436.27, 
  38.47) .. (439.53, 41.88).. controls (443.95, 46.5) and (442.29, 52.95) .. 
  (437.16, 55.66).. controls (431.03, 58.89) and (427.02, 57.51) .. (422.52, 
  53.22) -- cycle;

      \path[draw=c640000,line cap=butt,line join=miter,line width=0.5pt,miter 
  limit=4.0,dash pattern=on 3.0pt off 1.0pt] (440.85, -28.46).. controls (435.9,
   -19.51) and (420.64, -21.76) .. (416.85, -26.15).. controls (413.07, -30.54) 
  and (415.88, -44.2) .. (422.52, -45.99).. controls (430.68, -48.2) and 
  (446.87, -39.34) .. (440.85, -28.46) -- cycle;

      \path[draw=c640000,line cap=butt,line join=miter,line width=0.5pt,miter 
  limit=4.0,dash pattern=on 3.0pt off 1.0pt] (440.0, -5.31).. controls (434.22, 
  0.83) and (417.41, -5.49) .. (416.18, -13.83).. controls (415.25, -20.12) and 
  (425.37, -26.82) .. (431.5, -25.15).. controls (438.44, -23.27) and (444.93, 
  -10.55) .. (440.0, -5.31) -- cycle;

      \path[draw=c640000,line cap=butt,line join=miter,line width=0.5pt,miter 
  limit=4.0,dash pattern=on 3.0pt off 1.0pt] (442.08, 2.38).. controls (444.38, 
  10.88) and (426.6, 21.82) .. (419.69, 16.37).. controls (413.35, 11.37) and 
  (410.65, -2.5) .. (427.39, -6.6).. controls (432.96, -7.96) and (440.0, -5.31)
   .. (442.08, 2.38) -- cycle;

      \path[draw=c640000,line cap=butt,line join=miter,line width=0.5pt,miter 
  limit=4.0,dash pattern=on 3.0pt off 1.0pt] (439.53, 27.71).. controls (436.37,
   35.94) and (438.13, 44.41) .. (429.14, 42.94).. controls (420.64, 41.55) and 
  (417.88, 29.24) .. (418.54, 24.66).. controls (419.2, 20.09) and (422.29, 
  11.31) .. (430.56, 13.74).. controls (439.53, 16.37) and (442.39, 20.27) .. 
  (439.53, 27.71) -- cycle;

    \end{scope}
    \path[fill=black,line cap=butt,line join=miter,line width=1.0pt,miter 
  limit=4.0] (331.72, -49.55) circle (1.42pt);

    \path[fill=black,line cap=butt,line join=miter,line width=1.0pt,miter 
  limit=4.0] (336.92, 48.3) circle (1.42pt);

    \path[draw=black,line cap=butt,line join=miter,line width=1.0pt,miter 
  limit=4.0] (331.72, -49.55).. controls (331.72, -49.55) and (336.97, -24.17) 
  .. (337.39, -11.28).. controls (337.75, 0.09) and (334.53, 11.36) .. (334.55, 
  22.74).. controls (334.57, 31.29) and (336.8, 48.3) .. (336.8, 48.3);

    \node[text=c640000,line cap=butt,line join=miter,line width=1.0pt,miter 
  limit=4.0,anchor=south west] (text18) at (354.44, 5.36){$p_3$};

    \node[text=black,line cap=butt,line join=miter,line width=1.0pt,miter 
  limit=4.0,anchor=south west] (text11) at (348.73, -56.63){$\gamma(0)$};

    \node[text=black,line cap=butt,line join=miter,line width=1.0pt,miter 
  limit=4.0,anchor=south west] (text11-9) at (301.13, 47.2){$\gamma(1)$};

    \node[text=c640000,line cap=butt,line join=miter,line width=1.0pt,miter 
  limit=4.0,anchor=south west] (text11-9-5) at (354.39, -33.95){$p_1$};

    \node[text=c640000,line cap=butt,line join=miter,line width=1.0pt,miter 
  limit=4.0,anchor=south west] (text11-9-5-1) at (354.39, -14.49){$p_2$};

    \node[text=c640000,line cap=butt,line join=miter,line width=1.0pt,miter 
  limit=4.0,anchor=south west] (text11-9-5-4-7-9) at (352.79, 25.44){$p_4$};

  \end{scope}

\end{tikzpicture}
	
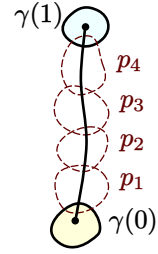
\captionof{figure}{Illustration of a path $(p_i)_{i\in I}$ covering $\im(\gamma)$.}
	\label{figure:covered path}
\end{minipage}

\begin{definition}\label{definition:paths}
	A \emph{path} in an ordered locale $(X,\Leq)$ is a finite sequence $p=(p_n)_{n=0}^N$ of non-empty open regions, called \emph{steps}, such that $p_n\Leq p_{n+1}$ for all ${0\leq n <N}$. 
	
	We call $p_0$ and $p_N$ the \emph{starting point} and \emph{endpoint} of $p$, respectively. When the indexing is not specified, the endpoint of $p$ will be denoted $p_\top$, and the starting point $p_\bot$. 
	
	We say $p$ \emph{inhabits} $V\in\Opens X$ if $p_n\sqleq V$ for some index $n$. We say $p$ \emph{lands in} $V$ if its endpoint is contained in $V$.
	
	Given paths $p=(p_n)_{n=0}^N$ and $q=(q_m)_{m=0}^M$ with $p_N = q_0$, we define their \emph{concatenation} as the path $q\cdot p$ whose steps are defined as $(q\cdot p)_k = p_k$ for $0\leq k\leq N$ and $(q\cdot p)_k=q_{k-N}$ for $N\leq k \leq M+N$. 
\end{definition}

\begin{definition}\label{definition:refinement of paths}
	We say a path $q$ \emph{refines} a path $p$ if every step of $p$ contains some step of $q$: for every $n$ we can find $m$ such that $q_m\sqleq p_n$. In that case we write~$q \refines p$. 
\end{definition}
\noindent\begin{minipage}{0.68\textwidth}
	\begin{remark}
		Note that $\refines$ is just an instance of the \emph{upper order}: for any preordered set $(S,\leq)$ we get a preorder on the powerset $\Powerset(S)$ defined by $A\LeqU B$ iff $B\subseteq \up A$. See e.g.~\cite[\S 11.1]{vickers1989TopologyLogic}. That $q\refines p$ does not mean that every $q_m$ has to be contained in some $p_n$. For instance, any path $p$ can be refined by specifying further steps $p_\top \Leq A$ or~$B\Leq p_\bot$. Conversely, a given step~$p_n$ may contain multiple steps of $q$. See \cref{figure:refined path} for intuition. Interpreting a path $p$ as an approximation of a point-wise curve $\gamma$, the fact that $q\refines p$ is interpreted as $q$ being a more accurate approximation of~$\gamma$. Alternatively, we can think of~$q$ as a tighter `restriction' of where information can flow, a curve $\gamma$ being an idealised `infinitely tight' flow.
	\end{remark}
\end{minipage}%
\hfill%
\begin{minipage}{.44\textwidth}\centering
	\definecolor{cffebdc}{RGB}{255,235,220}
\definecolor{c640000}{RGB}{100,0,0}
\definecolor{cefffff}{RGB}{239,255,255}
\definecolor{c003c64}{RGB}{0,60,100}

\def \globalscale {1.000000}
\begin{tikzpicture}[y=.75pt, x=.75pt, yscale=\globalscale,xscale=\globalscale, every node/.append style={scale=\globalscale}, inner sep=0pt, outer sep=0pt]
  \begin{scope}[shift={(-271.5, 73.66)}]
    \begin{scope}[shift={(-130.49, -2.13)}]
      \path[draw=black,fill=cffebdc,line cap=butt,line join=miter,line 
  width=1.0pt,miter limit=4.0] (484.89, -45.99).. controls (479.94, -37.05) and 
  (456.85, -34.16) .. (450.87, -40.32).. controls (442.5, -48.96) and (445.99, 
  -62.14) .. (453.71, -65.84).. controls (464.82, -71.16) and (490.91, -56.88) 
  .. (484.89, -45.99) -- cycle;

      \path[draw=black,fill=cffebdc,line cap=butt,line join=miter,line 
  width=1.0pt,miter limit=4.0] (472.02, -0.99).. controls (466.25, 5.15) and 
  (449.43, -1.17) .. (448.2, -9.52).. controls (447.27, -15.8) and (457.39, 
  -22.5) .. (463.52, -20.84).. controls (470.47, -18.95) and (476.95, -6.24) .. 
  (472.02, -0.99) -- cycle;

      \path[draw=black,fill=cffebdc,line cap=butt,line join=miter,line 
  width=1.0pt,miter limit=4.0] (479.22, 30.54).. controls (479.22, 39.05) and 
  (469.31, 44.41) .. (460.32, 42.94).. controls (451.82, 41.55) and (442.37, 
  37.52) .. (442.68, 25.79).. controls (442.8, 21.17) and (450.87, 13.53) .. 
  (460.7, 19.74).. controls (468.6, 24.73) and (479.22, 16.37) .. (479.22, 
  30.54) -- cycle;

    \end{scope}
    \begin{scope}[shift={(-18.34, -10.22)}]
      \node[text=c640000,line cap=butt,line join=miter,line width=1.0pt,miter 
  limit=4.0,anchor=south west] (text11-9-5) at (380.45, -48.74){$p_1$};

      \node[text=c640000,line cap=butt,line join=miter,line width=1.0pt,miter 
  limit=4.0,anchor=south west] (text11-9-5-1) at (380.36, -3.25){$p_2$};

      \node[text=c640000,line cap=butt,line join=miter,line width=1.0pt,miter 
  limit=4.0,anchor=south west] (text18) at (380.38, 36.07){$p_3$};

    \end{scope}
    \path[draw=black,fill=cefffff,line cap=butt,line join=miter,line 
  width=1.0pt,miter limit=4.0] (326.05, -28.29).. controls (324.5, -30.59) and 
  (326.42, -35.53) .. (329.13, -35.85).. controls (332.21, -36.21) and (335.97, 
  -31.05) .. (334.55, -28.29).. controls (333.26, -25.76) and (327.6, -25.98) ..
   (326.05, -28.29) -- cycle;

    \path[draw=black,fill=cefffff,line cap=butt,line join=miter,line 
  width=1.0pt,miter limit=4.0] (326.42, -43.97).. controls (324.64, -46.87) and 
  (328.69, -52.33) .. (332.09, -52.47).. controls (335.13, -52.6) and (339.28, 
  -49.46) .. (337.76, -45.34).. controls (336.78, -42.68) and (328.42, -40.73) 
  .. (326.42, -43.97) -- cycle;

    \path[draw=black,fill=cefffff,line cap=butt,line join=miter,line 
  width=1.0pt,miter limit=4.0] (334.55, -56.63).. controls (331.67, -57.06) and 
  (328.71, -62.01) .. (330.45, -64.34).. controls (332.45, -67.0) and (341.74, 
  -66.41) .. (340.22, -62.3).. controls (339.24, -59.64) and (337.2, -56.24) .. 
  (334.55, -56.63) -- cycle;

    \path[draw=black,fill=cefffff,line cap=butt,line join=miter,line 
  width=1.0pt,miter limit=4.0] (326.98, -8.64).. controls (325.46, -11.05) and 
  (326.41, -15.56) .. (328.88, -16.95).. controls (331.49, -18.41) and (335.76, 
  -16.89) .. (337.39, -14.11).. controls (339.01, -11.33) and (337.55, -8.07) ..
   (335.37, -6.76).. controls (332.91, -5.28) and (328.5, -6.22) .. (326.98, 
  -8.64) -- cycle;

    \path[draw=black,fill=cefffff,line cap=butt,line join=miter,line 
  width=1.0pt,miter limit=4.0] (326.16, 60.54).. controls (325.33, 58.38) and 
  (328.01, 55.68) .. (330.18, 54.88).. controls (332.51, 54.03) and (335.76, 
  53.97) .. (337.39, 56.75).. controls (339.01, 59.53) and (336.46, 61.52) .. 
  (334.55, 62.42).. controls (331.96, 63.65) and (327.18, 63.21) .. (326.16, 
  60.54) -- cycle;

    \path[draw=black,fill=cefffff,line cap=butt,line join=miter,line 
  width=1.0pt,miter limit=4.0] (330.45, 33.23).. controls (328.41, 30.34) and 
  (326.18, 25.01) .. (328.88, 22.74).. controls (331.17, 20.82) and (335.76, 
  22.79) .. (337.39, 25.57).. controls (339.01, 28.35) and (340.13, 35.26) .. 
  (337.39, 36.91).. controls (334.65, 38.56) and (331.96, 35.37) .. (330.45, 
  33.23) -- cycle;

    \begin{scope}[shift={(-0.52, -2.78)}]
      \node[text=c003c64,line cap=butt,line join=miter,line width=1.0pt,miter 
  limit=4.0,anchor=south west] (text46) at (297.44, -62.34){$q_1$};

      \node[text=c003c64,line cap=butt,line join=miter,line width=1.0pt,miter 
  limit=4.0,anchor=south west] (text46-7) at (297.44, -48.43){$q_2$};

      \node[text=c003c64,line cap=butt,line join=miter,line width=1.0pt,miter 
  limit=4.0,anchor=south west] (text46-7-1) at (297.44, -32.57){$q_3$};

      \node[text=c003c64,line cap=butt,line join=miter,line width=1.0pt,miter 
  limit=4.0,anchor=south west] (text46-7-1-4) at (297.44, -11.89){$q_4$};

      \node[text=c003c64,line cap=butt,line join=miter,line width=1.0pt,miter 
  limit=4.0,anchor=south west] (text46-7-1-8) at (297.44, 27.03){$q_5$};

      \node[text=c003c64,line cap=butt,line join=miter,line width=1.0pt,miter 
  limit=4.0,anchor=south west] (text46-7-1-7) at (297.44, 57.13){$q_6$};

    \end{scope}
  \end{scope}

\end{tikzpicture}
	
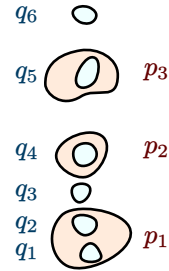
\captionof{figure}{Typical example of {$q\refines p$}.}
	\label{figure:refined path}
\end{minipage}

\begin{lemma}\label{lemma:properties of refinement relation}
	The refinement relation $\refines$ is a preorder and respects concatenation:
	\begin{enumerate}[label={(\alph*)},font=\normalfont]
		\item $p\refines p$;
		\item $r\refines q\refines p$ implies $r\refines p$;
		\item if $q\refines p$ and $q'\refines p'$, then $(q'\cdot q)\refines (p'\cdot p)$. 
	\end{enumerate}
\end{lemma}
\begin{proof}
	The first two statements follow since $\refines$ is the upper order induced by the inclusion relation $\sqleq$. For (c), the steps of $p'\cdot p$ corresponding to $p$ are refined by~$q$, and those corresponding to $p'$ are refined by $q'$.
\end{proof}

The following construction shows how to canonically refine a path $p$, given some open subregion in its endpoint.

\begin{construction}\label{construction:path restriction}
	Consider a path $p=(p_n)_{n=0}^N$ in a parallel ordered locale ${(X,\Leq)}$, and take any open region $W\sqleq p_N$ in the endpoint. We are going to define a new path: $p|_W$, called the \emph{restriction of $p$ to $W$,} that refines $p$ and has endpoint $W$. Each step of $p|_W$ shall be recursively defined, based on the following procedure using~\eqref{axiom:wedge-}:
	\[
	\begin{tikzcd}[column sep=-.1cm, row sep=-.15cm]
		\cdots & \Leq & {p_{N-2}\wedge\Down(p_{N-1}\wedge\Down W)} & \Leq & {p_{N-1}\wedge\Down W} & \Leq & W \\
		&& \rotatebox[origin=c]{-90}{$\sqleq$} && \rotatebox[origin=c]{-90}{$\sqleq$} && \rotatebox[origin=c]{-90}{$\sqleq$} \\
		\cdots & \Leq & {p_{N-2}} & \Leq & {p_{N-1}} & \Leq & {p_N}.
	\end{tikzcd}
	\]
	More formally, define the endpoint by $(p|_W)_N:= W$, and then use the recursive formula
	\[
	(p|_W)_n := p_n\wedge \Down(p|_W)_{n+1}
	\]
	for $0\leq n < N$. We get $(p|_W)_n \Leq (p|_W)_{n+1}$ by~\eqref{axiom:wedge-} (via \cref{lemma:wedge iff strong wedge}). Since the locale is parallel ordered, it follows that all steps of $p|_W$ are non-empty if and only if $W$ itself is non-empty. Lastly, $p|_W\refines p$ since $(p|_W)_n\sqleq p_n$ by construction. Note that if $W$ is the entire endpoint we get $p|_{p_\top}=p$.
	
	Using axiom~\eqref{axiom:wedge+} we get an analogous construction, where for $V\sqleq p_0$ we get a refinement of $p$ with starting point $V$ that we denote by $p|^V$. 
\end{construction}
\begin{figure}[b]\centering
	\begin{subfigure}[b]{0.55\textwidth}\centering
		\definecolor{cffebdc}{RGB}{255,235,220}
\definecolor{cededed}{RGB}{237,237,237}
\definecolor{cffffdc}{RGB}{255,255,220}
\definecolor{c644700}{RGB}{100,71,0}
\definecolor{c640000}{RGB}{100,0,0}
\definecolor{cefffff}{RGB}{239,255,255}
\definecolor{c001764}{RGB}{0,23,100}

\def \globalscale {1.000000}
\begin{tikzpicture}[y=.75pt, x=.75pt, yscale=\globalscale,xscale=\globalscale, every node/.append style={scale=\globalscale}, inner sep=0pt, outer sep=0pt]
  \begin{scope}[shift={(-271.503, 73.658)}]
    \path[fill=cffebdc,line cap=butt,line join=miter,line width=1.0pt,miter 
  limit=4.0] (334.266, 34.399) -- (350.74, 34.399) -- (379.825, -45.51) -- 
  (305.181, -45.51) -- cycle;

    \path[fill=cededed,line cap=butt,line join=miter,line width=1.001pt,miter 
  limit=4.0] (337.654, -0.206) -- (356.473, -1.019) -- (372.666, -45.51) -- 
  (321.18, -45.468) -- (337.654, -0.206);

    \path[draw=black,line cap=butt,line join=miter,line width=0.5pt,miter 
  limit=4.0,dash pattern=on 0.5pt off 2.0pt] (358.247, 32.742) -- (386.728, 
  -45.51);

    \path[draw=black,line cap=butt,line join=miter,line width=0.5pt,miter 
  limit=4.0,dash pattern=on 0.5pt off 2.0pt] (337.654, -0.206) -- (321.18, 
  -45.51);

    \path[draw=black,line cap=butt,line join=miter,line width=0.5pt,miter 
  limit=4.0,dash pattern=on 0.5pt off 2.0pt] (356.473, -1.019) -- (372.666, 
  -45.51);

    \path[draw=black,line cap=butt,line join=miter,line width=0.5pt,miter 
  limit=4.0,dash pattern=on 0.5pt off 2.0pt] (313.703, 31.445) -- (285.694, 
  -45.51);

    \path[draw=black,fill=cffffdc,line cap=butt,line join=miter,line 
  width=1.0pt,miter limit=4.0] (320.988, 22.939).. controls (326.589, 23.412) 
  and (332.501, 25.948) .. (338.344, 27.042).. controls (344.6, 28.215) and 
  (354.252, 24.033) .. (357.256, 29.674).. controls (359.454, 33.799) and 
  (355.434, 40.406) .. (351.185, 42.314).. controls (339.312, 47.646) and 
  (314.722, 37.53) .. (313.703, 31.445).. controls (312.675, 25.305) and 
  (316.265, 22.541) .. (320.988, 22.939) -- cycle;

    \path[draw=black,fill=cffebdc,line cap=butt,line join=miter,line 
  width=1.0pt,miter limit=4.0] (345.007, 40.155).. controls (340.99, 40.753) and
   (334.063, 38.456) .. (334.266, 34.399).. controls (334.446, 30.808) and 
  (340.654, 29.971) .. (344.237, 30.281).. controls (346.793, 30.502) and 
  (350.455, 31.849) .. (350.74, 34.399).. controls (351.041, 37.091) and 
  (347.686, 39.757) .. (345.007, 40.155) -- cycle;

    \path[draw=black,fill=cffffdc,line cap=butt,line join=miter,line 
  width=1.0pt,miter limit=4.0] (300.587, -29.036).. controls (304.706, -19.782) 
  and (325.327, -25.541) .. (337.654, -26.688).. controls (343.982, -27.277) and
   (355.709, -25.043) .. (356.132, -31.384).. controls (356.441, -36.031) and 
  (347.73, -35.66) .. (343.167, -36.592).. controls (332.625, -38.746) and 
  (319.116, -38.117) .. (310.896, -37.273).. controls (302.261, -36.386) and 
  (298.467, -33.799) .. (300.587, -29.036) -- cycle;

    \path[draw=black,fill=cffebdc,line cap=butt,line join=miter,line 
  width=1.0pt,miter limit=4.0] (328.464, -25.453) -- (323.888, -38.025).. 
  controls (330.271, -38.125) and (337.186, -37.814) .. (343.166, -36.592).. 
  controls (347.729, -35.66) and (356.442, -36.031) .. (356.132, -31.385).. 
  controls (355.709, -25.044) and (343.982, -27.276) .. (337.654, -26.687).. 
  controls (334.918, -26.433) and (331.746, -25.947) .. (328.464, -25.453) -- 
  cycle;

    \path[draw=black,fill=cffffdc,line cap=butt,line join=miter,line 
  width=1.0pt,miter limit=4.0] (314.032, 4.321).. controls (311.034, 4.977) and 
  (305.929, 3.269) .. (305.795, 0.203).. controls (305.681, -2.418) and 
  (309.883, -3.757) .. (312.502, -3.916).. controls (314.681, -4.048) and 
  (317.953, -3.313) .. (318.456, -1.189).. controls (318.998, 1.103) and 
  (316.333, 3.818) .. (314.032, 4.321) -- cycle;

    \path[draw=black,fill=cffebdc,line cap=butt,line join=miter,line 
  width=1.0pt,miter limit=4.0] (337.172, -2.233).. controls (336.865, -5.788) 
  and (342.345, -8.051) .. (345.891, -8.443).. controls (349.836, -8.879) and 
  (355.986, -7.692) .. (356.865, -3.821).. controls (357.629, -0.462) and 
  (353.386, 3.23) .. (350.01, 3.912).. controls (345.36, 4.853) and (337.581, 
  2.493) .. (337.172, -2.233) -- cycle;

    \node[text=c644700,line cap=butt,line join=miter,line width=1.0pt,miter 
  limit=4.0,anchor=south west] (text7) at (286.498, -2.173){$p_1$};

    \node[text=c640000,line cap=butt,line join=miter,line width=1.0pt,miter 
  limit=4.0,anchor=south west] (text8) at (378.505, -6.104){$p_1\wedge \Down W$};

    \node[text=c640000,line cap=butt,line join=miter,line width=1.0pt,miter 
  limit=4.0,anchor=south west] (text9) at (342.543, 47.097){$W$};

    \node[text=c640000,line cap=butt,line join=miter,line width=1.0pt,miter 
  limit=4.0,anchor=south west] (text12) at (310.564, -62.396){$p_0\wedge 
  \Down(p_1\wedge \Down W)$};

    \node[text=c644700,line cap=butt,line join=miter,line width=1.0pt,miter 
  limit=4.0,anchor=south west] (text13) at (275.876, -33.154){$p_0$};

    \node[text=c644700,line cap=butt,line join=miter,line width=1.0pt,miter 
  limit=4.0,anchor=south west] (text14) at (291.193, 30.275){$p_\top$};

    \path[draw=black,line cap=butt,line join=miter,line width=0.5pt,miter 
  limit=4.0,dash pattern=on 0.5pt off 2.0pt] (334.266, 34.399) -- (305.182, 
  -45.51);

    \path[draw=black,line cap=butt,line join=miter,line width=0.5pt,miter 
  limit=4.0,dash pattern=on 0.5pt off 2.0pt] (350.74, 34.399) -- (379.825, 
  -45.51);

  \end{scope}
%
%
%
%
%
%
%
%
%
%
%
%
%
%
%
%
%
%
%
%
%

\end{tikzpicture}
		\caption{Path restriction in Minkowski-type space.}
		\label{figure:path restriction Minkowski}
	\end{subfigure}\hfil
	\begin{subfigure}[b]{0.4\textwidth}\centering
		\definecolor{c644700}{RGB}{100,71,0}
\definecolor{c640000}{RGB}{100,0,0}
\definecolor{cededed}{RGB}{237,237,237}
\definecolor{cffffdc}{RGB}{255,255,220}
\definecolor{cffebdc}{RGB}{255,235,220}
\definecolor{cefffff}{RGB}{239,255,255}
\definecolor{c001764}{RGB}{0,23,100}

\def \globalscale {1.000000}
\begin{tikzpicture}[y=.75pt, x=.75pt, yscale=\globalscale,xscale=\globalscale, every node/.append style={scale=\globalscale}, inner sep=0pt, outer sep=0pt]
  \begin{scope}[shift={(-271.503, 73.658)}]
    \begin{scope}[line width=1.493pt,cm={ 1.899,-0.0,-0.0,1.899,(-301.836, -4.234)}]
      \node[text=c644700,line cap=butt,line join=miter,line width=0.527pt,miter 
  limit=4.0,anchor=south west] (text7) at (310.185, 6.89){$p_1$};

      \node[text=c640000,line cap=butt,line join=miter,line width=0.527pt,miter 
  limit=4.0,anchor=south west] (text9) at (353.399, 24.159){$W$};

      \node[text=c644700,line cap=butt,line join=miter,line width=0.527pt,miter 
  limit=4.0,anchor=south west] (text13) at (310.087, -22.496){$p_0$};

      \node[text=c644700,line cap=butt,line join=miter,line width=0.527pt,miter 
  limit=4.0,anchor=south west] (text14) at (310.185, 24.181){$p_\top$};

    \end{scope}
    \path[fill=cededed,line cap=butt,line join=miter,line width=1.0pt,miter 
  limit=4.0] (335.783, 46.479) -- (360.764, 47.111) -- (360.764, -56.632) -- 
  (335.783, -56.632) -- cycle;

    \path[draw=black,fill=cffffdc,line cap=butt,line join=miter,line 
  width=1.0pt,miter limit=4.0] (317.545, 39.746).. controls (322.183, 39.746) 
  and (343.159, 39.313) .. (354.395, 36.911).. controls (360.832, 35.535) and 
  (365.967, 38.459) .. (365.967, 45.63).. controls (365.967, 52.802) and 
  (358.95, 52.59) .. (354.395, 53.919).. controls (346.108, 56.337) and 
  (337.054, 53.928) .. (328.536, 52.528).. controls (321.13, 51.311) and (307.1,
   54.57) .. (306.85, 46.479).. controls (306.864, 39.419) and (313.827, 39.746)
   .. (317.545, 39.746) -- cycle;

    \path[draw=black,line cap=butt,line join=miter,line width=0.5pt,miter 
  limit=4.0,dash pattern=on 0.5pt off 2.0pt] (306.85, 46.479) -- (306.85, 
  -56.632);

    \path[draw=black,line cap=butt,line join=miter,line width=0.5pt,miter 
  limit=4.0,dash pattern=on 0.5pt off 2.0pt] (366.294, 47.464) -- (366.294, 
  -56.632);

    \path[draw=black,fill=cffffdc,line cap=butt,line join=miter,line 
  width=1.0pt,miter limit=4.0] (310.883, 19.664).. controls (314.93, 22.043) and
   (318.798, 20.473) .. (324.961, 19.211).. controls (331.124, 17.949) and 
  (346.065, 15.095) .. (350.939, 6.833).. controls (353.006, 3.329) and (348.05,
   -1.898) .. (350.203, -5.35).. controls (353.232, -10.205) and (366.937, 
  -6.301) .. (366.049, -11.954).. controls (365.429, -15.896) and (357.934, 
  -14.113) .. (354.118, -12.946).. controls (348.359, -11.185) and (336.163, 
  -10.691) .. (340.033, -1.633).. controls (342.606, 3.226) and (345.85, 5.239) 
  .. (339.788, 7.569).. controls (333.726, 9.898) and (332.746, 11.792) .. 
  (325.336, 9.926).. controls (319.494, 7.578) and (310.524, 4.068) .. (307.177,
   12.471).. controls (306.455, 15.07) and (308.558, 18.297) .. (310.883, 
  19.664) -- cycle;

    \path[draw=black,fill=cffebdc,line cap=butt,line join=miter,line 
  width=1.0pt,miter limit=4.0] (335.782, 16.549) -- (335.782, 9.281).. controls 
  (336.902, 8.782) and (338.156, 8.196) .. (339.788, 7.569).. controls (345.849,
   5.24) and (342.606, 3.226) .. (340.033, -1.633).. controls (336.162, -10.69) 
  and (348.36, -11.185) .. (354.119, -12.946).. controls (355.894, -13.489) and 
  (358.46, -14.157) .. (360.764, -14.329) -- (360.764, -8.47).. controls 
  (356.83, -8.042) and (351.881, -8.04) .. (350.203, -5.351).. controls (348.05,
   -1.899) and (353.006, 3.33) .. (350.938, 6.834).. controls (348.081, 11.678) 
  and (341.764, 14.654) .. (335.782, 16.549) -- cycle;

    \path[draw=black,fill=cffffdc,line cap=butt,line join=miter,line 
  width=1.0pt,miter limit=4.0] (306.85, -47.428).. controls (305.352, -45.214) 
  and (307.969, -41.621) .. (310.193, -40.139).. controls (314.271, -37.42) and 
  (320.283, -41.305) .. (324.887, -39.624).. controls (329.978, -37.766) and 
  (332.935, -32.211) .. (337.679, -29.592).. controls (343.686, -26.276) and 
  (350.299, -23.827) .. (357.059, -22.653).. controls (359.807, -22.176) and 
  (364.771, -19.942) .. (365.427, -22.653).. controls (366.96, -28.985) and 
  (354.957, -29.995) .. (348.047, -31.593).. controls (342.331, -32.914) and 
  (331.55, -36.498) .. (332.909, -42.585).. controls (334.147, -48.129) and 
  (343.99, -46.277) .. (349.649, -45.782).. controls (354.914, -45.321) and 
  (360.741, -36.404) .. (364.469, -40.15).. controls (366.893, -42.584) and 
  (363.207, -47.352) .. (360.764, -49.767).. controls (357.998, -52.502) and 
  (353.712, -53.631) .. (349.828, -53.84).. controls (342.739, -54.221) and 
  (336.129, -50.079) .. (329.188, -48.588).. controls (325.871, -47.876) and 
  (322.55, -47.041) .. (319.165, -46.828).. controls (315.063, -46.571) and 
  (309.152, -50.832) .. (306.85, -47.428) -- cycle;

    \path[draw=black,fill=cffebdc,line cap=butt,line join=miter,line 
  width=1.0pt,miter limit=4.0] (360.764, -21.725).. controls (359.447, -22.045) 
  and (358.095, -22.473) .. (357.058, -22.653).. controls (350.298, -23.827) and
   (343.685, -26.276) .. (337.679, -29.592).. controls (337.014, -29.959) and 
  (336.388, -30.388) .. (335.783, -30.852) -- (335.783, -36.644).. controls 
  (339.129, -34.072) and (344.537, -32.404) .. (348.047, -31.593).. controls 
  (351.783, -30.73) and (356.995, -30.03) .. (360.764, -28.576) -- (360.764, 
  -21.725) -- cycle(360.764, -39.572).. controls (357.411, -40.864) and 
  (353.389, -45.455) .. (349.649, -45.782).. controls (345.608, -46.135) and 
  (339.434, -47.18) .. (335.783, -45.729) -- (335.783, -50.605).. controls 
  (340.413, -52.272) and (345.01, -54.099) .. (349.828, -53.84).. controls 
  (353.712, -53.631) and (357.999, -52.503) .. (360.764, -49.768) -- (360.764, 
  -39.572) -- cycle;

    \path[draw=black,fill=cffebdc,line cap=butt,line join=miter,line 
  width=1.0pt,miter limit=4.0] (336.408, 47.464).. controls (335.349, 45.326) 
  and (338.366, 42.447) .. (340.625, 41.679).. controls (347.096, 39.481) and 
  (361.179, 40.688) .. (360.294, 47.464).. controls (359.834, 50.984) and 
  (353.198, 47.734) .. (349.649, 47.743).. controls (345.234, 47.755) and 
  (338.368, 51.42) .. (336.408, 47.464) -- cycle;

    \path[draw=black,line cap=butt,line join=miter,line width=0.5pt,miter 
  limit=4.0,dash pattern=on 0.5pt off 2.0pt] (335.783, 46.479) -- (335.783, 
  -56.632);

    \path[draw=black,line cap=butt,line join=miter,line width=0.5pt,miter 
  limit=4.0,dash pattern=on 0.5pt off 2.0pt] (360.764, 47.111) -- (360.764, 
  -56.632);

  \end{scope}
%
%
%
%
%
%
%
%
%
%
%
%
%
%
%
%
%
%
%
%
%

\end{tikzpicture}
		\caption{Path restriction in vertical-$\mathbb{R}^2$.}
	\end{subfigure}
	\caption{Examples of \cref{construction:path restriction}.}
	\label{figure:path restriction}
\end{figure}

\begin{remark}
	Note that the construction only works on paths defined in terms of finite sequences, at least without resorting to a higher form of induction.
	
	It is tempting to define $p|_W$ by the simpler formula $(p|_W)_n = p_n\wedge \Down W$. Unfortunately, this does not always define a path. As \cref{figure:path restriction}(a) shows, it can happen that intersecting a path step-wise with $\Down W$ does not preserve the causal order~$\Leq$. Here $p_0\Leq p_1\Leq p_\top\sqgeq W$, where $p_1$ is the disjoint union of the two regions in the middle, but $p_0\not\Leq p_1\wedge \Down W$ since $p_0\not\sqleq \Down (p_1\wedge \Down W)$.
	
	Similarly, it might seem desirable to consider only paths $p$ whose steps $p_n$ are \emph{connected} (e.g.~as done in the unordered setting of \cite{kennison1989WhatFundamentalGroup}). However, we can see in \cref{figure:path restriction}(b) that \cref{construction:path restriction} does not preserve this property: $p_0$ is connected but $(p|_W)_0$ is not. A more suitable option is to only consider paths consisting of \emph{convex} regions \cite[\S 6.1]{schaaf2024TowardsPointFreeSpacetimes}, which \cref{construction:path restriction} does preserve, but for our purposes here we do not find it necessary to do so.
\end{remark}

We obtain some elementary properties of path restriction.

\begin{lemma}\label{lemma:restriction of paths is functorial}
	Let $p$ be a path, with $W\sqleq V\sqleq p_\top$. Then $p|_W = (p|_V)|_W$.
\end{lemma}
\begin{proof}
	By construction, we get $((p|_V)|_W)_n\sqleq (p|_V)_n$. Using this, we find that
	\[
	((p|_V)|_W)_n 
	=
	p_n \wedge \Down (p|_V)_{n+1} \wedge \Down ((p|_V)|_W)_{n+1}
	=
	p_n \wedge \Down ((p|_V)|_W)_{n+1}.
	\]
	From this equation we see that the result will follow by induction as soon as we can show $((p|_V)|_W)_N = (p|_W)_N$, where $N$ is the top index of $p$. But this equation is just the statement that $W=W$, so we are done.
\end{proof}

The following result shows that the path restriction operation loses no information about a path, as long as the entire endpoint is covered.
\begin{lemma}\label{lemma:join over restrictions}
	Suppose that $p=(p_n)_{n=0}^N$ is a path, and take an open cover $(W_i)_{i\in I}$ of the endpoint $p_N$. Then for all $0\leq n\leq N$:
	\[
	p_n = \bigvee_{i\in I} \left( p|_{W_i} \right)_n.
	\]
\end{lemma}
\begin{proof}
	The statement holds for $n=N$ since $(p|_{W_i})_N=W_i$. We prove the remaining cases via induction. Suppose the equation holds for the index $n+1$. We then calculate:
	\[
	\bigvee_{i\in I}(p|_{W_i})_{n} 
	=
	p_{n}\wedge\bigvee_{i\in I}\Down(p|_{W_i})_{n+1}
	=
	p_{n}\wedge \Down \bigvee_{i\in I}(p|_{W_i})_{n+1}
	= p_{n}\wedge \Down p_{n+1}
	= p_{n}.
	\]
	The first step is infinite distributivity applied to the definition of $(p|_{W_i})_n$, the second is due to~\eqref{axiom:LV}, the third is the induction hypothesis, and the final equality follows by ${p_{n}\Leq p_{n+1}}$.
\end{proof}

Under certain conditions, the path restriction operation preserves refinements.
\begin{lemma}\label{lemma:path restriction preserves refinement}
	Let $q\refines p$ with the property that
	\begin{quote}
		if $q_m\sqleq p_n\neq p_\top$, then $\exists k\geq m: q_k\sqleq p_{n+1}$,
	\end{quote}
	and take non-empty $W\sqleq q_\top \sqleq p_\top$. Then $q|_W\refines p|_{W}$.
\end{lemma}
The condition says that for every step $q_m$ inhabiting $p_n$, there exists some $q_k$ in the future that inhabits the next step $p_{n+1}$.
\begin{proof}
	Write explicitly $p=(p_n)_{n=0}^N$ and $q=(q_m)_{m=0}^M$. We need to prove that for every $n$ there exists an $m$ such that $(q|_W)_m\sqleq (p|_W)_n$. For the top indices we get $(q|_W)_M = W = (p|_W)_N$. Similarly, since $q\refines p$ we get $q_m\sqleq p_{N-1}$ for some~$m$, and since $(q|_W)_m\Leq W$ it follows that ${(q|_W)_m\sqleq p_{N-1}\wedge \Down W = (p|_W)_{N-1}}$. In fact, this inclusion holds for every $m$ such that $q_m$ is contained in $p_{N-1}$.
	
	Pose as an induction hypothesis (case $n+1$) that for all $m$ with $q_m\sqleq p_{n+1}$ we have $(q|_W)_m \sqleq (p|_W)_{n+1}$. We claim the same holds for $n$. If $q_m\sqleq p_n$ then by assumption there exists $k\geq m$ such that ${q_m\Leq q_k\sqleq p_{n+1}}$. We thus get
	\begin{align*}
		(q|_W)_m
		&:=
		q_m \wedge \Down (q|_W)_{m+1}
		\\&\sqleq 
		p_n \wedge \Down (q|_W)_k &&\text{($q_m\sqleq p_n$ and $m\leq k$)}
		\\&\sqleq 
		p_n \wedge \Down (p|_W)_{n+1}, &&\text{(induction hypothesis)}
	\end{align*}
	so it follows that $(q|_W)_m\sqleq (p|_W)_n$. The fact that $q|_W\refines p|_{W}$ now follows directly from~${q\refines p}$.
\end{proof}

\section{Causal coverages in ordered locales}
\label{section:causal coverages ordered locales}
We are now ready to define the point-free analogue of `causal coverage' discussed in the~\nameref{section:introduction}. First we need to refine our definition of refinements.
\begin{definition}\label{definition:locally refines}
	A \emph{local past refinement} of a path $p$ consists of a family $(q^j)_{j\in J}$ of paths such that:
	\begin{enumerate}
		\item $\bigvee_{j\in J} q^j_\top = p_\top$;
		\item $q^j\refines p|_{q^j_\top}$ for all $j\in J$.
	\end{enumerate}
	In that case we write $(q^j)_{j\in J}\refines p$.
	
	A \emph{(global) past refinement} is a local past refinement where $J$ contains one element. Explicitly, for a path $p$ this consists of a refinement $q\refines p$ such that $q_\top = p_\top$. 
\end{definition}

\begin{remark}
	The definition of a past refinement specifically excludes refinements $q$ that propagate beyond the original path $p$, which is allowed for arbitrary refinements (see \cref{figure:refined path}). Moreover, the fact that $\bigvee q^j_\top = p_\top$ guarantees that all information at the endpoint $p_\top$ is reached by the refinements~$q^j$. We use this to define past causal coverages. For future causal coverages, one defines (local) future refinements dually.
\end{remark}

\begin{definition}\label{definition:causal coverage locale}\label{definition:causal coverage from ordered locale}
	Let $(X,\Leq)$ be a parallel ordered locale with~\eqref{axiom:LV}, and pick opens $U,A\in\Opens X$. We say that $A$ \emph{covers $U$ from below} if $A\sqleq \Down U$, and every path landing in $U$ can be \emph{locally refined} to inhabit $A$:
	\begin{quote}
		for every path $p$ landing in $U$ there exists a local past refinement $(q^j)\refines p$ such that each $q^j$ inhabits $A$.
	\end{quote}
	
	An analogous definition can be made for regions that cover $U$ from \emph{above}. Denote by $\CovLeq^\pm(U)\subseteq \Opens X$ the set of all opens that cover $U$ from above/below. We call $\CovLeq^\pm$ the \emph{(future/past) causal coverage} induced by $\Leq$.
\end{definition}

\noindent\begin{minipage}{0.65\textwidth}
	\begin{remark}
		Very explicitly, if $A\in \CovLeq^-(U)$ and $p$ is a path landing in $U$, this means we can find some family~$(q^j)_{j\in J}$ of paths, together with an open cover $(W_j)_{j\in J}$ of the endpoint $p_\top$, such that $q^j_\top = W_j$ and $q^j$ refines the path $p|_{W_j}$, and such that for every $j\in J$ there is some index $a_j$ such that $q^j_{a_j}\sqleq A$.
	\end{remark}
	\begin{example}
		In \cref{figure:non-globally refinable} we illustrate an example of a path in vertical-$\mathbb{R}^2$ that is not globally past refinable to inhabit $A$, but only locally. Nevertheless, we want to think of $A$ as covering $U$ from below in this picture. This shows that local refinements are necessary for a proper definition of $\CovLeq^\pm$. 
	\end{example}	
\end{minipage}%
\hfill%
\begin{minipage}{.45\textwidth}\centering
	\definecolor{cededed}{RGB}{237,237,237}
\definecolor{cffffdc}{RGB}{255,255,220}
\definecolor{c644700}{RGB}{100,71,0}
\definecolor{cffebdc}{RGB}{255,235,220}
\definecolor{ce0ffdc}{RGB}{224,255,220}
\definecolor{c0d6400}{RGB}{13,100,0}
\definecolor{cefffff}{RGB}{239,255,255}
\definecolor{c001764}{RGB}{0,23,100}
\definecolor{c640000}{RGB}{100,0,0}

\def \globalscale {1.000000}
\begin{tikzpicture}[y=.75pt, x=1.1pt, yscale=\globalscale,xscale=\globalscale, every node/.append style={scale=\globalscale}, inner sep=0pt, outer sep=0pt]
\begin{scope}[shift={(-271.503, 73.658)}]
\path[fill=cededed,line cap=butt,line join=miter,line width=1.0pt,miter 
limit=4.0,shift={(12.427, -0.0)}] (320.379, 36.911) -- (337.387, 36.911) -- 
(338.166, -76.475) -- (320.629, -76.475) -- cycle;

\path[draw=black,fill=cffffdc,line cap=butt,line join=miter,line 
width=1.0pt,miter limit=4.0] (334.552, 45.415).. controls (300.537, 45.415) 
and (300.537, 45.415) .. (300.537, 36.911).. controls (300.537, 28.407) and 
(300.537, 28.407) .. (334.552, 28.407).. controls (374.237, 28.407) and 
(374.237, 28.407) .. (374.237, 36.911).. controls (374.237, 45.415) and 
(374.237, 45.415) .. (334.552, 45.415) -- cycle;

\path[draw=black,line cap=butt,line join=miter,line width=0.5pt,miter 
limit=4.0,dash pattern=on 0.5pt off 2.0pt] (300.537, 36.911) -- (300.537, 
-76.475);

\path[draw=black,line cap=butt,line join=miter,line width=0.5pt,miter 
limit=4.0,dash pattern=on 0.5pt off 2.0pt] (373.987, 36.911) -- (374.238, 
-76.475);

\node[text=c644700,line cap=butt,line join=miter,line width=1.0pt,miter 
limit=4.0,anchor=south west] (text5) at (361.4, 32.742){$U$};

\begin{scope}[shift={(12.427, -0.0)},blend group=multiply]
\path[draw=black,fill=cffebdc,line cap=butt,line join=miter,line 
width=1.0pt,miter limit=4.0] (331.85, -59.467).. controls (328.095, -58.427) 
and (321.631, -58.57) .. (320.511, -62.301).. controls (319.46, -65.803) and 
(324.32, -69.785) .. (327.927, -70.382).. controls (331.522, -70.976) and 
(337.002, -68.743) .. (337.519, -65.136).. controls (337.899, -62.491) and 
(334.425, -60.18) .. (331.85, -59.467) -- cycle;

\path[draw=black,fill=cffebdc,line cap=butt,line join=miter,line 
width=1.0pt,miter limit=4.0] (329.983, -22.934).. controls (326.106, -23.2) 
and (319.88, -25.689) .. (320.379, -29.543).. controls (320.84, -33.1) and 
(326.988, -33.892) .. (330.504, -33.183).. controls (333.621, -32.554) and 
(337.969, -29.877) .. (337.529, -26.727).. controls (337.14, -23.939) and 
(332.792, -22.741) .. (329.983, -22.934) -- cycle;

\path[draw=black,fill=cffebdc,line cap=butt,line join=miter,line 
width=1.0pt,miter limit=4.0] (330.598, 7.688).. controls (326.85, 8.019) and 
(320.314, 6.66) .. (320.379, 2.899).. controls (320.463, -1.917) and (329.058,
-4.168) .. (333.669, -2.774).. controls (335.833, -2.12) and (337.81, 0.678) 
.. (337.387, 2.899).. controls (336.869, 5.619) and (333.357, 7.445) .. 
(330.598, 7.688) -- cycle;

\path[draw=black,fill=cffebdc,line cap=butt,line join=miter,line 
width=1.0pt,miter limit=4.0] (331.718, 41.439).. controls (327.797, 41.997) 
and (321.161, 40.848) .. (320.733, 36.911).. controls (320.375, 33.621) and 
(325.574, 31.275) .. (328.883, 31.242).. controls (332.29, 31.208) and 
(337.241, 33.507) .. (337.387, 36.911).. controls (337.491, 39.327) and 
(334.112, 41.098) .. (331.718, 41.439) -- cycle;

\path[draw=black,line cap=butt,line join=miter,line width=0.5pt,miter 
limit=4.0,dash pattern=on 0.5pt off 2.0pt] (320.379, 36.911) -- (320.629, 
-76.475);

\path[draw=black,line cap=butt,line join=miter,line width=0.5pt,miter 
limit=4.0,dash pattern=on 0.5pt off 2.0pt] (337.387, 36.911) -- (338.166, 
-76.475);

\path[draw=black,fill=ce0ffdc,line cap=butt,line join=miter,line 
width=1.0pt,miter limit=4.0] (302.283, 10.838).. controls (287.731, 10.838) 
and (288.11, 5.547) .. (288.11, 2.895).. controls (288.11, 0.244) and 
(292.546, -4.946) .. (302.283, -2.774).. controls (312.02, -0.602) and 
(327.265, -5.608) .. (327.543, 1.825).. controls (327.795, 8.565) and 
(309.321, 10.838) .. (302.283, 10.838) -- cycle;

\path[draw=black,fill=ce0ffdc,line cap=butt,line join=miter,line 
width=1.0pt,miter limit=4.0] (347.637, -25.768).. controls (333.085, -25.768) 
and (324.96, -28.603) .. (324.96, -34.272).. controls (324.96, -39.942) and 
(339.133, -39.942) .. (349.394, -38.433).. controls (358.263, -37.129) and 
(361.81, -34.272) .. (361.81, -31.438).. controls (361.81, -28.603) and 
(354.675, -25.768) .. (347.637, -25.768) -- cycle;

\end{scope}
\node[text=c0d6400,line cap=butt,line join=miter,line width=1.0pt,miter 
limit=4.0,anchor=south west] (text6) at (289.137, 0.061){$A$};

\node[text=c0d6400,line cap=butt,line join=miter,line width=1.0pt,miter 
limit=4.0,anchor=south west] (text6-4) at (377.245, -29.579){$A$};

\path[draw=black,fill=c0d6400,line cap=butt,line join=miter,line 
width=0.5pt,miter limit=4.0,dash pattern=on 4.0pt off 1.5pt] (338.96, -0.691) 
-- (338.96, -30.472);

\end{scope}
%
%
%
%
%
%
%
%
%
%
%
%
%
%
%
%
%
%
%
%
%
%
%
%
%
%
%

\end{tikzpicture}
	
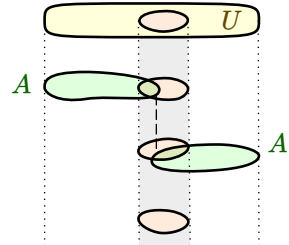
\captionof{figure}{Non-globally refinable path.}
	\label{figure:non-globally refinable}
\end{minipage}

\begin{example}\label{example:causal coverage for equality}
	We claim that in the ordered locale $(X,=)$ the causal coverage relation is equivalent to equality. Note first that paths here are just non-empty regions, and path refinement corresponds to the inclusion order. That $A\in \Cov_{=}^-(U)$ first implies that $A\sqleq \Down U = U$. Moreover, the constant path $p = (U)$ must admit a local refinement, that being just an open cover $(U_j)_{j\in J}$ of $U$, each of which inhabits~$A$, meaning $U_j\sqleq A$. Thus $U= \bigvee_{j\in J} U_j \sqleq A$, and so the equality $A=U$ follows.
\end{example}

In the following we collect some important properties of $\CovLeq^\pm$. See \cref{figure:properties of coverage from locale} for intuition. In proving these properties we have the fundamental axioms of a coverage (see \cref{section:introduction}) in mind.

\begin{figure}[b]\centering
	\begin{subfigure}[b]{0.4\textwidth}\centering
		\definecolor{cffffdc}{RGB}{255,255,220}
\definecolor{c644700}{RGB}{100,71,0}
\definecolor{ce0ffdc}{RGB}{224,255,220}
\definecolor{c0d6400}{RGB}{13,100,0}
\definecolor{cffebdc}{RGB}{255,235,220}
\definecolor{cefffff}{RGB}{239,255,255}
\definecolor{c001764}{RGB}{0,23,100}
\definecolor{c640000}{RGB}{100,0,0}

\def \globalscale {1.000000}
\begin{tikzpicture}[y=.75pt, x=.75pt, yscale=\globalscale,xscale=\globalscale, every node/.append style={scale=\globalscale}, inner sep=0pt, outer sep=0pt]
	\path[draw=black,line cap=butt,line join=miter,line width=0.5pt,miter 
	limit=4.0,dash pattern=on 0.5pt off 2.0pt] (38.771, 117.494) -- (5.299, 
	25.529);

	\path[draw=black,fill=cffffdc,line cap=butt,line join=miter,line 
	width=1.0pt,miter limit=4.0] (40.373, 119.073).. controls (37.142, 116.895) 
	and (35.14, 110.805) .. (37.538, 107.734).. controls (41.608, 102.521) and 
	(50.767, 107.612) .. (57.38, 107.734).. controls (70.986, 107.986) and 
	(93.148, 96.598) .. (98.177, 109.243).. controls (100.364, 114.744) and 
	(99.062, 124.426) .. (85.727, 121.907).. controls (80.079, 120.84) and 
	(76.972, 115.535) .. (71.819, 114.647).. controls (61.387, 112.849) and 
	(49.151, 124.988) .. (40.373, 119.073) -- cycle;

	\node[text=c644700,line cap=butt,line join=miter,line width=1.0pt,miter 
	limit=4.0,anchor=south west] (text12) at (40.539, 108){$U$};

	\path[draw=black,line cap=butt,line join=miter,line width=0.5pt,miter 
	limit=4.0,dash pattern=on 0.5pt off 2.0pt] (99.9, 113.403) -- (131.884, 
	25.529);

	\path[draw=black,fill=ce0ffdc,line cap=butt,line join=miter,line 
	width=1.0pt,miter limit=4.0] (48.876, 82.222).. controls (42.011, 83.391) and 
	(32.803, 93.567) .. (28.768, 87.892).. controls (26.764, 85.074) and (31.753, 
	81.19) .. (34.703, 79.388).. controls (52.606, 68.449) and (76.632, 75.639) ..
	(97.065, 70.884).. controls (102.886, 69.529) and (109.363, 61.536) .. 
	(114.073, 65.214).. controls (116.314, 66.964) and (115.205, 71.508) .. 
	(113.417, 73.718).. controls (109.222, 78.903) and (102.735, 79.388) .. 
	(94.231, 79.388).. controls (79.083, 79.388) and (63.809, 79.68) .. (48.876, 
	82.222) -- cycle;

	\path[draw=black,fill=ce0ffdc,line cap=butt,line join=miter,line 
	width=1.0pt,miter limit=4.0] (29.034, 45.372).. controls (23.552, 44.878) and 
	(14.705, 50.424) .. (12.521, 45.372).. controls (11.419, 42.824) and (15.631, 
	40.292) .. (18.226, 39.305).. controls (29.585, 34.981) and (42.402, 42.036) 
	.. (54.546, 42.537).. controls (75.041, 43.383) and (96.875, 34.954) .. 
	(116.084, 42.153).. controls (118.983, 43.239) and (124.146, 45.406) .. 
	(123.046, 48.3).. controls (119.633, 57.278) and (103.828, 47.823) .. (94.231,
	48.207).. controls (82.862, 48.661) and (71.58, 51.591) .. (60.215, 51.041)..
	controls (49.663, 50.531) and (39.555, 46.321) .. (29.034, 45.372) -- cycle;

	\node[text=c0d6400,line cap=butt,line join=miter,line width=1.0pt,miter 
	limit=4.0,anchor=south west] (text15) at (119.131, 69.509){$A$};

	\node[text=c0d6400,line cap=butt,line join=miter,line width=1.0pt,miter 
	limit=4.0,anchor=south west] (text16) at (126.187, 45.047){$B$};

		%
		%
		%
		%
		%
		%
		%
		%
		%
		%
		%
		%
		%
		%
		%
		%
		%
		%
		%
		%
		%
		%
		%
		%
		%
		%
		%
	
\end{tikzpicture}
		\caption{Transitivity.}
	\end{subfigure}\hfil
	\begin{subfigure}[b]{0.4\textwidth}\centering
		\definecolor{cededed}{RGB}{237,237,237}
\definecolor{cffffdc}{RGB}{255,255,220}
\definecolor{c644700}{RGB}{100,71,0}
\definecolor{ce0ffdc}{RGB}{224,255,220}
\definecolor{c0d6400}{RGB}{13,100,0}
\definecolor{cffebdc}{RGB}{255,235,220}
\definecolor{c640000}{RGB}{100,0,0}
\definecolor{cefffff}{RGB}{239,255,255}
\definecolor{c001764}{RGB}{0,23,100}

\def \globalscale {1.000000}
\begin{tikzpicture}[y=.75pt, x=.75pt, yscale=\globalscale,xscale=\globalscale, every node/.append style={scale=\globalscale}, inner sep=0pt, outer sep=0pt]
	
	\begin{scope}[blend group=multiply]
	\path[fill=cededed,line cap=butt,line join=miter,line width=1.0pt,miter 
	limit=4.0] (77.041, 111.568) -- (92.058, 113.122) -- (124.041, 25.248) -- 
	(46.042, 25.529) -- cycle;

	\path[draw=black,fill=cffffdc,line cap=butt,line join=miter,line 
	width=1.0pt,miter limit=4.0] (40.373, 119.073).. controls (37.142, 116.895) 
	and (35.14, 110.805) .. (37.538, 107.734).. controls (41.608, 102.521) and 
	(50.767, 107.612) .. (57.38, 107.734).. controls (70.986, 107.986) and 
	(93.148, 96.598) .. (98.177, 109.243).. controls (100.364, 114.744) and 
	(99.062, 124.426) .. (85.727, 121.907).. controls (80.079, 120.84) and 
	(76.972, 115.535) .. (71.819, 114.647).. controls (61.387, 112.849) and 
	(49.151, 124.988) .. (40.373, 119.073) -- cycle;
	\end{scope}
	\path[draw=black,line cap=butt,line join=miter,line width=0.5pt,miter 
	limit=4.0,dash pattern=on 0.5pt off 2.0pt] (38.771, 117.494) -- (5.299, 
	25.529);

	\node[text=c644700,line cap=butt,line join=miter,line width=1.0pt,miter 
	limit=4.0,anchor=south west] (text12) at (40.539, 108){$U$};

	\path[draw=black,line cap=butt,line join=miter,line width=0.5pt,miter 
	limit=4.0,dash pattern=on 0.5pt off 2.0pt] (99.9, 113.403) -- (131.884, 
	25.529);

	\path[draw=black,fill=ce0ffdc,line cap=butt,line join=miter,line 
	width=1.0pt,miter limit=4.0] (26.067, 80.567).. controls (24.568, 77.081) and 
	(21.665, 70.558) .. (20.892, 66.127).. controls (20.397, 63.292) and (27.749, 
	65.214) .. (30.346, 65.214).. controls (50.189, 65.214) and (100.703, 65.214) 
	.. (106.551, 65.214).. controls (108.028, 65.214) and (118.233, 61.516) .. 
	(116.315, 66.786) -- (111.156, 80.959).. controls (109.788, 84.718) and 
	(106.344, 80.519) .. (100.999, 80.468).. controls (79.555, 80.263) and 
	(53.271, 80.7) .. (34.703, 81.147).. controls (31.391, 81.226) and (28.156, 
	85.429) .. (26.067, 80.567) -- cycle;

	\node[text=c0d6400,line cap=butt,line join=miter,line width=1.0pt,miter 
	limit=4.0,anchor=south west] (text15) at (119.131, 69.509){$A$};

	\path[draw=black,fill=cffebdc,line cap=butt,line join=miter,line 
	width=1.0pt,miter limit=4.0] (77.041, 111.568).. controls (76.808, 108.806) 
	and (80.881, 106.695) .. (83.648, 106.515).. controls (86.873, 106.306) and 
	(91.757, 108.141) .. (92.046, 111.36).. controls (92.277, 113.93) and (88.704,
	116.081) .. (86.14, 116.361).. controls (82.732, 116.733) and (77.329, 
	114.984) .. (77.041, 111.568) -- cycle;

	\node[text=c640000,line cap=butt,line join=miter,line width=1.0pt,miter 
	limit=4.0,anchor=south west] (text17) at (102.562, 113.489){$W$};

	\path[draw=black,line cap=butt,line join=miter,line width=0.5pt,miter 
	limit=4.0,dash pattern=on 0.5pt off 2.0pt] (77.041, 111.568) -- (45.623, 
	25.248);

	\path[draw=black,line cap=butt,line join=miter,line width=0.5pt,miter 
	limit=4.0,dash pattern=on 0.5pt off 2.0pt] (92.058, 113.122) -- (124.041, 
	25.248);

	\path[draw=black,fill=cffebdc,line cap=butt,line join=miter,line 
	width=1.0pt,miter limit=4.0] (103.803, 80.855).. controls (102.936, 80.634) 
	and (102.001, 80.478) .. (101.0, 80.468).. controls (89.902, 80.362) and 
	(77.538, 80.431) .. (65.588, 80.581) -- (60.1, 65.215) -- (106.551, 65.215).. 
	controls (106.998, 65.215) and (108.269, 64.879) .. (109.736, 64.553) -- cycle;

	\node[text=c640000,line cap=butt,line join=miter,line width=1.0pt,miter 
	limit=4.0,anchor=south west] (text22) at (61, 48.962){$A\wedge \Down W$};

		%
		%
		%
		%
		%
		%
		%
		%
		%
		%
		%
		%
		%
		%
		%
		%
		%
		%
		%
		%
		%
		%
		%
		%
		%
		%
		%
	
\end{tikzpicture}
		\caption{Pullback stability.}
	\end{subfigure}
	\caption{Illustrations of the properties (c) and (d$^-$) of $\CovLeq^-$ from \cref{lemma:properties of coverage from locale}.}
	\label{figure:properties of coverage from locale}
\end{figure}
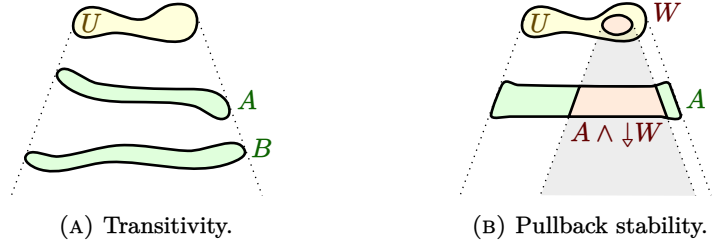
\begin{proposition}\label{lemma:properties of coverage from locale}
	For any parallel ordered locale $(X,\Leq)$ satisfying~\eqref{axiom:LV}, the causal coverages have the following properties:
	\begin{enumerate}[font=\normalfont]
		\item[(a)] $U\in\CovLeq^\pm(U)$;
		\item[(b)] $\Down U \in \CovLeq^-(U)$ and $\Up U \in\CovLeq^+(U)$;
		\item[(c)] if $B\in \CovLeq^\pm(A)$ and $A\in \CovLeq^\pm(U)$, then $B\in \CovLeq^\pm(U)$;
		\item[(d$^-$)] if $A\in \CovLeq^-(U)$ and $W\sqleq U$, then $A\wedge \Down W\in\CovLeq^-(W)$;
		\item[(d$^+$)] if $B\in \CovLeq^+(U)$ and $W\sqleq U$, then $B\wedge \Up W\in\CovLeq^+(W)$;
		\item[(e)] if $A\in \CovLeq^-(U)$ then $A\Leq U$, and if $B\in\CovLeq^+(U)$ then $U\Leq B$;
		\item[(f)] $\CovLeq^\pm(\varnothing)=\{\varnothing\}$.
	\end{enumerate}
\end{proposition}
\begin{proof}
	As usual, we only prove the past versions. Statement (a) is trivial: if $p$ is a path with endpoint in $U$, then $p$ already inhabits $U$. The claim in (b) follows similarly.
	
	For (c), note first that if $B\in \CovLeq^-(A)$ and $A\in\CovLeq^-(U)$ then $B\sqleq \Down A\sqleq \Down \Down U = \Down U$. Now let $p=(p_n)_{n=0}^N$ be a path landing in $U$. Since $A$ covers $U$ we get an open cover $(W_j)_{j\in J}$ of $p_N$, together with refinements $q^j$ of $p|_{W_j}$ with endpoints $q^j_{M_j}=W_j$. These refinements inhabit $A$, so there are indices $a_j$ such that $q^j_{a_j}\sqleq A$. Denote $\hat{q}^j = (q^j_m)_{m=0}^{a_j}$ and $\check{q}^j = (q^j_m)_{m=a_j}^{M_j}$. In turn, since $B$ covers $A$ we get open covers $(V^j_i)_{i\in I_j}$ of $q^j_{a_j}$, together with refinements $r^{ji}$ of $\hat{q}^j|_{V^j_i}$ that inhabit $B$ and have endpoint $V^j_i$. The concatenation $\check{q}^j|^{V^j_i}\cdot r^{ji}$ then forms a refinement of $q^j$, and hence of $p|_{W_j}$, that inhabits $B$.
	Moreover, using (the dual of) \cref{lemma:join over restrictions} we find
	\[
	\bigvee_{j\in J}\bigvee_{i\in I_j}(\check{q}^j|^{V^j_i})_{M_j}
	=
	\bigvee_{j\in J}q^j_{M_j}
	= 
	\bigvee_{j\in J}W_j
	=
	p_N.
	\]
	This gives the desired open cover and refinements to obtain $B\in\CovLeq^-(U)$.	
	
	For (d$^-$), clearly we have $A\wedge \Down W\sqleq \Down W$. Further, if $p$ is a path landing in~$W$, it also lands in $U$. Thus we can find local refinements $q^j$ that inhabit $A$ and land in the endpoint of $p$. But since the endpoint of $q^j$ is contained in $W$, each of its steps is contained in $\Down W$, so $q^j$ must inhabit $A\wedge \Down W$. 
	
	For (e), if $A\in\CovLeq^-(U)$ we already know $A\sqleq \Down U$. Now consider the path $p$ consisting of the single step:~$U$. Hence there exists an open cover $(U_i)_{i\in I}$ of~$U$, together with paths $q^i$ that inhabit $A$ and have endpoint $U_i$. Hence, for each $i\in I$, there exists some index $a_i$ such that $q^i_{a_i}\sqleq A$, which implies $U_i\sqleq \Up A$. In total, this gives $U=\bigvee_{i\in I} U_i\sqleq \Up A$. That $A\Leq U$ now follows by~\eqref{axiom:cones give order}.
	
	Lastly, for (f), if $A\in\CovLeq^-(\varnothing)$, then $A\sqleq \Down \varnothing=\varnothing$, so $\CovLeq^-(\varnothing)\subseteq \{\varnothing\}$, and the converse inclusion follows by (a).
\end{proof}

\begin{remark}
	Note in this proof that only property~(c) relies on the path restriction operation from \cref{construction:path restriction}, which in turn relies on parallel orderedness.
\end{remark}

The following gives an analogue of the fundamental axiom~\eqref{axiom:V}.

\begin{lemma}\label{lemma:coverage locale satisfies Cov-V weak}
	In a parallel ordered locale $(X,\Leq)$ with~\eqref{axiom:LV} we have:
	\[
	\forall i\in I: A_i\in\CovLeq^\pm(U_i)
	\quad \text{implies} \quad 
	\bigvee_{i\in I}A_i\in \CovLeq^\pm\left(\bigvee_{i\in I}U_i\right).
	\]
\end{lemma}
\begin{proof}
	Suppose that $A_i\in \CovLeq^-(U_i)$ for all $i\in I$, for some families $(A_i)_{i\in I}$ and $(U_i)_{i\in I}$ of opens. In the case that $\bigvee_{i\in I}U_i$ is empty, the result follows by \cref{lemma:properties of coverage from locale}(f). For the rest of the proof we can therefore assume that the join, and hence the index $I$, are non-empty. In that case, using monotonicity of the localic cones we get $\bigvee_{i\in I}A_i\sqleq \Down\bigvee_{i\in I}U_i$.
	
	We are left to show any path that lands in $\bigvee_{i\in I}U_i$ can be locally refined to inhabit $\bigvee_{i\in I} A_i$. For that, take a path $p=(p_n)_{n=0}^N$ that lands in $\bigvee_{i\in I}U_i$. Denote by $J\subseteq I$ the set of indices $j$ for which $W_j:=p_N\wedge U_j$ is non-empty. Since $p_N$ is non-empty, the set $J$ is non-empty, and $p_N = \bigvee_{j\in J} W_j$. The restrictions $p|_{W_j}$ land in $U_j$. Since $A_j$ covers $U_j$ from below, we get an open cover $(V^j_k)_{k\in K_j}$ of $W_j$, together with refinements $q^{jk}\refines (p|_{W_j})|_{V^j_k}$ that inhabit $A_j$ and have endpoint $V^j_k$. Using \cref{lemma:restriction of paths is functorial} we get $(p|_{W_j})|_{V^j_k} = p|_{V^j_k}$. This shows that the $q^{jk}$'s provide the desired local refinement of $p$. 
	%
	%
\end{proof}

Put differently, for any family $(U_i)_{i\in I}$ of opens the function
\[
\prod_{i\in I}\CovLeq^\pm(U_i)\longrightarrow\Opens X;\qquad (A_i)_{i\in I}\longmapsto \bigvee_{i\in I} A_i
\]
lands in $\CovLeq^\pm\left(\bigvee_{i\in I} U_i\right)$. The following result sharpens this further:~the image of this function in fact gives the whole set $\CovLeq^\pm\left(\bigvee_{i\in I} U_i\right)$. This means that the causal cover of a join can always be written as the union of causal covers of the constituents.

\begin{lemma}\label{lemma:coverage locale satisfies Cov-V strong}
	In a parallel ordered locale $(X,\Leq)$ with~\eqref{axiom:LV} we have:
	\[\tag{$\Cov$-$\vee$}\label{axiom:Cov-V strong}
	\CovLeq^\pm\left(\bigvee_{i\in I} U_i\right) = \left\{\bigvee_{i\in I} A_i : (A_i)_{i\in I}\in\prod_{i\in I} \CovLeq^\pm(U_i)\right\}.
	\]
\end{lemma}
\begin{proof}
	The inclusion from right-to-left is precisely \cref{lemma:coverage locale satisfies Cov-V weak}. For the converse, take $A\in \CovLeq^-\left(\bigvee_{i\in I} U_i\right)$. Given \cref{lemma:properties of coverage from locale}(f) we can assume without loss of generality that the indexing set is non-empty. We need to show that $A$ can be written as a join of open regions that cover the $U_i$'s from below. By \cref{lemma:properties of coverage from locale}(d$^-$) we get $A\wedge \Down U_i\in \CovLeq^-(U_i)$ for every $i\in I$. Using axiom~\eqref{axiom:LV} we find the desired expression: $A = A\wedge\Down \bigvee_{i\in I} U_i= \bigvee_{i\in I} A\wedge \Down U_i$. The proof for future covers is dual.
\end{proof}

\section{Causal coverages as Grothendieck topologies}
\label{section:causal coverages as grothendieck topologies}
In this section we explore the fact that the properties of $\CovLeq^\pm$ showcased in \cref{lemma:properties of coverage from locale} share a remarkable resemblance to the fundamental axioms of a coverage. In particular, we shall prove that the causal coverage relation in an ordered locale canonically defines a type of \emph{Grothendieck topology}. These are very important mathematical structures in category theory, and relate to both geometry and logic~\cite{maclane1994SheavesGeometryLogic}. Grothendieck topologies are among the most general incarnations of the notion of coverage, and one of their main points is to study how `local data' is amalgamated into `global data'. In the causal setting we shall interpret it as dictating how `past data' evolves into `future data'.

For convenience, we first recall the definition of a Grothendieck topology in full generality on a category~$\cat{C}$. Recall that a \emph{category} consists of a collection of \emph{objects}, together with a collection of \emph{arrows} between them that can be composed via a unital and associative operation~\cite{maclane1998CategoriesWorkingMathematician}. An arrow $f\colon A\to B$ between objects $A,B\in\cat{C}$ of a category has \emph{domain} $A$, and \emph{codomain} $B$. 

\begin{example}\label{example:category of opens}
	Let $X$ be a locale. Then $\Opens X$ is a category, the \emph{category of opens} of $X$, where the objects are the open regions $U\in\Opens X$, and there exists a unique arrow $U\to V$ iff the inclusion $U\sqleq V$ holds. Thus, for the moment, we can think of categories as an intricate network of abstract regions, where the arrows tell us how these regions are related and included in one another.
\end{example}

\begin{definition}\label{definition:sieve}
	A \emph{sieve} $R$ on an object $C\in\cat{C}$ of a category is a set of arrows with codomain $C$, such that if $(D\xrightarrow{g}C)\in R$, and $E\xrightarrow{h}D$ is an arbitrary arrow in~$\cat{C}$, then $(E\xrightarrow{g}D\xrightarrow{h}C)\in R$. 
\end{definition}

For any object $C\in\cat{C}$ we get the \emph{maximal sieve}:
\[
t_C:= \left\{\text{all arrows with codomain $C$}\right\}.
\]
If $R$ is a sieve on $C\in\cat{C}$ and $h\colon D\to C$ is an arrow, the \emph{pullback} of $R$ along $h$ is the sieve on the domain~$D$ defined as:
\[
h^\ast(R):=\left\{E\xrightarrow{~f~}D: h\circ f\in R \right\}.
\]


\begin{definition}[{\cite[\S III.2]{maclane1994SheavesGeometryLogic}}]\label{definition:grothendieck topology}
	A \emph{Grothendieck topology} on a category $\cat{C}$ is a function $J$ that assigns to each object $C\in\cat{C}$ a collection $J(C)$ of sieves on $C$, called \emph{covering sieves}, such that the following axioms are satisfied:
	\begin{enumerate}
		\item[(i)] the maximal sieve $t_C$ is in $J(C)$;
		\item[(ii)] if $S\in J(C)$ then $h^\ast(S)\in J(D)$ for any arrow $h\colon D\to C$;
		\item[(iii)] if $S\in J(C)$ and $R$ is a sieve on $C$ such that $h^\ast(R)\in J(D)$ for all $h\colon D\to C$ in $S$, then $R\in J(C)$.
	\end{enumerate}
	A category $\cat{C}$ equipped with a Grothendieck topology $J$ is called a \emph{site}.
\end{definition}

\begin{remark}[Intuition]
	As this definition is quite abstract, it helps to unpack it in the setting where $\cat{C}$ is the category of opens from \cref{example:category of opens}. Thus the objects are the opens of $X$, and an arrow $f\colon V\to U$ is just an inclusion $f\colon V\sqleq U$, which is unique if it exists. In the following we shall be somewhat informal in conflating the arrow $f$ with the open subset $V\sqleq U$ itself. In so doing, the maximal sieve on $U\in\Opens X$ can be identified with the principal ideal generated by~$U$:
	\[
	t_U= \{V\in\Opens X:V\sqleq U\}.
	\]
	In general, a sieve $R$ on $U\in\Opens X$ is a collection of subregions of $U$ that are down-closed under the inclusion relation $\sqleq$. If $h\colon V\sqleq U$ is an arrow and $R$ is a sieve on~$U$, then its pullback can be calculated using intersections:
	\[
	h^\ast(R) = \{V\wedge W: W\in R\}.
	\]
	With this in mind, we can reinterpret the axioms of a Grothendieck topology in \cref{definition:grothendieck topology} for the \emph{canonical coverage} $J$ on $\Opens X$, which is defined as:
	\[
	R\in J(U)
	\qquad\text{if and only if}\qquad
	\bigvee R = U.
	\]
	In other words, the covering sieves $R\in J(U)$ are precisely the down-closed open covers of $U$. It is straightforward to see that the axioms hold:
	\begin{enumerate}
		\item[(i)] $\bigvee\{V\in\Opens X: V\sqleq U\} = U$;
		\item[(ii)] if $\bigvee S = U$ and $V\sqleq U$ then $V\wedge \bigvee S = V$;
		\item[(iii)] if $\bigvee S = U$ and $R$ is a sieve on $U$ such that $V\wedge \bigvee R = V$ for all $V\in S$, then $\bigvee R = U$. This is true because: $\bigvee S = \bigvee_{V\in S} (V\wedge \bigvee R) = \bigvee S \wedge \bigvee R$.
	\end{enumerate}
\end{remark}

Motivated by~\cref{lemma:properties of coverage from locale}, we now want to interpret the causal coverages $\CovLeq^\pm$ of an ordered locale as Grothendieck topologies $J^\pm$. To motivate the definition, we compare this structure to the canonical coverage $J$ on $\Opens X$. There, a sieve `$R$ covers $U$' if $R$ forms an open cover for $U$. In the causal setting we want to say `$R$ \emph{causally} covers $U$' if $R$ forms an open cover for some $A\in \CovLeq^-(U)$. 

But this poses a problem, since with respect to the inclusion order $\sqleq$ the sieves~$R$ on $U$ by definition only contain opens $V\in R$ that are subregions $V\sqleq U$, while the interesting causal covering behaviour of $\CovLeq^\pm$ happens when the regions $V$ are disjoint from $U$. We thus need some means to say when a sieve $R$ is covering for $U$ while $R$ itself is not strictly a sieve on $U$ itself. Rather, we want to consider sieves on~$\Down U$. For this, we propose the following generalisation of Grothendieck topologies.
To state it, we need to recall that an \emph{endofunctor} $T \colon \cat{C} \to \cat{C}$ sends objects $C \in \cat{C}$ to objects $T(C) \in \cat{C}$, and arrows $f \colon D \to C$ in $\cat{C}$ to arrows $T(f) \colon T(D) \to T(C)$ in $\cat{C}$ in a way that respects composition and its units, and that a \emph{natural transformation} $\mu \colon T^2 \to T$ is a family of morphisms $\mu_C \colon T(T(C)) \to T(C)$, indexed by objects $C \in \cat{C}$, satisfying $\mu_C \circ T(T(f)) = T(f) \circ \mu_D$ for each arrow $f \colon D \to C$ in $\cat{C}$.

\begin{definition}\label{definition:modified grothendieck topology}
	Let $\cat{C}$ be a category, together with an endofunctor $T\colon \cat{C}\to \cat{C}$ and a natural transformation $\mu\colon T^2\to T$. A \emph{$(T,\mu)$-Grothendieck topology} on $\cat{C}$ is a function $J_T$ that assigns to each object $C\in\cat{C}$ a collection $J_T(C)$ of sieves on~$T(C)$, called \emph{covering sieves}, such that the following axioms are satisfied:
	\begin{enumerate}
		\item[(i)] the maximal sieve $t_{T(C)}$ is in $J_T(C)$;
		\item[(ii)] if $S\in J_T(C)$ then $T(h)^\ast(S)\in J_T(D)$ for any arrow $h\colon D\to C$;
		\item[(iii)] if $S\in J_T(C)$ and $R$ is a sieve on $T(C)$ such that $(\mu_C\circ T(h))^\ast(R)\in J_T(D)$ for all $h\colon D\to T(C)$ in $S$, then $R\in J_T(C)$.
	\end{enumerate}
\end{definition}

Setting $T(C)=C$, $T(f)=f$, and $\mu_C = \id_C$ of course returns the ordinary notion of Grothendieck topology on $\cat{C}$. The natural transformation $\mu$ is needed only to make sense of axiom~(iii), where for $h\colon D\to T(C)$ in a covering sieve $S\in J_T(C)$ we need to be able to construct the sieve $(\mu_C\circ T(h))^\ast(R)$ on $T(D)$.


The following result shows that the causal coverage $\CovLeq^-$ of an ordered locale gives rise to a canonical $\Down$-Grothendieck topology, and also explains the intuition of \cref{definition:modified grothendieck topology}.

\begin{theorem}\label{theorem:causal grothendieck topology}
	Let $(X,\Leq)$ be a parallel ordered locale with~\eqref{axiom:LV}. The causal coverage $\CovLeq^-$ defines a $\Down$-Grothendieck topology~$J^-$ on~$\Opens X$ via
	\[
	R\in J^-(U)
	\qquad\text{if and only if}\qquad
	\bigvee R\in \CovLeq^-(U).
	\]
\end{theorem}
\begin{proof}
	We check the axioms (see \cref{figure:properties of causal grothendieck topology} for intuition).
	\begin{enumerate}
		\item[(i)] That $t_{\Down U}\in J^-(U)$ just means $\bigvee \{V\in\Opens X: V\sqleq \Down U\} = \Down U \in \CovLeq^-(U)$, which holds by \cref{lemma:properties of coverage from locale}(b);
		
		\item[(ii)] Let $S\in J^-(U)$ and take an arrow $h\colon W\sqleq U$. We need to show that $(\Down h)^\ast(S)\in J^-(W)$. As above, we identify arrows $V\sqleq \Down U$ in $S$ with their codomains, which allows us to write the pullback as
		\[
		(\Down h)^\ast(S) = \{V\wedge \Down W: V\in S\}.
		\]
		Indeed, if $V\in S$ then $V\wedge \Down W\in (\Down h)^\ast(S)$ since $V\wedge \Down W\sqleq V\in S$, and~$S$ is downwards closed with respect to the inclusion relation. Conversely, if $V\in (\Down h)^\ast(S)$ then the composition $V\sqleq \Down W\sqleq \Down U$ is in $S$, and in particular we get $V=V\wedge \Down W$.
		Thus taking the join over $(\Down h)^\ast(S)$ and using infinite distributivity we see
		\[
		(\Down h)^\ast(S) \in J^-(W)
		\qquad\text{if and only if}\qquad
		\Down W\wedge \bigvee S\in \CovLeq^-(W),
		\]
		and the right hand side holds by \cref{lemma:properties of coverage from locale}(d$^-$).
		
		\item[(iii)] Let $S\in J^-(U)$, and let $R$ be a sieve on $\Down U$ with the property that for every $h\colon V\sqleq \Down U$ in $S$ we get $(\mu_U\circ \Down h)^\ast(R)\in J^-(V)$. In this case the natural map $\mu$ is just the family of equations  $\mu_U\colon \Down\Down U = \Down U$. Similar to the above, we can calculate the pullback as
		\[
		(\mu_U\circ \Down h)^\ast(R) = \{W\wedge \Down V: W\in R\}.
		\]
		Translating into the language of causal coverages, we get ${\bigvee S \in \CovLeq^-(U)}$, and $R$ is a sieve such that $\Down V\wedge \bigvee R\in \CovLeq^-(V)$ for every $V\in S$. From~\eqref{axiom:Cov-V strong} and using~\eqref{axiom:LV} with infinite distributivity it then follows that
		\[
		\Down \bigvee S \wedge \bigvee R = \bigvee_{V\in S}\left(\Down V\wedge \bigvee R\right)\in \CovLeq^-\left(\bigvee S\right).
		\]
		By \cref{lemma:properties of coverage from locale}(c) it thus follows that $\Down \bigvee S\wedge \bigvee R\in \CovLeq^-(U)$. Since $R$ is a sieve on $\Down U$ we have $\bigvee R\sqleq \Down U$, and since it contains a region that covers~$U$ from below we must in fact have $\bigvee R\in \CovLeq^-(U)$. \qedhere
	\end{enumerate}
\end{proof}
\begin{figure}[t]\centering
	\begin{subfigure}[b]{0.45\textwidth}\centering
		\definecolor{cededed}{RGB}{237,237,237}
\definecolor{cffffdc}{RGB}{255,255,220}
\definecolor{c644700}{RGB}{100,71,0}
\definecolor{cffebdc}{RGB}{255,235,220}
\definecolor{c640000}{RGB}{100,0,0}
\definecolor{ce0ffdc}{RGB}{224,255,220}
\definecolor{c0d6400}{RGB}{13,100,0}
\definecolor{cefffff}{RGB}{239,255,255}
\definecolor{c001764}{RGB}{0,23,100}
\definecolor{ccd0000}{RGB}{205,0,0}

\def \globalscale {1.000000}
\begin{tikzpicture}[y=.8pt, x=.8pt, yscale=\globalscale,xscale=\globalscale, every node/.append style={scale=\globalscale}, inner sep=0pt, outer sep=0pt]
	\begin{scope}[blend group=multiply]
		
		\path[fill=cededed,line cap=butt,line join=miter,line width=1.0pt,miter 
		limit=4.0,fill opacity =.6] (74.388, 114.653) -- (43.885, 31.032) -- (125.109, 31.032) -- 
		(95.607, 112.017) -- cycle;

		\path[draw=black,fill=cffffdc,line cap=butt,line join=miter,line 
		width=1.0pt,miter limit=4.0] (22.593, 119.073).. controls (18.632, 116.784) 
		and (14.434, 112.289) .. (14.861, 107.734).. controls (15.11, 105.073) and 
		(18.023, 102.99) .. (20.53, 102.065).. controls (31.848, 97.889) and (44.547, 
		105.094) .. (56.609, 104.899).. controls (81.253, 104.502) and (106.809, 
		93.231) .. (116.137, 102.065).. controls (119.572, 105.318) and (119.965, 
		112.627) .. (116.908, 116.238).. controls (105.394, 129.839) and (81.597, 
		121.862) .. (63.78, 122.183).. controls (50.014, 122.432) and (34.515, 
		125.959) .. (22.593, 119.073) -- cycle;
	\end{scope}

	\node[text=black,line cap=butt,line join=miter,line width=1.0pt,miter 
	limit=4.0,anchor=south west] (text1) at (20.53, 104.899){$U$};

	\path[draw=black,fill=cffebdc,line cap=butt,line join=miter,line 
	width=1.0pt,miter limit=4.0] (74.25, 109.552).. controls (76.209, 105.864) and
	(82.116, 105.361) .. (86.25, 105.951).. controls (89.93, 106.476) and 
	(95.994, 108.32) .. (95.607, 112.017).. controls (94.923, 118.545) and 
	(83.131, 121.933) .. (77.223, 119.073).. controls (74.23, 117.624) and (72.69,
	112.488) .. (74.25, 109.552) -- cycle;

	\node[text=c640000,line cap=butt,line join=miter,line width=1.0pt,miter 
	limit=4.0,anchor=south west] (text2) at (77.223, 107.734){$W$};

	\path[draw=black,line cap=butt,line join=miter,line width=0.5pt,miter 
	limit=4.0,dash pattern=on 0.5pt off 2.0pt] (14.861, 110.569) -- (-14.114, 
	31.032);

	\path[draw=black,line cap=butt,line join=miter,line width=0.5pt,miter 
	limit=4.0,dash pattern=on 0.5pt off 2.0pt] (118.789, 111.932) -- (148.095, 
	31.032);

	\begin{scope}[blend group=multiply]
		
		\path[draw=black,fill=ce0ffdc,line cap=butt,line join=miter,line 
		width=1.0pt,miter limit=4.0] (17.633, 75.401).. controls (13.117, 76.488) and 
		(7.375, 78.94) .. (3.712, 76.082).. controls (1.631, 74.458) and (1.116, 
		70.737) .. (2.165, 68.314).. controls (3.999, 64.076) and (9.329, 61.442) .. 
		(13.925, 60.99).. controls (20.81, 60.311) and (33.128, 61.399) .. (33.346, 
		68.314).. controls (33.527, 74.057) and (23.219, 74.056) .. (17.633, 75.401) 
		-- cycle;

		\path[draw=black,fill=ce0ffdc,line width=1.0pt] (54.753, 60.726) -- (60.393, 
		76.217).. controls (54.256, 75.801) and (49.404, 74.325) .. (44.695, 70.866)..
		controls (42.415, 69.191) and (42.601, 64.444) .. (44.536, 62.38).. controls 
		(47.788, 58.911) and (49.813, 59.842) .. (54.753, 60.726) -- cycle;

		\path[draw=black,fill=cffebdc,line cap=butt,line join=miter,line 
		width=1.0pt,miter limit=4.0] (82.892, 73.718).. controls (80.7, 74.626) and 
		(77.727, 75.065) .. (75.773, 73.718).. controls (74.218, 72.647) and (73.465, 
		70.292) .. (73.756, 68.427).. controls (74.239, 65.329) and (76.509, 61.443) 
		.. (79.622, 61.075).. controls (83.375, 60.633) and (88.329, 64.277) .. 
		(88.562, 68.049).. controls (88.726, 70.716) and (85.361, 72.696) .. (82.892, 
		73.718) -- cycle;

		\path[draw=black,fill=ce0ffdc,line width=1.0pt] (109.318, 74.344) -- (112.743,
		64.933).. controls (115.258, 65.761) and (116.939, 67.233) .. (116.908, 
		69.844).. controls (116.867, 73.279) and (113.503, 74.186) .. (109.318, 
		74.344) -- cycle;

		\path[draw=black,fill=ce0ffdc,line cap=butt,line join=miter,line 
		width=1.0pt,miter limit=4.0] (41.966, 73.718).. controls (38.155, 74.332) and 
		(33.819, 71.38) .. (31.869, 68.049).. controls (31.151, 66.823) and (30.921, 
		64.845) .. (31.869, 63.786).. controls (34.702, 60.621) and (40.855, 60.259) 
		.. (44.536, 62.38).. controls (46.35, 63.425) and (47.406, 66.063) .. (47.067,
		68.13).. controls (46.659, 70.619) and (44.456, 73.317) .. (41.966, 73.718) 
		-- cycle;

		\path[draw=black,fill=ce0ffdc,line cap=butt,line join=miter,line 
		width=1.0pt,miter limit=4.0] (124.341, 77.135).. controls (120.392, 77.816) 
		and (114.55, 74.862) .. (114.073, 70.884).. controls (113.755, 68.23) and 
		(117.575, 66.777) .. (119.743, 65.214).. controls (121.833, 63.708) and 
		(124.08, 61.795) .. (126.656, 61.757).. controls (129.372, 61.718) and 
		(134.137, 62.474) .. (134.049, 65.19).. controls (133.562, 77.067) and 
		(129.397, 76.263) .. (124.341, 77.135) -- cycle;
		
		\path[draw=black,fill=cffebdc,line width=1.0pt] (54.753, 60.726).. controls 
		(57.601, 61.236) and (61.41, 61.731) .. (67.002, 61.356).. controls (70.936, 
		61.093) and (76.405, 64.484) .. (76.485, 68.426).. controls (76.57, 72.605) 
		and (70.922, 76.222) .. (66.743, 76.324).. controls (64.465, 76.38) and 
		(62.365, 76.351) .. (60.393, 76.217) -- cycle;

		\path[draw=black,fill=cffebdc,line width=1.0pt] (109.318, 74.344).. controls 
		(105.236, 74.498) and (100.373, 73.939) .. (97.042, 74.308).. controls 
		(94.683, 74.57) and (91.926, 75.581) .. (89.922, 74.308).. controls (87.348, 
		72.674) and (84.691, 68.768) .. (86.03, 66.028).. controls (87.261, 63.51) and
		(91.428, 64.142) .. (94.23, 64.173).. controls (99.411, 64.232) and (107.751,
		63.291) .. (112.743, 64.933) -- cycle;
	\end{scope}

	\path[draw=black,line cap=butt,line join=miter,line width=0.5pt,miter 
	limit=4.0,dash pattern=on 0.5pt off 2.0pt] (74.388, 114.653) -- (44.013, 
	31.199);

	\path[draw=black,line cap=butt,line join=miter,line width=0.5pt,miter 
	limit=4.0,dash pattern=on 0.5pt off 2.0pt] (95.607, 112.017) -- (125.022, 
	31.199);

	\node[text=c0d6400,line cap=butt,line join=miter,line width=1.0pt,miter 
	limit=4.0,anchor=south west] (text26) at (12.588, 80.183){$\text{$V$s in 
			$S$}$};

	\path[decoration={brace},decorate,draw=c0d6400,line width=1.0pt,text=c0d6400] (140.085, 76.553) -- node[right=.4em] {$S$}(140.085, 62.38);

	\path[decoration={brace,mirror},decorate,draw=c640000,line width=1.0pt,text=c640000] (55.5, 56.71) -- node[below=.5em] {$(\Down h)^\ast(S)$}(113, 56.614);

		%
		%
		%
		%
		%
		%
		%
		%
		%
		%
		%
		%
		%
		%
		%
		%
		%
		%
		%
		%
		%
		%
		%
		%
		%
		%
		%
		%
		%
		%
	
\end{tikzpicture}
		\caption{Axiom (ii).}
	\end{subfigure}
	\begin{subfigure}[b]{0.45\textwidth}\centering
		\definecolor{cededed}{RGB}{237,237,237}
\definecolor{cffffdc}{RGB}{255,255,220}
\definecolor{ce0ffdc}{RGB}{224,255,220}
\definecolor{c0d6400}{RGB}{13,100,0}
\definecolor{c640000}{RGB}{100,0,0}
\definecolor{cffebdc}{RGB}{255,235,220}
\definecolor{cefffff}{RGB}{239,255,255}
\definecolor{c001764}{RGB}{0,23,100}
\definecolor{ccd0000}{RGB}{205,0,0}
\definecolor{c644700}{RGB}{100,71,0}

\def \globalscale {1.000000}
\begin{tikzpicture}[y=.9pt, x=.9pt, yscale=\globalscale,xscale=\globalscale, every node/.append style={scale=\globalscale}, inner sep=0pt, outer sep=0pt]
	\path[fill=cededed,opacity=0.6,line cap=butt,line join=miter,line 
	width=1.0pt,miter limit=4.0] (67.801, 86.222) -- (47.774, 31.199) -- (108.569,
	31.199) -- (89.086, 84.524) -- cycle;

	\path[draw=black,fill=cffffdc,line cap=butt,line join=miter,line 
	width=1.0pt,miter limit=4.0] (26.199, 110.569).. controls (27.651, 94.954) and
	(48.66, 104.257) .. (60.215, 103.271).. controls (73.407, 102.145) and 
	(86.661, 104.899) .. (99.9, 104.899).. controls (106.994, 104.899) and 
	(119.036, 105.399) .. (118.32, 111.932).. controls (117.621, 118.309) and 
	(105.672, 114.198) .. (99.28, 114.738).. controls (76.884, 116.631) and 
	(59.105, 116.238) .. (31.869, 116.238).. controls (28.263, 116.238) and 
	(25.952, 113.23) .. (26.199, 110.569) -- cycle;

	\node[text=black,line cap=butt,line join=miter,line width=1.0pt,miter 
	limit=4.0,anchor=south west] (text1) at (31.447, 104){$U$};

	\path[draw=black,line cap=butt,line join=miter,line width=0.5pt,miter 
	limit=4.0,dash pattern=on 0.5pt off 2.0pt] (26.403, 112.367) -- (-3.14, 
	31.199);

	\path[draw=black,line cap=butt,line join=miter,line width=0.5pt,miter 
	limit=4.0,dash pattern=on 0.5pt off 2.0pt] (118.789, 111.932) -- (148.173, 
	31.199);

	\path[draw=black,line cap=butt,line join=miter,line width=0.5pt,miter 
	limit=4.0,dash pattern=on 0.5pt off 2.0pt] (67.801, 86.222) -- (47.774, 
	31.199);

	\path[draw=black,line cap=butt,line join=miter,line width=0.5pt,miter 
	limit=4.0,dash pattern=on 0.5pt off 2.0pt] (89.086, 84.524) -- (108.404, 
	31.447);

	\path[draw=black,fill=ce0ffdc,line cap=butt,line join=miter,line 
	width=1.0pt,miter limit=4.0] (43.207, 62.38).. controls (57.38, 62.38) and 
	(66.354, 64.488) .. (65.884, 56.71).. controls (65.425, 49.106) and (54.546, 
	53.876) .. (43.207, 53.876);

	\path[draw=black,fill=ce0ffdc,line cap=butt,line join=miter,line 
	width=1.0pt,miter limit=4.0] (111.239, 62.38).. controls (102.735, 62.38) and 
	(87.426, 63.023) .. (88.562, 56.71).. controls (89.711, 50.319) and (102.735, 
	53.876) .. (111.239, 53.876);

	\path[draw=black,line cap=butt,line join=miter,line width=1.0pt,miter 
	limit=4.0,dash pattern=on 1.0pt off 1.0pt] (43.207, 62.38) -- (40.373, 62.38);

	\path[draw=black,line cap=butt,line join=miter,line width=1.0pt,miter 
	limit=4.0,dash pattern=on 1.0pt off 1.0pt] (43.207, 53.876) -- (40.373, 
	53.876);

	\path[draw=black,line cap=butt,line join=miter,line width=1.0pt,miter 
	limit=4.0,dash pattern=on 1.0pt off 1.0pt] (111.239, 62.38) -- (114.073, 
	62.38);

	\path[draw=black,line cap=butt,line join=miter,line width=1.0pt,miter 
	limit=4.0,dash pattern=on 1.0pt off 1.0pt] (111.239, 53.876) -- (114.073, 
	53.876);

	\path[decoration={brace},decorate,draw=c0d6400,line cap=butt,line join=miter,line width=1.0pt,miter 
	limit=4.0,text = c0d6400] (133, 96.395) -- node[right=.4em] {$S$}(133, 82.222);

	\path[decoration={brace},decorate,draw=c0d6400,line cap=butt,line join=miter,line width=1.0pt,miter 
	limit=4.0,text=c0d6400] (132, 63.313) -- node[right=.4em, fill = white, inner sep = 2pt] {$R$}(132, 51.975);
	
	\path[decoration={brace,mirror},decorate,draw=c640000,line cap=butt,line join=miter,line width=1.0pt,miter 
	limit=4.0,text = c640000] (57.38, 48.206) -- node[below=.5em,scale = .9] {$(\mu_U{\circ}\Down h)^\ast(R)$}(99.9, 48.206);

	\node[text=c640000,line cap=butt,line join=miter,line width=1.0pt,miter 
	limit=4.0,anchor=south west] (text32) at (72.604, 91.84){$V$};

	\node[text=c0d6400,line cap=butt,line join=miter,line width=1.0pt,miter 
	limit=4.0,anchor=south west] (text35) at (29.478, 55.39){$\cdots$};

	\node[text=c0d6400,line cap=butt,line join=miter,line width=1.0pt,miter 
	limit=4.0,anchor=south west] (text35-5) at (112.619, 55.39){$\cdots$};

	\begin{scope}[blend group=multiply]
		
		\path[draw=black,fill=cffebdc,line cap=butt,line join=miter,line 
		width=1.0pt,miter limit=4.0] (70.353, 61.88).. controls (75.526, 63.481) and 
		(81.344, 63.428) .. (86.597, 62.113).. controls (89.586, 61.365) and (94.964, 
		60.712) .. (94.687, 57.644).. controls (94.149, 51.667) and (83.626, 51.9) .. 
		(77.625, 51.897).. controls (72.416, 51.895) and (63.493, 51.709) .. (62.812, 
		56.873).. controls (62.418, 59.864) and (67.471, 60.987) .. (70.353, 61.88) --
		cycle;

		\path[draw=black,fill=cffebdc,line width=1.0pt] (55.661, 52.721).. controls 
		(61.436, 52.042) and (65.604, 52.071) .. (65.884, 56.71).. controls (66.134, 
		60.853) and (63.683, 62.181) .. (59.215, 62.527) -- cycle;

		\path[draw=black,fill=cffebdc,line width=1.0pt] (97.276, 62.021).. controls 
		(92.047, 61.46) and (87.956, 60.08) .. (88.562, 56.71).. controls (89.276, 
		52.738) and (94.577, 52.609) .. (100.529, 53.08) -- cycle;

		\path[draw=black,fill=ce0ffdc,line cap=butt,line join=miter,line 
		width=1.0pt,miter limit=4.0] (20.395, 92.681).. controls (17.518, 90.587) and 
		(17.543, 84.524) .. (20.063, 82.012).. controls (21.513, 80.567) and (24.182, 
		81.878) .. (26.199, 82.222).. controls (31.217, 83.078) and (40.827, 81.699) 
		.. (40.77, 86.789).. controls (40.692, 93.859) and (26.111, 96.841) .. 
		(20.395, 92.681) -- cycle;

		\path[draw=black,fill=ce0ffdc,line cap=butt,line join=miter,line 
		width=1.0pt,miter limit=4.0] (53.326, 92.088).. controls (47.696, 91.7) and 
		(37.387, 89.417) .. (38.514, 83.887).. controls (39.587, 78.623) and (49.176, 
		82.057) .. (54.546, 82.222).. controls (60.291, 82.399) and (70.971, 79.339) 
		.. (71.554, 85.057).. controls (72.213, 91.535) and (59.823, 92.535) .. 
		(53.326, 92.088) -- cycle;

		\path[draw=black,fill=ce0ffdc,line cap=butt,line join=miter,line 
		width=1.0pt,miter limit=4.0] (124.322, 94.671).. controls (121.831, 98.269) 
		and (115.007, 95.787) .. (111.239, 93.561).. controls (108.306, 91.828) and 
		(103.962, 88.061) .. (105.569, 85.057).. controls (107.799, 80.891) and 
		(115.748, 82.535) .. (119.743, 85.057).. controls (122.744, 86.952) and 
		(126.342, 91.752) .. (124.322, 94.671) -- cycle;

		\path[draw=black,fill=ce0ffdc,line cap=butt,line join=miter,line 
		width=1.0pt,miter limit=4.0] (108.894, 86.777).. controls (108.905, 92.162) 
		and (99.533, 94.504) .. (94.231, 93.561).. controls (90.284, 92.859) and 
		(84.322, 88.811) .. (85.727, 85.057).. controls (87.701, 79.781) and (97.043, 
		82.878) .. (102.585, 83.887).. controls (104.861, 84.302) and (108.89, 84.464)
		.. (108.894, 86.777) -- cycle;

		\path[draw=black,fill=cffebdc,line cap=butt,line join=miter,line 
		width=1.0pt,miter limit=4.0] (80.853, 88.757).. controls (77.224, 89.788) and 
		(72.611, 90.562) .. (69.541, 88.368).. controls (68.006, 87.271) and (66.804, 
		84.821) .. (67.525, 83.077).. controls (69.102, 79.258) and (74.695, 77.93) ..
		(78.826, 77.988).. controls (82.726, 78.042) and (89.09, 79.197) .. (89.357, 
		83.088).. controls (89.59, 86.487) and (84.13, 87.826) .. (80.853, 88.757) -- 
		cycle;
	\end{scope}

		%
		%
		%
		%
		%
		%
		%
		%
		%
		%
		%
		%
		%
		%
		%
		%
		%
		%
		%
		%
		%
		%
		%
		%
		%
		%
		%
		%
		%
		%
	
\end{tikzpicture}
		\caption{Axiom (iii).}
	\end{subfigure}
	\caption{Illustrations for \cref{theorem:causal grothendieck topology}.}
	\label{figure:properties of causal grothendieck topology}
\end{figure}

\begin{remark}
	One might think $J^-$ could define a Grothendieck topology (in the original sense of \cref{definition:grothendieck topology}) on the so-called \emph{Kleisli category} $\Kleisli(\Down)$ induced by the past localic cone monad. Explicitly, this is the category whose objects are opens $U\in \Opens X$ with a unique arrow $V\to U$ precisely if $V\sqleq \Down U$. We denote these arrows by $V\sqleqdown U$. That $U\sqleqdown U$ follows since $U\sqleq \Down U$, and that $W\sqleqdown V\sqleqdown U$ implies $W\sqleqdown U$ follows since $W\sqleq \Down V\sqleq \Down \Down U = \Down U$. Therefore $\Kleisli(\Down)$ is just the set $\Opens X$ equipped with the preorder $\sqleqdown$.
	
	In this preorder, a sieve $R$ on $U$ is indeed some collection of open regions $V\sqleq \Down U$, which at first glance seems like what we were looking for. However, defining $J^-$ as above, we see that axiom~(ii) of a Grothendieck topology can fail. Namely, if $S\in J^-(U)$ is a covering sieve and $h\colon W\sqleqdown U$, this does not force $W\sqleq U$. It can happen that $W$ is in the past of and disjoint from $\bigvee S$, so that $h^\ast(S) = \{V\wedge \Down W:V\in S\}\notin J^-(W)$, since $\Down W\wedge \bigvee S=\varnothing\notin \CovLeq^-(W)$. In order to interpret $\CovLeq^\pm$ as Grothendieck topologies we are therefore seemingly forced to consider the generalised notion of Grothendieck topology in \cref{definition:modified grothendieck topology}.
\end{remark}
%

\begin{remark}
	The fact that the localic cones are monads suggests a refinement of the \cref{definition:modified grothendieck topology} of $(T,\mu)$-Grothendieck topologies on $\cat{C}$ to the setting where $T$ is in fact a monad. This gives the additional data of a natural unit map $\eta\colon \id_\cat{C}\to T$, allowing us to state the additional axiom:
	\begin{enumerate}
		\item[(i$'$)] $(\eta_C)_\ast(t_C)\in J_T(C)$.
	\end{enumerate}
	Here $t_C$ is the maximal sieve on $C\in\cat{C}$, and we have the \emph{pushforward} sieve on~$T(C)$:
	\[
	(\eta_C)_\ast(t_C):=\{ \eta_C\circ f:f\in t_C\}.
	\]
	
	In the setting of causal coverages, where $J_T = J^-$ from \cref{theorem:causal grothendieck topology}, we get unit $\eta_U\colon U\sqleq \Down U$, and axiom~(i$'$) translates to $\bigvee (\eta_U)_\ast(t_U) = U\in \CovLeq^-(U)$, which holds by \cref{lemma:properties of coverage from locale}(a). We can interpret this as a compatibility condition between~$J^-$ and the canonical Grothendieck topology induced by the frame structure of $\Opens X$.
	
	Pushing that train of thought even further, we can think generalising ordered locales $(X,\Leq)$ to triples $(\cat{C},J,J_T)$, where $J$ encodes the ordinary topology of $X$, and $J_T$ the causal coverage of $\Leq$. Then the compatibility condition (i$'$) can be extended to the condition that
		\begin{enumerate}
			\item[(i$''$)] if $R\in J(C)$ then $(\eta_C)_\ast(R)\in J_T(C)$.
		\end{enumerate}
	If $J$ is the canonical coverage of $\Opens X$ and $J_T=J^-$, then this corresponds to saying that if $\bigvee R = U$ then $\bigvee R \in\CovLeq^-(U)$, which is again just~\cref{lemma:properties of coverage from locale}(a).
\end{remark}

\begin{example}\label{example:causal coverage for equality as Grothendieck topology}
	In \cref{example:causal coverage for equality} we saw that in the ordered locale $(X,=)$, where equality is the causal order, the causal coverage relation $\Cov_{=}^-$ is just equality. Since $\Up = \id_{\Opens X} = \Down$, via \cref{theorem:causal grothendieck topology} the coverage relation $\Cov_{=}^-$ induces an ordinary Grothendieck topology on $\Opens X$. It is in fact just the canonical Grothendieck topology of the locale.
\end{example}

\section{Causal coverages in spacetimes}
In the following we describe the behaviour of $\CovLeq^\pm$ in spacetimes, and relate it to the intuition using monotone curves outlined in the~\nameref{section:introduction}. For spacetimes we get two distinct notions of coverages, depending on whether one works with the causal or chronological curves.

\begin{definition}\label{definition:spacetime causal coverages}
	In a smooth spacetime $M$, we define the following notion of coverage. For $A,U\in \Opens M$ we say \emph{$A$ covers $U$ from below with (causal) curves} and write $A\in\Covcaus^-(U)$ if $A\subseteq J^-(U)$ and:
	\begin{quote}
		for every future-directed causal curve $\gamma\colon [a,b]\to M$ with $\gamma(b)\in U$, either ${\im(\gamma)\cap A \neq \varnothing}$ or $J^-(\gamma(a))\cap A\neq\varnothing$.
	\end{quote}
	There is an analogous chronological definition. We say \emph{$A$ covers $U$ from below with timelike curves} and write $A\in\Covchron^-(U)$ if $A\subseteq I^-(U)$ and:
	\begin{quote}
		for every future-directed timelike curve $\gamma\colon [a,b]\to M$ with $\gamma(b)\in U$, either ${\im(\gamma)\cap A\neq \varnothing}$ or $I^-(\gamma(a))\cap A\neq\varnothing$.
	\end{quote}
\end{definition}
\noindent\begin{minipage}{0.64\textwidth}
	\begin{remarknumbered}\label{remark:causal cover contained in chronological cover}
		We note that since all timelike curves are causal, it follows immediately that \[\Covcaus^-(U)\subseteq \Covchron^-(U).\]
		The converse inclusion does not generally hold. \cref{figure:non-intersecting lightlike curve} shows a typical $A\in \CovLeq^-(U)$, together with a causal curve $\gamma$ that is lightlike as it crosses~$A$. Then $A\setminus\im(\gamma)\in \Covchron^-(U)$, but of course ${A\setminus\im(\gamma)\notin\Covcaus^-(U)}$, since $\gamma$ itself provides a counterexample.
	\end{remarknumbered}
\end{minipage}%
\hfill%
\begin{minipage}{.43\textwidth}\centering
	\definecolor{cffffdc}{RGB}{255,255,220}
\definecolor{ce0ffdc}{RGB}{224,255,220}
\definecolor{c640000}{RGB}{100,0,0}
\definecolor{c644700}{RGB}{100,71,0}
\definecolor{c0d6400}{RGB}{13,100,0}
\definecolor{cffebdc}{RGB}{255,235,220}
\definecolor{cefffff}{RGB}{239,255,255}
\definecolor{c001764}{RGB}{0,23,100}

\def \globalscale {1.000000}
\begin{tikzpicture}[y=.75pt, x=.75pt, yscale=\globalscale,xscale=\globalscale, every node/.append style={scale=\globalscale}, inner sep=0pt, outer sep=0pt]



	\path[draw=black,fill=cffffdc,line cap=butt,line join=miter,line 
	width=1.0pt,miter limit=4.0] (32.14, 101.69).. controls (32.359, 92.299) and 
	(51.352, 102.343) .. (60.215, 99.23).. controls (66.644, 96.973) and (70.759, 
	90.046) .. (77.223, 87.892).. controls (85.727, 85.057) and (98.049, 91.676) 
	.. (97.065, 99.23).. controls (95.821, 108.791) and (78.736, 106.065) .. 
	(69.114, 106.673).. controls (56.703, 107.458) and (31.851, 114.123) .. 
	(32.14, 101.69) -- cycle;

	\path[draw=black,line cap=butt,line join=miter,line width=0.5pt,miter 
	limit=4.0,dash pattern=on 0.5pt off 2.0pt] (32.14, 104.525) -- (4.55, 28.702);

	\path[draw=black,line cap=butt,line join=miter,line width=0.5pt,miter 
	limit=4.0,dash pattern=on 0.5pt off 2.0pt] (96.557, 100.941) -- (123.339, 
	26.934);

	\path[draw=black,fill=ce0ffdc,line width=1.0pt] (53.786, 62.38) -- (47.596, 
	45.372) -- (116.557, 45.372) -- (110.509, 62.38) -- cycle;

	\path[draw=black,fill=ce0ffdc,line width=1.0pt] (51.771, 62.38) -- (16.82, 
	62.423) -- (10.831, 45.372) -- (45.581, 45.377) -- cycle;

	\begin{scope}[decoration={
			markings,
			mark=between positions .1 and .9 step 50pt with {\arrow[line width = .75pt]{>}}}
		]
		\path[postaction = {decorate},draw=c640000,line cap=butt,line join=miter,line width=1.0pt,miter 
	limit=4.0] (40.373, 28.364) -- (67.2, 102.065);
	\end{scope}

	\path[fill=black,line cap=butt,line join=miter,line width=1.0pt,miter 
	limit=4.0] (67.096, 101.915) ellipse (1.738pt and 1.739pt);

	\node[text=c640000,line cap=butt,line join=miter,line width=1.0pt,miter 
	limit=4.0,scale=1.0,anchor=south west] (text22) at (47.248, 29.945){$\gamma$};

	\node[text=c644700,line cap=butt,line join=miter,line width=1.0pt,miter 
	limit=4.0,scale=1.0,anchor=south west] (text23) at (80.068, 93.567){$U$};

	\node[text=c0d6400,line cap=butt,line join=miter,line width=1.0pt,miter 
	limit=4.0,scale=1.0,anchor=south west] (text24) at (20.182, 49.559){$A$};

	\node[text=c0d6400,line cap=butt,line join=miter,line width=1.0pt,miter 
	limit=4.0,scale=1.0,anchor=south west] (text24-6) at (101.589, 49.559){$A$};

		%
		%
		%
		%
		%
		%
		%
		%
		%
		%
		%
		%
		%
		%
		%
		%
		%
		%
		%
		%
		%
		%
		%
		%
		%
		%
		%
	
\end{tikzpicture}
	
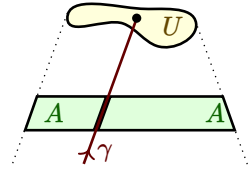
\captionof{figure}{Non-intersecting lightlike curve.}
	\label{figure:non-intersecting lightlike curve}
\end{minipage}

In this definition we made use of `bounded' causal and timelike curves, defined on closed intervals. In relativity theory it is common to make use of causal curves defined on half-open intervals $(a,b]$ that are \emph{inextendible} into the past. We think of these as causal curves that (when in reverse) propagate indefinitely into the past.
To make this rigorous, we adopt \cite[Definition~5.17]{landsman2021FoundationsGeneralRelativity}.

\begin{definition}\label{definition:past inextendible}
	A curve $\gamma\colon (a,b]\to M$ in a smooth spacetime is called \emph{past extendible} if the limit $\lim_{t\searrow a}\gamma(t)$ exists. Otherwise $\gamma$ is called \emph{past inextendible}.
\end{definition}


\begin{lemma}\label{lemma:past inextendible intersects then bounded intersects}
	In a smooth spacetime $M$, we have $A\in\Covcaus^-(U)$ as soon as ${A\subseteq J^-(U)}$ and:
	\begin{quote}
		every past inextendible causal curve landing in $U$ intersects $A$.
	\end{quote}
\end{lemma}
\begin{proof}
	Let $\gamma\colon [a,b]\to M$ be a causal curve landing in $U$. By~\cite[Theorem~2.5.7]{chrusciel2020ElementsCausality} we can extend $\gamma$ to a past inextendible causal curve $\overline{\gamma}\colon (c,b]\to M$ such that ${\overline{\gamma}|_{[a,b]} = \gamma}$. By assumption $\overline{\gamma}^{-1}(A)$ is open and non-empty, so since~$\mathbb{R}$ is regular there exists some closed interval $[n,m]\subseteq \overline{\gamma}^{-1}(A)$ with $n<m$. If $n\geq a$ then $\gamma$ already intersects $A$. If $n\leq a$ then $\overline{\gamma}(n)\in A \cap J^-(\gamma(a))$. Since $\gamma$ was arbitrary this gives $A\in\Covcaus^-(U)$.
	%
	%
\end{proof}
\noindent\begin{minipage}{0.64\textwidth}
	\begin{remarknumbered}\label{remark:difference causal coverage bounded or inextendible}
		An analogous result holds for the chronological coverage. The converse of \cref{lemma:past inextendible intersects then bounded intersects} is not generally true. In \cref{figure:non-intersecting inextendible curve} we have a past inextendible curve $\gamma$ in Minkowski space with one point removed that does not intersect~$A$. Note however that $A\in\Covcaus^-(U)$ does hold, since for every restriction $\gamma|_{[c,b]}$ to a bounded curve we do get $J^-(\gamma(c))\cap A\neq\varnothing$.
	\end{remarknumbered}
\end{minipage}%
\hfill%
\begin{minipage}{.45\textwidth}\centering
	\definecolor{cffffdc}{RGB}{255,255,220}
\definecolor{ce0ffdc}{RGB}{224,255,220}
\definecolor{c640000}{RGB}{100,0,0}
\definecolor{c644700}{RGB}{100,71,0}
\definecolor{c0d6400}{RGB}{13,100,0}
\definecolor{cff0909}{RGB}{255,9,9}
\definecolor{cffebdc}{RGB}{255,235,220}
\definecolor{cefffff}{RGB}{239,255,255}
\definecolor{c001764}{RGB}{0,23,100}

\def \globalscale {1.000000}
\begin{tikzpicture}[y=.85pt, x=.8pt, yscale=\globalscale,xscale=\globalscale, every node/.append style={scale=\globalscale}, inner sep=0pt, outer sep=0pt]
	\path[draw=black,fill=cffffdc,line cap=butt,line join=miter,line 
	width=1.0pt,miter limit=4.0] (56.027, 113.047).. controls (47.549, 111.504) 
	and (31.296, 115.619) .. (30.89, 107.012).. controls (30.512, 98.997) and 
	(45.721, 100.355) .. (53.664, 99.216).. controls (65.543, 97.511) and (83.846,
	89.25) .. (89.632, 100.827).. controls (92.724, 107.012) and (87.288, 
	115.488) .. (80.467, 116.102).. controls (71.163, 116.939) and (64.104, 
	114.517) .. (56.027, 113.047) -- cycle;

	\path[draw=black,line cap=butt,line join=miter,line width=0.5pt,miter 
	limit=4.0,dash pattern=on 0.5pt off 2.0pt] (30.89, 107.012) -- (6.357, 39.591);

	\path[draw=black,line cap=butt,line join=miter,line width=0.5pt,miter 
	limit=4.0,dash pattern=on 0.5pt off 2.0pt] (90.267, 105.21) -- (114.082, 
	39.404);

	\path[draw=black,fill=ce0ffdc,line width=1.0pt] (104.682, 65.379) -- (15.946, 
	65.379) -- (10.045, 49.917) -- (110.26, 49.917) -- cycle;

	\begin{scope}[decoration={
			markings,
			mark=between positions .5 and .8 step 20pt with {\arrow[line width = .75pt]{>}}}
		]
		\path[postaction = {decorate},draw=c640000,line cap=butt,line join=miter,line width=1.0pt,miter 
		limit=4.0] (71.595, 72.056).. controls (71.969, 80.562) and (68.385, 88.849) 
		.. (68.106, 97.359).. controls (68.005, 100.436) and (68.765, 106.57) .. 
		(68.765, 106.57);
	\end{scope}

	\path[fill=black,line cap=butt,line join=miter,line width=1.0pt,miter 
	limit=4.0] (68.765, 105.891) ellipse (1.546pt and 1.546pt);

	\node[text=c640000,line cap=butt,line join=miter,line width=1.0pt,miter 
	limit=4.0,scale=1.0,anchor=south west] (text22) at (73.832, 84.015){$\gamma$};

	\node[text=c644700,line cap=butt,line join=miter,line width=1.0pt,miter 
	limit=4.0,scale=1.0,anchor=south west] (text23) at (74.675, 103.482){$U$};

	\node[text=c0d6400,line cap=butt,line join=miter,line width=1.0pt,miter 
	limit=4.0,scale=1.0,anchor=south west] (text24) at (16.246, 54.218){$A$};

	\path[draw=cff0909, fill = white,line width=2pt] 
	(71.622, 70.276) ellipse (3.055pt and 3.055pt);

	\node[text=cff0909,line width=0.601pt,dash pattern=on 6.009pt off 
	1.803pt,anchor=south west] (text2) at (28.442, 79.264){$\text{point}$};

	\node[text=cff0909,line width=0.601pt,dash pattern=on 6.009pt off 
	1.803pt,anchor=south west] (text3) at (22.639, 73.151){$\text{removed}$};

		%
		%
		%
		%
		%
		%
		%
		%
		%
		%
		%
		%
		%
		%
		%
		%
		%
		%
		%
		%
		%
		%
		%
		%
		%
		%
		%
	
\end{tikzpicture}
	
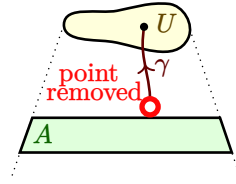
\captionof{figure}{Non-intersecting inextendible curve.}
	\label{figure:non-intersecting inextendible curve}
\end{minipage}

In the following we compare the localic and curve-wise definitions of coverages in a spacetime. First we need a small lemma.
\begin{lemma}\label{lemma:path restriction preserves points}
	Let $(S,\leq)$ be an ordered space with open cones, and consider a path $(p_n)_{n=0}^N$ and a chain $x_0\leq \cdots \leq x_N$ with $x_n\in p_n$ for all $n$. If $x_N\in W\subseteq U$ for $W\in\Opens S$, then $x_n\in (p|_W)_n$ for all $n$.
\end{lemma}
\begin{proof}
	Since $x_N\in W = (p|_W)_N$, for an induction proof it suffices to show that if $x_{n+1}\in (p|_W)_{n+1}$ then $x_n \in (p|_W)_n$. But $x_n\leq x_{n+1}$, so using the open cone condition:
	\[
	x_n \in p_n \cap \down x_{n+1}
	\subseteq
	p_n \cap \Down (p|_W)_{n+1}=: (p|_W)_n.\qedhere
	\] 
\end{proof}

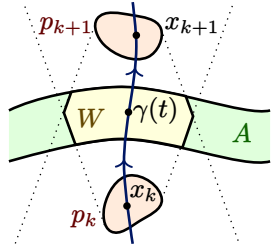
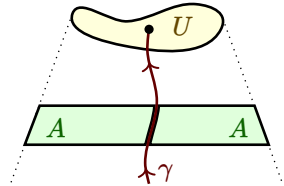
\begin{figure}[b]\centering
	\begin{subfigure}[b]{0.4\textwidth}\centering
		\definecolor{ce0ffdc}{RGB}{224,255,220}
\definecolor{cffebdc}{RGB}{255,235,220}
\definecolor{cffffdc}{RGB}{255,255,220}
\definecolor{c0d6400}{RGB}{13,100,0}
\definecolor{c644700}{RGB}{100,71,0}
\definecolor{c640000}{RGB}{100,0,0}
\definecolor{c001764}{RGB}{0,23,100}
\definecolor{cefffff}{RGB}{239,255,255}

\def \globalscale {1.000000}
\begin{tikzpicture}[y=.85pt, x=.85pt, yscale=\globalscale,xscale=\globalscale, every node/.append style={scale=\globalscale}, inner sep=0pt, outer sep=0pt]
  \begin{scope}[shift={(-271.503, 73.658)}]
    \path[fill=ce0ffdc,line cap=butt,line join=miter,line width=1.0pt,miter 
  limit=4.0] (394.08, 8.565).. controls (382.741, 5.73) and (348.989, 19.619) ..
   (326.067, 18.925).. controls (309.639, 18.427) and (289.198, 14.234) .. 
  (277.86, 8.565) -- (277.86, -14.112).. controls (277.86, -14.112) and (309.66,
   -4.726) .. (326.049, -4.251).. controls (348.953, -3.588) and (381.569, 
  -17.508) .. (394.08, -14.112) -- cycle;

    \path[draw=black,line cap=butt,line join=miter,line width=1.001pt,miter 
  limit=4.0] (277.86, -14.112).. controls (277.86, -14.112) and (309.66, -4.726)
   .. (326.049, -4.251).. controls (348.953, -3.588) and (381.569, -17.508) .. 
  (394.08, -14.112)(394.08, 8.565).. controls (382.741, 5.73) and (348.989, 
  19.619) .. (326.067, 18.925).. controls (309.639, 18.427) and (289.198, 
  14.234) .. (277.86, 8.565);

    \path[draw=black,fill=cffebdc,line cap=butt,line join=miter,line 
  width=1.0pt,miter limit=4.0] (317.109, 49.379).. controls (314.873, 46.189) 
  and (317.235, 40.841) .. (319.944, 38.041).. controls (324.726, 33.096) and 
  (333.552, 29.464) .. (339.786, 32.372).. controls (343.988, 34.331) and 
  (347.085, 40.733) .. (345.455, 45.074).. controls (343.803, 49.475) and 
  (337.29, 50.218) .. (332.653, 50.989).. controls (327.515, 51.844) and 
  (320.099, 53.645) .. (317.109, 49.379) -- cycle;

    \path[draw=black,fill=cffebdc,line cap=butt,line join=miter,line 
  width=1.0pt,miter limit=4.0] (320.379, -33.955).. controls (319.526, -37.756) 
  and (319.778, -43.458) .. (323.214, -45.294).. controls (327.183, -47.415) and
   (332.1, -43.11) .. (335.923, -40.735).. controls (339.845, -38.299) and 
  (346.184, -35.727) .. (345.891, -31.12).. controls (345.511, -25.148) and 
  (337.435, -19.667) .. (331.495, -20.394).. controls (325.694, -21.105) and 
  (321.659, -28.252) .. (320.379, -33.955) -- cycle;

    \path[draw=black,line cap=butt,line join=miter,line width=0.5pt,miter 
  limit=4.0,dash pattern=on 0.5pt off 2.0pt] (345.891, -31.12) -- (377.875, 
  56.754);

    \path[draw=black,line cap=butt,line join=miter,line width=0.5pt,miter 
  limit=4.0,dash pattern=on 0.5pt off 2.0pt] (320.379, -42.459) -- (284.269, 
  56.754);

    \path[draw=black,line cap=butt,line join=miter,line width=0.5pt,miter 
  limit=4.0,dash pattern=on 0.5pt off 2.0pt] (317.109, 49.379) -- (280.694, 
  -50.669);

    \path[draw=black,line cap=butt,line join=miter,line width=0.5pt,miter 
  limit=4.0,dash pattern=on 0.5pt off 2.0pt] (345.455, 45.074) -- (380.41, 
  -50.963);

    \path[draw=black,fill=cffffdc,line cap=butt,line join=miter,line 
  width=1.0pt,miter limit=4.0] (356.543, 14.943).. controls (346.146, 17.293) 
  and (335.407, 19.308) .. (326.342, 19.034).. controls (319.796, 18.836) and 
  (312.617, 18.043) .. (305.627, 16.812) -- (302.305, 7.528) -- (307.491, 
  -6.721).. controls (314.315, -5.352) and (321.102, -4.293) .. (326.324, 
  -4.141).. controls (335.043, -3.889) and (345.164, -5.753) .. (355.053, 
  -7.967) -- (359.815, 5.953) -- cycle;

    \node[text=c0d6400,line cap=butt,line join=miter,line width=1.0pt,miter 
  limit=4.0,anchor=south west] (text12) at (377.072, -5.609){$A$};

    \node[text=c644700,line cap=butt,line join=miter,line width=1.0pt,miter 
  limit=4.0,anchor=south west] (text13) at (308.056, 1.221){$W$};

    \node[text=c640000,line cap=butt,line join=miter,line width=1.0pt,miter 
  limit=4.0,anchor=south west] (text14) at (305.761, -43.311){$p_k$};

    \node[text=c640000,line cap=butt,line join=miter,line width=1.0pt,miter 
  limit=4.0,anchor=south west] (text15) at (291.657, 42.089){$p_{k+1}$};

	\begin{scope}[decoration={
	markings,
	mark=between positions .35 and .95 step 35pt with {\arrow[line width = .75pt]{>}}}
	]
    \path[postaction = {decorate},draw=c001764,line cap=butt,line join=miter,line width=1.0pt,miter 
  limit=4.0] (332.933, -50.818).. controls (332.933, -50.818) and (329.063, 
  -27.117) .. (329.044, -15.16).. controls (329.015, 2.294) and (334.342, 
  19.458) .. (334.552, 36.911).. controls (334.633, 43.564) and (332.933, 
  56.805) .. (332.933, 56.805);
	\end{scope}

    \path[fill=black,line cap=butt,line join=miter,line width=1.0pt,miter 
  limit=4.0] (334.552, 42.026) circle (1.417pt);

    \path[fill=black,line cap=butt,line join=miter,line width=1.0pt,miter 
  limit=4.0] (330.382, -33.967) circle (1.417pt);

    \path[fill=black,line cap=butt,line join=miter,line width=1.0pt,miter 
  limit=4.0] (331.059, 7.795) circle (1.417pt);

    \node[text=black,line cap=butt,line join=miter,line width=1.0pt,miter 
  limit=4.0,anchor=south west] (text16) at (347.259, 42.065){$x_{k+1}$};

    \node[text=black,line cap=butt,line join=miter,line width=1.0pt,miter 
  limit=4.0,anchor=south west] (text16-8) at (331.657, -31.621){$x_k$};

    \node[text=black,line cap=butt,line join=miter,line width=1.0pt,miter 
  limit=4.0,anchor=south west] (text17) at (333.126, 2.794){$\gamma(t)$};

  \end{scope}
%
%
%
%
%
%
%
%
%
%
%
%
%
%
%
%
%
%
%
%
%
%
%
%
%
%
%

\end{tikzpicture}
		\caption{Construction of $W$.}
	\end{subfigure}\hfil
	\begin{subfigure}[b]{0.4\textwidth}\centering
		\definecolor{cffffdc}{RGB}{255,255,220}
\definecolor{ce0ffdc}{RGB}{224,255,220}
\definecolor{c640000}{RGB}{100,0,0}
\definecolor{c644700}{RGB}{100,71,0}
\definecolor{c0d6400}{RGB}{13,100,0}
\definecolor{cffebdc}{RGB}{255,235,220}
\definecolor{cefffff}{RGB}{239,255,255}
\definecolor{c001764}{RGB}{0,23,100}

\def \globalscale {1.000000}
\begin{tikzpicture}[y=.85pt, x=.85pt, yscale=\globalscale,xscale=\globalscale, every node/.append style={scale=\globalscale}, inner sep=0pt, outer sep=0pt]
  \path[draw=black,fill=cffffdc,line cap=butt,line join=miter,line 
  width=1.0pt,miter limit=4.0] (56.479, 101.047).. controls (48.295, 100.625) 
  and (34.949, 112.226) .. (32.14, 104.525).. controls (29.039, 96.019) and 
  (46.299, 91.894) .. (55.181, 90.153).. controls (69.462, 87.355) and (91.695, 
  84.55) .. (98.203, 97.569).. controls (101.68, 104.525) and (94.726, 108.003) 
  .. (87.055, 108.693).. controls (76.592, 109.634) and (66.971, 101.588) .. 
  (56.479, 101.047) -- cycle;

  \path[draw=black,line cap=butt,line join=miter,line width=0.5pt,miter 
  limit=4.0,dash pattern=on 0.5pt off 2.0pt] (32.14, 104.525) -- (4.55, 28.702);

  \path[draw=black,line cap=butt,line join=miter,line width=0.5pt,miter 
  limit=4.0,dash pattern=on 0.5pt off 2.0pt] (98.917, 102.499) -- (125.699, 
  28.492);

  \path[draw=black,fill=ce0ffdc,line width=1.0pt] (70.53, 63.48).. controls 
  (70.08, 57.601) and (67.871, 51.885) .. (66.506, 46.091) -- (119.443, 46.091) 
  -- (113.115, 63.48) -- cycle;

  \path[draw=black,fill=ce0ffdc,line width=1.0pt] (68.574, 63.48) -- (17.515, 
  63.48) -- (10.878, 46.091) -- (64.55, 46.091).. controls (65.915, 51.885) and 
  (68.125, 57.601) .. (68.574, 63.48) -- cycle;

\begin{scope}[decoration={
markings,
mark=between positions .1 and .95 step 41pt with {\arrow[line width = .75pt]{>}}}
]
  \path[postaction = {decorate},draw=c640000,line cap=butt,line join=miter,line width=1.0pt,miter 
  limit=4.0] (66.193, 28.702).. controls (66.193, 28.702) and (64.454, 34.809) 
  .. (64.455, 37.943).. controls (64.456, 47.304) and (69.654, 56.178) .. 
  (69.656, 65.538).. controls (69.658, 73.123) and (65.886, 80.328) .. (65.452, 
  87.901).. controls (65.254, 91.357) and (66.193, 98.259) .. (66.193, 98.259);
\end{scope}

  \path[fill=black,line cap=butt,line join=miter,line width=1.0pt,miter 
  limit=4.0] (66.193, 97.496) ellipse (1.738pt and 1.739pt);

  \node[text=c640000,line cap=butt,line join=miter,line width=1.0pt,miter 
  limit=4.0,scale=1.0,anchor=south west] (text22) at (70.313, 30.035){$\gamma$};

  \node[text=c644700,line cap=butt,line join=miter,line width=1.0pt,miter 
  limit=4.0,scale=1.0,anchor=south west] (text23) at (76.634, 94.787){$U$};

  \node[text=c0d6400,line cap=butt,line join=miter,line width=1.0pt,miter 
  limit=4.0,scale=1.0,anchor=south west] (text24) at (20.182, 49.559){$A$};

  \node[text=c0d6400,line cap=butt,line join=miter,line width=1.0pt,miter 
  limit=4.0,scale=1.0,anchor=south west] (text24-6) at (101.589, 49.559){$A$};

%
%
%
%
%
%
%
%
%
%
%
%
%
%
%
%
%
%
%
%
%
%
%
%
%
%
%

\end{tikzpicture}
		\caption{Counterexample of converse.}
	\end{subfigure}
	\caption{Illustrations for the proof of \cref{proposition: localic coverage contains chronological coverage}.}
	\label{figure:causal coverage in spacetime}
\end{figure}
\begin{proposition}\label{proposition: localic coverage contains chronological coverage}
	In any smooth spacetime: $\Covchron^-(U)\subseteq \CovLeq^-(U)$.
\end{proposition}
\begin{proof}
	Take a curve-wise covering region $A\in \Covchron^-(U)$, and take a localic path ${p=(p_n)_{n=0}^N}$ landing in~$U$. By definition, each step $p_n$ is non-empty. Starting with $x_N\in p_N\subseteq U$, this allows us to inductively construct a chain of chronologically related points:
	\[
	\begin{tikzcd}[cramped,row sep=tiny]
		{x_0} & {x_1} & \cdots & {x_{N-1}} & {x_N} \\
		{p_0} & {p_1} & \cdots & {p_{N-1}} & {p_N.}
		\arrow["\chron"{marking, allow upside down}, draw=none, from=1-1, to=1-2]
		\arrow["\chron"{marking, allow upside down}, draw=none, from=1-2, to=1-3]
		\arrow["\chron"{marking, allow upside down}, draw=none, from=1-3, to=1-4]
		\arrow["\chron"{marking, allow upside down}, draw=none, from=1-4, to=1-5]
		\arrow["\in"{marking, allow upside down}, draw=none, from=1-1, to=2-1]
		\arrow["\in"{marking, allow upside down}, draw=none, from=1-2, to=2-2]
		\arrow["\Leq"{marking, allow upside down}, draw=none, from=2-1, to=2-2]
		\arrow["\Leq"{marking, allow upside down}, draw=none, from=2-2, to=2-3]
		\arrow["\in"{marking, allow upside down}, draw=none, from=1-5, to=2-5]
		\arrow["\in"{marking, allow upside down}, draw=none, from=1-4, to=2-4]
		\arrow["\Leq"{marking, allow upside down}, draw=none, from=2-3, to=2-4]
		\arrow["\Leq"{marking, allow upside down}, draw=none, from=2-4, to=2-5]
	\end{tikzcd}
	\]
	By definition of the chronology relation, from this chain we obtain a timelike curve $\gamma \colon [a,b]\to M$ that starts at $x_0$, passes each $x_n$, and ends in $x_N$. By hypothesis, there either exists $t\in [a,b]$ such that $\gamma(t)\in A$, or $I^-(\gamma(a))\cap A \neq\varnothing$. The proof in the latter case is analogous and simpler, so we leave it to the reader. We can therefore assume there exists an index~$k$ such that $x_k\chron \gamma(t)\chron x_{k+1}$ with $\gamma(t)\in A$, and this gives a non-empty open
	\[
	\gamma(t)\in A\cap \Up p_k\cap \Down p_{k+1}=: W.
	\]
	See \cref{figure:causal coverage in spacetime}(a) for intuition. Since $p_k\Leq \Up p_k\cap \Down p_{k+1}\Leq p_{k+1}$, by parallel orderedness (via \cref{lemma:wedge iff strong wedge}) we obtain causal relations
	\[
	p_k\cap \Down W\Leq W\Leq p_{k+1}\cap \Up W.
	\]
	We will use this fact to refine $p$ into a path that inhabits $W\subseteq A$. This refinement is constructed as follows. First, split the original path in two by defining ${\hat{p}= (p_n)_{n=0}^k}$ and $\check{p}= (p_n)_{n=k+1}^N$, and take the refinements $\hat{p}|_{p_k\cap \Down W}$ and $\check{p}|^{p_{k+1}\cap \Up W}$. The desired path is defined via the concatenation of these refinements, interjected with the single-step path $W$ in the middle:
	\[
	q^{x_N}:= \check{p}|^{p_{k+1}\cap \Up W}\cdot W \cdot \hat{p}|_{p_k\cap \Down W}.
	\]
	This is well-defined by the previous equation. Using \cref{lemma:properties of refinement relation}(c) and \cref{lemma:path restriction preserves refinement} we see that $q^{x_N}\refines (\check{p}\cdot \hat{p})|_{q^{x_N}_\top}=p|_{q^{x_N}_\top}$,
	and by construction it inhabits $A$. Finally, since in particular we have $W\ni\gamma(t)\caus x_{k+1}\in p_{k+1}$ it follows by (the dual of) \cref{lemma:path restriction preserves points} that ${x_N\in q^{x_N}}$. Repeating this construction for arbitrary $x\in p_N$ gives a family $(q^x)_{x\in p_N}$ of refinements of $p$ that inhabit $A$, such that ${\bigvee_{x\in p_N}q^x_\top = p_N}$. In other words, $(q^x)_{x\in p_N}\refines p$ is a local past refinement, and so ${A\in\CovLeq^-(U)}$. 
\end{proof}

\begin{remarknumbered}\label{remark:difference causal coverage localic or curve}
	The converse inclusion of \cref{proposition: localic coverage contains chronological coverage} does not generally hold. \cref{figure:causal coverage in spacetime}(b) provides a simple counterexample. Here we have a past inextendible timelike curve $\gamma$ in Minkowski space terminating in~$U$. The covering region is constructed by starting with a typical covering region $A\in\CovLeq^-(U)$, and cutting out the image $\im(\gamma)$. Since no localic path $p$ can pass through this infinitesimal hole, in the localic setting we still get that~$A\setminus \im(\gamma)\in\CovLeq^-(U)$.	
	From a philosophical point of view it can be argued that the localic definition is more appropriate, since in the definition of $\Covcaus^-$ or $\Covchron^-$ the only paths that can dodge $A\setminus\im(\gamma)$ are those where the information it carries is compacted into an infinitesimal space, which could be deemed physically unreasonable.
\end{remarknumbered}

\section{Abstract causal coverages}
Before moving on to domains of dependence, in this intermediate section we sketch a tentative definition where instead of the causal ordering $\Leq$ on a locale, the coverage relations $\CovLeq^\pm$ are treated as fundamental. The following list of axioms are supposed to be abstract versions of the properties in \cref{lemma:properties of coverage from locale,lemma:coverage locale satisfies Cov-V strong}. The challenge here is to state axioms for abstract causal coverages $\Cov^\pm$ without referring to the pre-existing structure of $\Leq$ or the localic cones $\Up$ and $\Down$.

\begin{definition}\label{definition:causal coverage abstract}
	Let $X$ be a locale. A \emph{causal coverage} on $X$ consists of two functions
	\[
	\Cov^\pm\colon \Opens X\longrightarrow\Powerset(\Opens X)
	\]
	satisfying the following axioms:
	\begin{enumerate}[label = (C\arabic*)]
		\item\label{axiom:Cov-unit} $U\in\Cov^\pm(U)$;
		\item\label{axiom:Cov-joins} axiom~\eqref{axiom:Cov-V strong} holds;
		\item\label{axiom:Cov-transitivity} if $B\in \Cov^\pm(A)$ and $A\in\Cov^\pm(U)$ then $B\in\Cov^\pm(U)$;
		\item\label{axiom:Cov-ideal} if $A,B\in \Cov^\pm(U)$ and $A\sqleq C\sqleq B$, then $C\in\Cov^\pm(U)$;
		\item\label{axiom:Cov-flip} if $A\in\Cov^\pm(U)$ then there exists $W\in\Cov^\mp(A)$ such that $U\sqleq W$.
	\end{enumerate}
	A \emph{causal site}\footnote{Now not to be confused with the homonymous but distinct notion in \cite{christensen2005CausalSitesQuantum}.} $(X,\Cov^\pm)$ is a locale $X$ equipped with a causal coverage $\Cov^\pm$.
\end{definition}

These axioms should be directly compared to the properties in \cref{lemma:properties of coverage from locale,lemma:coverage locale satisfies Cov-V weak,lemma:coverage locale satisfies Cov-V strong}. Indeed, we have the following.

\begin{proposition}
	If $(X,\Leq)$ is a parallel ordered locale satisfying~\eqref{axiom:LV}, then $\CovLeq^\pm$ from \cref{definition:causal coverage from ordered locale} define a causal coverage on $X$.
\end{proposition}
\begin{proof}
	Axioms~\ref{axiom:Cov-unit}--\ref{axiom:Cov-transitivity} follow directly from \cref{lemma:properties of coverage from locale,lemma:coverage locale satisfies Cov-V strong}. Axiom~\ref{axiom:Cov-ideal} holds since if $A,B\in\CovLeq^-(U)$ and $A\sqleq C\sqleq B$, then $C\sqleq B\sqleq \Down U$, and any path landing in $U$ can be refined to inhabit $A\sqleq C$. Thus $C\in \CovLeq^-(U)$. Lastly, for axiom~\ref{axiom:Cov-flip}, if $A\in\CovLeq^-(U)$ then by \cref{lemma:properties of coverage from locale}(e) we get $A\Leq U$, so we get ${U\sqleq \Up A\in\CovLeq^+(A)}$.
\end{proof}

In the following we show how we can recover the structure of an ordered locale from a causal site.
\begin{lemma}\label{lemma:causal coverage determines monads}
	If $\Cov^\pm\colon \Opens X\to\Powerset(\Opens X)$ satisfies axioms~\ref{axiom:Cov-unit}--\ref{axiom:Cov-transitivity} of a causal coverage, then the following functions define join-preserving monads on $\Opens X$:
	\[
	L^\pm\colon \Opens X\longrightarrow \Opens X;\qquad U\longmapsto \bigvee \Cov^\pm(U).
	\]
\end{lemma}
\begin{proof}
	Axiom~\ref{axiom:Cov-unit} immediately gives the unit $U\sqleq L^\pm (U)$. Axiom~\eqref{axiom:Cov-V strong} gives that $L^\pm (U)\in \Cov^\pm (U)$. In turn, it follows that $L^\pm \circ L^\pm (U)\in \Cov^\pm(L^\pm (U))$, so by axiom~\ref{axiom:Cov-transitivity} we get that $L^\pm \circ L^\pm (U)\in \Cov^\pm (U)$. This provides the multiplication $L^\pm \circ L^\pm (U)\sqleq L^\pm (U)$ of the monad. Lastly, if $U\sqleq V$ and $A\in \Cov^\pm(U)$, using~\ref{axiom:Cov-unit} and~\eqref{axiom:Cov-V strong} we get $A\sqleq A\vee V\in \Cov^\pm(U\vee V) = \Cov^\pm(V)$, which proves that $L^\pm$ are monotone.
	
	Finally, to prove $L^\pm$ preserve all joins, take opens $(U_i)_{i\in I}$. Since $L^\pm$ are monotone, we only need to show $L^\pm\left(\bigvee_{i\in I} U_i\right)\sqleq \bigvee_{i\in I} L^\pm (U_i)$. For that, take $A\in \Cov^\pm (\bigvee_{i\in I} U_i)$. By~\eqref{axiom:Cov-V strong} we can find $A_i\in \Cov^\pm (U_i)$ for each $i\in I$ such that $A=\bigvee_{i\in I} A_i\sqleq \bigvee_{i\in I} L^\pm (U_i)$, so we are done.
\end{proof}

The open region $L^\pm (V)$ is also called the \emph{future/past region of influence} of $V$. Since they are join-preserving monads, we think of $L^\pm$ as a different incarnation of the localic cones of an ordered locale.

We now show that they also relate to $\Cov^\pm$ analogously to how $\Up$ and $\Down$ relate to $\CovLeq^\pm$. First, $U\in \Cov^\pm(U)$ is baked in as~\ref{axiom:Cov-unit}, and $L^\pm(U)\in\Cov^\pm(U)$ follows by~\eqref{axiom:Cov-V strong}.

The following is the abstract analogue of \cref{lemma:properties of coverage from locale}(d).

\begin{lemma}\label{lemma:Cov-pullbacks}
	Let $\Cov^\pm$ be a causal coverage on $X$. If $A\in \Cov^\pm(U)$ and $W\sqleq U$, then $A\wedge L^\pm (W)\in \Cov^\pm(W)$.
\end{lemma}
\begin{proof}
	Note that $U= W\vee U$, so $A\in \Cov^\pm(W\vee U)$, and hence~\eqref{axiom:Cov-V strong} gives regions $B\in\Cov^\pm(W)$ and $C\in \Cov^\pm(W)$ such that $A= B\vee C$. Now observe that $B\sqleq A\wedge L^\pm (W) \sqleq L^\pm (W)\in \Cov^\pm(W)$, so the result follows by axiom~\ref{axiom:Cov-ideal}.
\end{proof}

The following shows that if the causal coverage comes from an ordered locale, this is literally the case.

\begin{proposition}\label{proposition:ordered locale from canonical coverage}
	If $(X,\Leq)$ is a parallel ordered locale satisfying~\eqref{axiom:LV}, then $\Leq$ is recovered by the monads induced by the causal coverage $\CovLeq^\pm$. 
\end{proposition}
\begin{proof}
	By \cref{lemma:properties of coverage from locale}(b) it follows that $L^-(U)= \bigvee\CovLeq^-(U) = \Down U$, and similarly $L^+ (U) = \Up U$. Hence the preorder induced by $L^\pm$ on $\Opens X$ is equal to $\Leq$ by~\eqref{axiom:cones give order}.
\end{proof}

We currently do not know if the converse to \cref{proposition:ordered locale from canonical coverage} also holds: if we start with an abstract causal coverage $\Cov^\pm$ on $X$ and induce the ordered locale $(X,\Leq)$ from the localic cones $L^\pm$, does $\CovLeq^\pm$ have to equal $\Cov^\pm$? We conjecture that, with suitable definitions, there is a correspondence between categories of ordered locales and causal sites, analogous to the adjunction between frames and sites in~\cite{ball2014ExtendingSemilatticesFrames}. This is left to future work.

	%
	%

\section{Domains of dependence}
\label{section:domains of dependence}
In this section we study an abstract version of \emph{domains of dependence} \cite{geroch1970DomainofDependence} in the setting of causal coverages. The idea is that, given a region $A\in\Opens X$, the future domain of dependence $D^+(A)$ is the largest region that is covered from below by~$A$. Thinking of information flows, it is the largest region where all information is determined by $A$.

\begin{definition}\label{definition:domain of dependence localic}
	Let $\Cov^\pm$ be a causal coverage on a locale $X$. The \emph{future/past domain of dependence} of $A\in\Opens X$ is defined as
	\[
	D^\pm (A) :=\bigvee\left\{V\in\Opens X: A\in\Cov^\mp(V)\right\}.
	\]
	By~\eqref{axiom:Cov-V strong}, the future/past domain of dependence $D^\pm(A)$ is uniquely characterised as the largest open region that is covered from below/above by $A$. See \cref{figure:typical D+} for intuition. In particular we have
	\[
	A\in\Cov^\mp \left(D^\pm(A) \right).
	\]
\end{definition}
\begin{figure}[b]\centering
	\begin{subfigure}[b]{0.4\textwidth}\centering
		\definecolor{cffebdc}{RGB}{255,235,220}
\definecolor{ce0ffdc}{RGB}{224,255,220}
\definecolor{c0d6400}{RGB}{13,100,0}
\definecolor{c640000}{RGB}{100,0,0}
\definecolor{cefffff}{RGB}{239,255,255}
\definecolor{cffffdc}{RGB}{255,255,220}
\definecolor{c001764}{RGB}{0,23,100}
\definecolor{c644700}{RGB}{100,71,0}

\def \globalscale {1.000000}
\begin{tikzpicture}[y=.75pt, x=.75pt, yscale=\globalscale,xscale=\globalscale, every node/.append style={scale=\globalscale}, inner sep=0pt, outer sep=0pt]
	\path[draw=black,fill=cffebdc,line cap=butt,line join=miter,line 
	width=1.001pt,miter limit=4.0] (29.4, 33.178) -- (67.873, 99.23) -- (107.055, 
	32.272) -- (83.925, 28.666) -- cycle;

	\path[draw=black,fill=ce0ffdc,line cap=butt,line join=miter,line 
	width=1.0pt,miter limit=4.0] (57.38, 36.868).. controls (48.885, 36.483) and 
	(38.767, 41.84) .. (31.869, 36.868).. controls (29.445, 35.121) and (28.07, 
	31.192) .. (29.034, 28.364).. controls (30.581, 23.826) and (36.297, 21.406) 
	.. (40.972, 20.34).. controls (46.932, 18.981) and (53.012, 23.144) .. 
	(59.123, 22.96).. controls (64.746, 22.791) and (70.027, 19.442) .. (75.652, 
	19.549).. controls (81.255, 19.656) and (86.565, 22.117) .. (91.977, 23.572)..
	controls (97.485, 25.054) and (107.734, 22.7) .. (108.404, 28.364).. controls
	(109.298, 35.929) and (96.005, 38.083) .. (88.562, 39.703).. controls 
	(78.364, 41.921) and (67.806, 37.341) .. (57.38, 36.868) -- cycle;

	\node[text=c0d6400,line cap=butt,line join=miter,line width=1.0pt,miter 
	limit=4.0,anchor=south west] (text1) at (37.538, 25.529){$A$};

	\node[text=c640000,line cap=butt,line join=miter,line width=1.0pt,miter 
	limit=4.0,anchor=south west] (text2) at (49, 49){$D^+(A)$};

		%
		%
		%
		%
		%
		%
		%
		%
		%
		%
		%
		%
		%
		%
		%
		%
		%
		%
		%
		%
		%
		%
		%
		%
		%
		%
		%
	
\end{tikzpicture}
		\caption{Minkowski-type space.}
	\end{subfigure}
	\begin{subfigure}[b]{0.4\textwidth}\centering
		\definecolor{cffebdc}{RGB}{255,235,220}
\definecolor{ce0ffdc}{RGB}{224,255,220}
\definecolor{c0d6400}{RGB}{13,100,0}
\definecolor{c640000}{RGB}{100,0,0}
\definecolor{cefffff}{RGB}{239,255,255}
\definecolor{cffffdc}{RGB}{255,255,220}
\definecolor{c001764}{RGB}{0,23,100}
\definecolor{c644700}{RGB}{100,71,0}

\def \globalscale {1.000000}
\begin{tikzpicture}[y=.75pt, x=.75pt, yscale=\globalscale,xscale=\globalscale, every node/.append style={scale=\globalscale}, inner sep=0pt, outer sep=0pt]
	\path[draw=black,fill=cffebdc,line cap=butt,line join=miter,line 
	width=1.0pt,miter limit=4.0] (28.037, 90.726) -- (28.11, 25.529) -- (101.56, 
	31.199) -- (101.426, 90.841) -- (101.426, 90.841);

	\path[draw=black,fill=ce0ffdc,line cap=butt,line join=miter,line 
	width=1.0pt,miter limit=4.0] (77.223, 36.868).. controls (67.774, 36.868) and 
	(58.132, 38.768) .. (48.876, 36.868).. controls (41.827, 35.421) and (32.995, 
	34.371) .. (29.034, 28.364).. controls (27.716, 26.366) and (27.292, 22.825) 
	.. (29.034, 21.184).. controls (36.345, 14.292) and (49.079, 22.707) .. 
	(59.123, 22.96).. controls (69.259, 23.215) and (79.863, 19.665) .. (89.539, 
	22.695).. controls (94.223, 24.162) and (101.456, 26.291) .. (101.56, 
	31.199).. controls (101.626, 34.287) and (97.154, 35.869) .. (94.231, 
	36.868).. controls (88.866, 38.701) and (82.892, 36.868) .. (77.223, 36.868) 
	-- cycle;

	\node[text=c0d6400,line cap=butt,line join=miter,line width=1.0pt,miter 
	limit=4.0,anchor=south west] (text1) at (35.85, 23.842){$A$};

	\node[text=c640000,line cap=butt,line join=miter,line width=1.0pt,miter 
	limit=4.0,anchor=south west] (text2) at (49, 56){$D^+(A)$};

	\node[text=c640000,line cap=butt,line join=miter,line width=1.0pt,miter 
	limit=4.0,anchor=south west] (text3) at (62.906, 92.709){$\vdots$};

	\path[draw=black,line cap=butt,line join=miter,line width=1.0pt,miter 
	limit=4.0,dash pattern=on 2.0pt off 1.0pt] (28.037, 90.726) -- (28.031, 99.23);

	\path[draw=black,line cap=butt,line join=miter,line width=1.0pt,miter 
	limit=4.0,dash pattern=on 2.0pt off 1.0pt] (101.426, 90.841) -- (101.326, 
	99.23);

		%
		%
		%
		%
		%
		%
		%
		%
		%
		%
		%
		%
		%
		%
		%
		%
		%
		%
		%
		%
		%
		%
		%
		%
		%
		%
		%
	
\end{tikzpicture}
		\caption{Vertical-$\mathbb{R}^2$.}
	\end{subfigure}
	\caption{Illustration of typical future domains of dependence.}
	\label{figure:typical D+}
\end{figure}
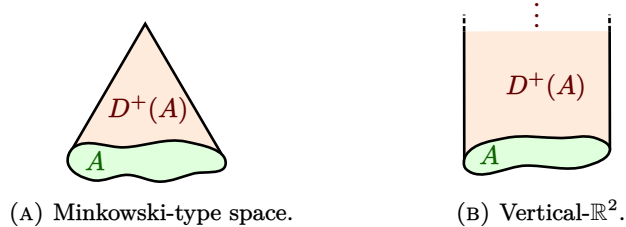
\begin{lemma}
	If $\Cov^\pm$ is a causal coverage on $X$, the domains of dependence define monads $D^\pm\colon\Opens X\to \Opens X$.
\end{lemma}
\begin{proof}
	First, axiom~\ref{axiom:Cov-unit} of a causal coverage gives $A\in \Cov^\mp(A)$, from which it follows immediately that $A\sqleq D^\pm (A)$. Next, if $D^\pm (A)\in\Cov^\mp(V)$, using ${A\in\Cov^\mp(D^\pm A)}$, we get by~\ref{axiom:Cov-transitivity} that $A\in \Cov^\mp(V)$, providing the multiplication ${D^\pm \circ D^\pm (A)\sqleq D^\pm (A)}$. 
	
	We are left to show monotonicity, for which we need the following equation:
	\[
	D^\pm (A) = \bigvee\left\{V\in\Opens X: A\wedge L^\mp (V)\in\Cov^\mp(V)\right\}.
	\]
	To prove this, denote the expression on the right hand side by $\tilde{D}^\pm (A)$. It is clear that we get an inclusion $D^\pm (A)\sqleq \tilde{D}^\pm (A)$. For the converse, it suffices to show that $A\in\Cov^\mp (\tilde{D}^\pm (A))$. To obtain this, note that~\eqref{axiom:Cov-V strong} gives
	\[
	A\wedge \bigvee\left\{V\in\Opens X: A\wedge L^\mp (V)\in\Cov^\mp(V)\right\} \in \Cov^\mp(\tilde{D}^\pm (A)).
	\]
	Further note that this region fits into a chain of inclusions $A\sqleq D^\pm (A)\sqleq \tilde{D}^\pm (A)\sqleq L^\mp \circ \tilde{D}^\pm (A)\in \Cov^\mp(\tilde{D}^\pm (A))$. Therefore, by axiom~\ref{axiom:Cov-ideal} of a causal coverage, we get $A\in\Cov^\mp (\tilde{D}^\pm (A))$, proving the equation.
	
	Now take opens $A\sqleq B$, and pick $V\in\Opens X$ such that $A\in\Cov^\mp (V)$. We get a chain of inclusions $A\sqleq B\wedge L^\mp (V)\sqleq L^\mp (V)$, so through axiom~\ref{axiom:Cov-ideal} it follows $B\wedge L^\mp (V)\in \Cov^\mp(V)$. This shows $V\sqleq \tilde{D}^\pm (B)$, and using the equation above it follows that $D^\pm (A)\sqleq D^\pm (B)$.
\end{proof}

Note that $D^\pm$ are generally not join-preserving. Typically $D^\pm(A\vee B)$ is much larger than $D^\pm(A)\vee D^\pm(B)$, cf.~the examples in \cref{figure:differences domains of dependence}.

The following shows that a causal coverage $\Cov^\pm$ is fully determined by its regions of influence and domain of dependence operators.

\begin{lemma}\label{lemma:causal coverage determined by L and D}
	Let $\Cov^\pm$ be a causal coverage on a locale $X$. Then:
	\[
	A\in \Cov^\pm(U)
	\qquad\text{if and only if}\qquad
	A\sqleq L^\pm (U)\text{ and }U\sqleq D^\mp (A).
	\]
\end{lemma}
\begin{proof}
	The `only if' direction follows by construction. Conversely, suppose that $A\sqleq L^\pm (U)$ and $U\sqleq D^\mp (A)$. Since $A\in \Cov^\pm(D^\mp (A))$ we get by \cref{lemma:Cov-pullbacks} some $B\in \Cov^\mp(U)$ such that $B\sqleq A$. From~\eqref{axiom:Cov-V strong} we know $L^\pm (U)\in \Cov^\pm (U)$, so now~\ref{axiom:Cov-ideal} gives $A\in\Cov^\pm(U)$, as desired. 
\end{proof}


\begin{remark}
	Recalling axiom~\eqref{axiom:cones give order} from \cref{section:parallel ordered locales}, this actually shows that $\Cov^\pm$ is the causal order of the ordered locale induced by the pair of monads $(L^\pm,D^\mp)$.
\end{remark}

\begin{lemma}
	We have $L^\pm\circ D^\pm = L^\pm$, and hence $D^\pm \sqleq L^\pm$. 
\end{lemma}
\begin{proof}
	Since $D^\pm$ are monads, the inclusions $L^\pm (A)\sqleq L^\pm\circ D^\pm(A)$ are clear. For the converse, recall that by \cref{lemma:causal coverage determines monads} the maps $L^\pm$ preserve all joins. Hence $L^\pm \circ D^\pm (A) = \bigvee\{L^\pm (V) : A\in \Cov^\mp (V)\}$. Now, if $A\in \Cov^\mp(V)$ we get by axiom~\ref{axiom:Cov-flip} an open $W\in \Cov^\pm(A)$ such that $V\sqleq W$. Thus $L^\pm (V)\sqleq L^\pm (W)\sqleq L^\pm \circ L^\pm (A) = L^\pm (A)$, and we may conclude $L^\pm \circ D^\pm (A) = L^\pm (A)$. Using the units of $L^\pm$ it immediately follows that $D^\pm (A)\sqleq L^\pm \circ D^\pm (A) = L^\pm (A)$. 
\end{proof}


\section{Domains of dependence in spacetime}\label{section:domains of dependence in spactime}
As mentioned, in spacetime there is already a well-established notion of domain of dependence, making use of inextendible curves. In this section we compare these to the localic versions, and see that they are generally distinct. Here we adopt~\cite[Definition~3.1]{minguzzi2019LorentzianCausalityTheory}.

\begin{definition}\label{definition:domain of dependence in spacetime}
	For a subset $A\subseteq M$ of a smooth spacetime, the \emph{future causal domain of dependence} of $A$ is the set
	\[
	\Dcaus^+(A) := \left\{x\in M:~\parbox{.4\textwidth}{every past inextendible \emph{causal} curve through $x$ intersects $A$}~\right\}.
	\]
	Analogously, the \emph{future chronological domain of dependence} of $A$ is the set
	\[
	\Dchron^+(A) := \left\{x\in M:~\parbox{.4\textwidth}{every past inextendible \emph{timelike} curve through $x$ intersects $A$}~\right\}.
	\]
\end{definition}

\begin{remark}
	Note that domains of dependence are usually mainly considered on \emph{achronal} subsets: those $S\subseteq M$ for which $I^-(S)\cap I^+(S)=\varnothing$ \cite[p105]{landsman2021FoundationsGeneralRelativity}. These are spacelike hypersurfaces on which initial data of partial differential equations are typically defined. In an ordered locale the condition $\Down U\wedge \Up U = \varnothing$ implies $U=\varnothing$, so the notion becomes trivial. We think a more developed notion of causal order $\Leq$ on the lattice of \emph{sublocales}, instead of just on $\Opens X$, will be fruitful: it will allow us to consider non-trivial achronal `sets' in an ordered locale. The development of a more general notion of ordered locale is ongoing; see also~\cite[Appendix~D]{schaaf2024TowardsPointFreeSpacetimes}.
\end{remark}

\begin{definition}
	For a smooth spacetime $M$ we define the following notions of domains of dependence for $A\subseteq M$:
	\begin{align*}
		\DCaus^+(A) &:= \bigcup\left\{U\in\Opens M: A\in\Covcaus^-(U) \right\},\\
		\DChron^+(A) &:= \bigcup\left\{U\in\Opens M: A\in\Covchron^-(U) \right\},\\
		\DLeq^+(A) &:= \bigcup\left\{U\in\Opens M: A\in\CovLeq^-(U) \right\}.
	\end{align*}
\end{definition}

Note that $\DLeq^+$ is just the localic domain of dependence from \cref{definition:domain of dependence localic} determined by the canonical causal coverage $\CovLeq^-$, and is hence a monad. We will now show how these five definitions are related but distinct.

\begin{enumerate}
	\item \emph{Bounded curves vs.~localic paths.}
	By \cref{proposition: localic coverage contains chronological coverage} we get for every ${U\in\Opens M}$ that $\Covcaus^-(U)\subseteq \Covchron^-(U)\subseteq \CovLeq^-(U)$,
	from which it immediately follows that
	\[
	\DCaus^+\subseteq \DChron^+\subseteq\DLeq^+.
	\]
	Due to the counterexample in \cref{remark:difference causal coverage localic or curve} that $\Covchron^-(U)\not\subseteq \CovLeq^-(U)$, see \cref{figure:causal coverage in spacetime}(b), the second inclusion will be strict. An example of what this can look like is in \cref{figure:differences domains of dependence}(b).
	
	\item \emph{Inextendible vs.~bounded curves.} Similarly, using \cref{lemma:past inextendible intersects then bounded intersects} and its chronological analogue we obtain inclusions
	\[
	\Dcaus^+\subseteq \DCaus^+
	\qquad\text{and}\qquad
	\Dchron^+\subseteq \DChron^+.
	\]
	But, again, these inclusions are strict by virtue of the counterexample in \cref{remark:difference causal coverage bounded or inextendible} and \cref{figure:non-intersecting inextendible curve}. This is showcased in \cref{figure:differences domains of dependence}(a).
	
	\item \emph{Chronological vs.~causal.} First, since all timelike curves are causal curves, we get an inclusion $\Dcaus^+\subseteq \Dchron^+$ of the traditional domains of dependence. It is well-known that this inclusion is strict in general, cf.~\cite[Remark~3.11]{minguzzi2019LorentzianCausalityTheory}.
	
	The counterexample in \cref{remark:causal cover contained in chronological cover} that ${\Covchron^-(U)\not\subseteq \Covcaus^-(U)}$ (see \cref{figure:non-intersecting lightlike curve}) also implies that the inclusion $\DCaus^+\subseteq \DChron^+$ is strict. To see how, one can draw a situation analogous to \cref{figure:differences domains of dependence}(b) where the curve removed is lightlike. In that case $\DCaus^+(A)$ equals the two smaller triangles as drawn, but $\DChron^+(A)$ will equal $\DLeq^+(A)$.
\end{enumerate}
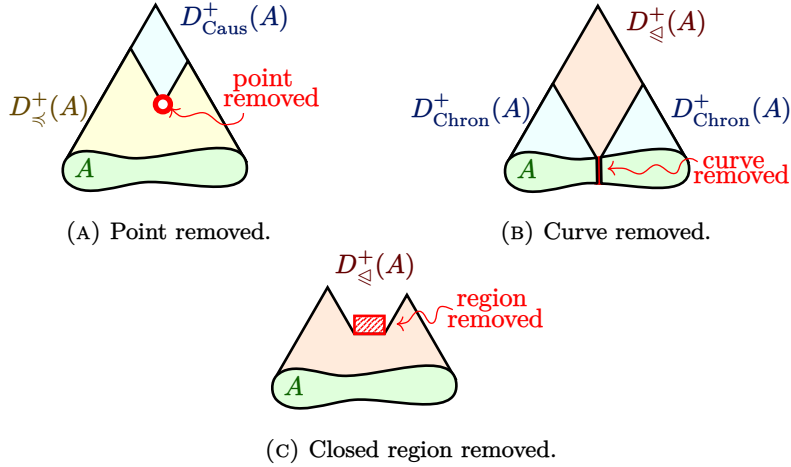
\begin{figure}[t]\centering
	\begin{subfigure}[b]{0.4\textwidth}\centering
		\definecolor{cefffff}{RGB}{239,255,255}
\definecolor{cffffdc}{RGB}{255,255,220}
\definecolor{cfe0000}{RGB}{254,0,0}
\definecolor{ce0ffdc}{RGB}{224,255,220}
\definecolor{c0d6400}{RGB}{13,100,0}
\definecolor{cfd0000}{RGB}{253,0,0}
\definecolor{c644700}{RGB}{100,71,0}
\definecolor{c001764}{RGB}{0,23,100}
\definecolor{cffebdc}{RGB}{255,235,220}
\definecolor{c640000}{RGB}{100,0,0}

\def \globalscale {1.000000}
\begin{tikzpicture}[y=.8pt, x=.8pt, yscale=\globalscale,xscale=\globalscale, every node/.append style={scale=\globalscale}, inner sep=0pt, outer sep=0pt]
	\path[draw=black,fill=cefffff,line cap=butt,line join=miter,line 
	width=1.0pt,miter limit=4.0] (20.53, 34.033) -- (58.259, 27.451) -- (105.106, 
	32.653) -- (62.42, 106.588) -- cycle;

	\path[draw=black,fill=cffffdc,line cap=butt,line join=miter,line 
	width=1.0pt,miter limit=4.0] (20.53, 34.033) -- (60.562, 27.451) -- (105.106, 
	32.653) -- (77.732, 80.066) -- (65.884, 59.545) -- (50.572, 86.067) -- (20.53,
	34.033);

	\node[text=cfe0000,line cap=butt,line join=miter,line width=1.0pt,miter 
	limit=4.0,anchor=south west] (text16) at (98.934, 65.934){$\text{point}$};

	\node[text=cfe0000,line cap=butt,line join=miter,line width=1.0pt,miter 
	limit=4.0,anchor=south west] (text17) at (94.508, 58.807){$\text{removed}$};

	\path[draw=black,fill=ce0ffdc,line cap=butt,line join=miter,line 
	width=1.0pt,miter limit=4.0] (20.53, 34.033).. controls (17.857, 31.361) and 
	(17.857, 25.367) .. (20.53, 22.695).. controls (29.216, 14.009) and (45.13, 
	21.799) .. (57.38, 22.695).. controls (73.573, 23.879) and (106.115, 13.555) 
	.. (105.569, 29.781).. controls (104.964, 47.78) and (69.706, 33.324) .. 
	(51.711, 34.033).. controls (41.325, 34.443) and (27.879, 41.383) .. (20.53, 
	34.033) -- cycle;

	\node[text=c0d6400,line cap=butt,line join=miter,line width=1.0pt,miter 
	limit=4.0,anchor=south west] (text1) at (24.069, 25.132){$A$};

	\path (20.53, 34.033) -- (62.42, 106.588);

	\path (105.106, 32.653) -- (62.42, 106.588);

	\path (65.884, 59.545) -- (50.572, 86.067);

	\path (65.884, 59.545) -- (77.732, 80.066);

	\path[draw=cfd0000,fill = white, line cap=butt,line join=miter,line width=2.0pt,miter 
	limit=4.0] (65.884, 59.545) ellipse (2.835pt and 2.835pt);

	\path[->,draw=red,line cap=butt,line join=miter,line width=0.5pt,miter limit=4.0]
	(103.926, 55.425).. controls (103.926, 55.425) and (104.913, 49.63) .. 
	(94.674, 49.943).. controls (91.534, 50.04) and (88.949, 54.786) .. (85.423, 
	55.492).. controls (81.778, 56.221) and (77.967, 52.871) .. (74.388, 53.876)..
	controls (72.401, 54.434) and (69.623, 57.832) .. (69.623, 57.832);

	\node[text=c644700,line cap=butt,line join=miter,line width=1.0pt,miter 
	limit=4.0,anchor=south west] (text11) at (-7, 48.457){$\Dcaus^+(A)$};

	\node[text=c001764,line cap=butt,line join=miter,line width=1.0pt,miter 
	limit=4.0,anchor=south west] (text12) at (73.914, 92.709){$\DCaus^+(A)$};

		%
		%
		%
		%
		%
		%
		%
		%
		%
		%
		%
		%
		%
		%
		%
		%
		%
		%
		%
		%
		%
		%
		%
		%
		%
		%
		%
	
\end{tikzpicture}
		\caption{Point removed.}
	\end{subfigure}
	\begin{subfigure}[b]{0.5\textwidth}\centering
		\definecolor{cffebdc}{RGB}{255,235,220}
\definecolor{cfe0000}{RGB}{254,0,0}
\definecolor{cefffff}{RGB}{239,255,255}
\definecolor{ce0ffdc}{RGB}{224,255,220}
\definecolor{c001764}{RGB}{0,23,100}
\definecolor{c640000}{RGB}{100,0,0}
\definecolor{c0d6400}{RGB}{13,100,0}
\definecolor{cffffdc}{RGB}{255,255,220}
\definecolor{c644700}{RGB}{100,71,0}

\def \globalscale {1.000000}
\begin{tikzpicture}[y=.8pt, x=.8pt, yscale=\globalscale,xscale=\globalscale, every node/.append style={scale=\globalscale}, inner sep=0pt, outer sep=0pt]
	\path[draw=black,fill=cffebdc,line cap=butt,line join=miter,line 
	width=1.0pt,miter limit=4.0] (20.53, 34.033) -- (62.331, 22.623) -- (63.78, 
	22.524) -- (105.106, 32.653) -- (62.42, 106.588) -- cycle;

	\path[draw=red,fill=cfe0000,line width=1.0pt] (61.527, 34.613) -- (64.77, 
	35.093) -- (64.922, 22.486) -- (61.279, 22.683) -- cycle;

	\node[text=cfe0000,line cap=butt,line join=miter,line width=1.0pt,miter 
	limit=4.0,anchor=south west] (text16) at (112.012, 29.968){$\text{curve}$};

	\node[text=cfe0000,line cap=butt,line join=miter,line width=1.0pt,miter 
	limit=4.0,anchor=south west] (text17) at (107.586, 22.84){$\text{removed}$};

	\path[draw=black,fill=cefffff,line cap=butt,line join=bevel,line 
	width=1.0pt,miter limit=4.0] (63.826, 34.751) -- (83.812, 69.536) -- (105.106,
	32.653) -- (84.722, 32.749) -- (68.719, 35.269) -- cycle;

	\path[draw=black,fill=cefffff,line cap=butt,line join=bevel,line 
	width=1.0pt,miter limit=4.0] (20.53, 34.033) -- (41.196, 69.828) -- (61.926, 
	34.431) -- (54.546, 34.033) -- cycle;

	\path[draw=black,fill=ce0ffdc,line join=bevel,line width=1.0pt] (93.007, 
	38.849).. controls (84.427, 39.007) and (73.888, 36.419) .. (63.867, 34.991) 
	-- (64.084, 22.524).. controls (81.163, 21.779) and (106.057, 15.281) .. 
	(105.57, 29.782).. controls (105.343, 36.531) and (100.243, 38.716) .. 
	(93.007, 38.849) -- cycle;

	\path[draw=black,fill=ce0ffdc,line join=bevel,line width=1.0pt] (62.096, 22.6)
	-- (61.879, 34.736).. controls (57.817, 34.157) and (54.955, 33.905) .. 
	(51.71, 34.033).. controls (41.325, 34.442) and (27.879, 41.382) .. (20.53, 
	34.033).. controls (17.857, 31.36) and (17.857, 25.367) .. (20.53, 22.695).. 
	controls (29.215, 14.009) and (45.13, 21.799) .. (57.381, 22.695).. controls 
	(59.103, 22.821) and (62.096, 22.6) .. (62.096, 22.6) -- cycle;

	\path (20.53, 34.033) -- (62.42, 106.588);

	\path (105.106, 32.653) -- (62.42, 106.588);

	\path (65.884, 59.545) -- (50.572, 86.067);

	\path (65.884, 59.545) -- (77.732, 80.066);

	\path[->,draw=red,line cap=butt,line join=miter,line width=0.5pt,miter limit=4.0]
	(109.583, 32.752).. controls (109.583, 32.752) and (107.335, 35.095) .. 
	(104.293, 34.081).. controls (99.984, 32.644) and (101.673, 27.973) .. 
	(98.855, 27.302).. controls (95.689, 26.547) and (93.218, 31.485) .. (89.966, 
	31.347).. controls (87.068, 31.223) and (85.268, 27.351) .. (82.379, 27.083)..
	controls (79.732, 26.837) and (77.453, 29.239) .. (74.825, 29.644).. controls
	(72.235, 30.043) and (66.966, 29.478) .. (66.966, 29.478);

	\node[text=c001764,line cap=butt,line join=miter,line width=1.0pt,miter 
	limit=4.0,anchor=south west] (text11) at (-26, 49.946){$\DChron^+(A)$};

	\node[text=c001764,line cap=butt,line join=miter,line width=1.0pt,miter 
	limit=4.0,anchor=south west] (text11-6) at (97, 49.615){$\DChron^+(A)$};

	\node[text=c640000,line cap=butt,line join=miter,line width=1.0pt,miter 
	limit=4.0,anchor=south west] (text12) at (75.823, 89.221){$\DLeq^+(A)$};

	\node[text=c0d6400,line cap=butt,line join=miter,line width=1.0pt,miter 
	limit=4.0,anchor=south west] (text1-0) at (24.069, 25.132){$A$};

		%
		%
		%
		%
		%
		%
		%
		%
		%
		%
		%
		%
		%
		%
		%
		%
		%
		%
		%
		%
		%
		%
		%
		%
		%
		%
		%
	
\end{tikzpicture}
		\caption{Curve removed.}
	\end{subfigure}
	\begin{subfigure}[b]{0.7\textwidth}\centering
		\definecolor{cffebdc}{RGB}{255,235,220}
\definecolor{cfe0000}{RGB}{254,0,0}
\definecolor{ce0ffdc}{RGB}{224,255,220}
\definecolor{c0d6400}{RGB}{13,100,0}
\definecolor{c640000}{RGB}{100,0,0}
\definecolor{cefffff}{RGB}{239,255,255}
\definecolor{cffffdc}{RGB}{255,255,220}
\definecolor{c001764}{RGB}{0,23,100}
\definecolor{c644700}{RGB}{100,71,0}

\def \globalscale {1.000000}
\begin{tikzpicture}[y=.8pt, x=.8pt, yscale=\globalscale,xscale=\globalscale, every node/.append style={scale=\globalscale}, inner sep=0pt, outer sep=0pt]
	\path[draw=black,fill=cffebdc,line cap=butt,line join=miter,line 
	width=1.0pt,miter limit=4.0] (20.53, 34.033) -- (63.78, 27.545) -- (105.106, 
	32.653) -- (82.204, 72.322) -- (71.554, 53.876) -- (57.38, 53.876) -- (44.683,
	75.868) -- (20.53, 34.033);

	\path[draw=black,line cap=butt,line join=miter,line width=0.5pt,miter 
	limit=4.0,dash pattern=on 0.5pt off 2.0pt] (20.53, 34.033) -- (44.683, 75.868)
	-- (57.38, 53.876) -- (71.554, 53.876) -- (82.204, 72.322) -- (105.106, 
	32.653);

	\node[text=cfe0000,line cap=butt,line join=miter,line width=1.0pt,miter 
	limit=4.0,anchor=south west] (text16) at (105.751, 64.44){$\text{region}$};

	\node[text=cfe0000,line cap=butt,line join=miter,line width=1.0pt,miter 
	limit=4.0,anchor=south west] (text17) at (101.325, 57.313){$\text{removed}$};

	\path[draw=black,fill=ce0ffdc,line cap=butt,line join=miter,line 
	width=1.0pt,miter limit=4.0] (20.53, 34.033).. controls (17.857, 31.361) and 
	(17.857, 25.367) .. (20.53, 22.695).. controls (29.216, 14.009) and (45.13, 
	21.799) .. (57.38, 22.695).. controls (73.573, 23.879) and (106.115, 13.555) 
	.. (105.569, 29.781).. controls (104.964, 47.78) and (69.706, 33.324) .. 
	(51.711, 34.033).. controls (41.325, 34.443) and (27.879, 41.383) .. (20.53, 
	34.033) -- cycle;

	\node[text=c0d6400,line cap=butt,line join=miter,line width=1.0pt,miter 
	limit=4.0,anchor=south west] (text1) at (24.069, 25.132){$A$};

	\path[->,draw=red,line cap=butt,line join=miter,line width=0.5pt,miter limit=4.0]
	(103.889, 65.296).. controls (103.889, 65.296) and (99.785, 67.535) .. 
	(97.954, 66.66).. controls (95.371, 65.424) and (96.97, 59.644) .. (94.195, 
	58.933).. controls (91.23, 58.174) and (90.942, 66.283) .. (86.704, 64.245).. 
	controls (84.744, 63.303) and (84.168, 60.513) .. (82.263, 59.464).. controls 
	(80.786, 58.651) and (77.291, 58.536) .. (77.291, 58.536);

	\node[text=c640000,line cap=butt,line join=miter,line width=1.0pt,miter 
	limit=4.0,anchor=south west] (text11) at (48.057, 76.844){$\DLeq^+(A)$};

	\path[pattern={Lines[angle = 45,distance = 1.5pt,line width = .5pt]}, pattern color=red,draw=red,line cap=butt,line join=miter,line width=1.0pt,miter limit=4.0]
	(57.38, 62.38) -- (71.554, 62.38) -- (71.554, 53.876) -- (57.38, 53.876) -- 
	cycle;

		%
		%
		%
		%
		%
		%
		%
		%
		%
		%
		%
		%
		%
		%
		%
		%
		%
		%
		%
		%
		%
		%
		%
		%
		%
		%
		%
	
\end{tikzpicture}
		\caption{Closed region removed.}
	\end{subfigure}
	\caption{Illustration of differences between localic and curve-wise domains of dependence in Minkowski spaces with either a point or a larger closed region removed. In (b) a curve is removed from $A$ (not the spacetime).}
	\label{figure:differences domains of dependence}
\end{figure}

In conclusion, the domains $\Dcaus^+$ and $\Dchron^+$ defined using inextendible curves differ from the domains $\DCaus^+$ and $\DChron^+$ defined in terms of bounded curves when we remove a point from the spacetime. The domains $\DCaus^+$ and $\DChron^+$ differ from the localic domain $\DLeq^+$ as soon as we remove a curve from $A$. In summary, we have:
\[
\begin{tikzcd}[cramped,sep=small]
	{\Dcaus^+} & {\Dchron^+} \\
	{\DCaus^+} & {\DChron^+} & {\DLeq^+.}
	\arrow["\subsetneq"{marking, allow upside down}, draw=none, from=1-1, to=1-2]
	\arrow["\subsetneq"{marking, allow upside down}, draw=none, from=1-1, to=2-1]
	\arrow["\subsetneq"{marking, allow upside down}, draw=none, from=2-1, to=2-2]
	\arrow["\subsetneq"{marking, allow upside down}, draw=none, from=1-2, to=2-2]
	\arrow["\subsetneq"{marking, allow upside down}, draw=none, from=2-2, to=2-3]
\end{tikzcd}
\]
\section{Conclusions}\label{section:conclusions}
We conclude the paper by discussing some avenues for future work.
\subsection{Deterministic sheaves}
The incarnation of causal coverages as a type of Grothendieck topology (\cref{theorem:causal grothendieck topology}) raises the question of \emph{sheaves} \cite{maclane1994SheavesGeometryLogic}. Recall that a \emph{presheaf} on a locale~$X$ is a functor $F\colon \Opens X^\op\to \Set$. That is, every region $U\in\Opens X$ is assigned a set $F(U)$, thought of `local data' living on $U$. If $U\sqleq V$ then we get a \emph{restriction function} $(-)|_U\colon F(V)\to F(U)$, interpreted as forgetting all data in $V$ that is not part of $U$. An important example is the presheaf of real-valued functions on a topological space~$S$:
\[
C(-,\mathbb{R})\colon \Opens S^\op\longrightarrow \Set.
\]
A region $U\in\Opens S$ gets sent to the set of continuous real-valued functions $C(U,\mathbb{R})$, and a real-valued function $f\colon V\to \mathbb{R}$ literally restricts to $f|_U$.

\emph{Sheaves} are presheaves where compatible `local data' amalgamates uniquely into `global data'. The presheaf $C(-,\mathbb{R})$ is a sheaf: if functions $f_i\colon U_i\to\mathbb{R}$ are defined on a family of opens $(U_i)_{i\in I}$ that cover some region $U\in\Opens S$, and if this family is compatible in the sense that they agree on all the intersections, then there exists a unique function $f\in C(U,\mathbb{R})$ that recovers the original family by restricting to the subregions: $f|_{U_i}=f_i$.
More generally, a presheaf $F$ is a sheaf if any compatible family defined on a covering sieve has a unique amalgamation on the covered object. The main task of a Grothendieck topology to is determine which presheaves are sheaves. See \cite[\S III.4]{maclane1994SheavesGeometryLogic} for details.

Extrapolating this behaviour, an appropriate sheaf condition with respect to the causal coverage relation on an ordered locale can instead be interpreted as:
\begin{quotation}
	\emph{`past data' evolves deterministically into `future data'.}
\end{quotation}
Since the future domain of dependence $D^+(A)$ is covered by $A$ from the past, $A\in\Cov^-(D^+(A))$, the sheaf condition in particular implies that any data living on $A$ has to propagate deterministically to $D^+(A)$. Domains of dependence in relativity theory are originally already used in this way to study how solutions to the Einstein equations evolve on Cauchy surfaces, see \cite[\S 7]{hawking1973LargeScaleStructure} for an overview. This raises the question if results and conditions involving domains of dependence such as \emph{global hyperbolicity} may be restated sheaf-theoretically. 

An archetypal example of a sheaf with respect to the causal coverage relation would be the following. 

\begin{example}
	Consider Minkowski space $\mathbb{R}^{n}$, together with the presheaf of solutions to the wave equation $F\colon (\Opens \mathbb{R}^n)^\op \to \Set$, defined by
	\[
	F(U):= \{\text{solutions to the wave equation on $U$}\},
	\]
	and where $F(U\subseteq V)$ is just the restriction map of functions. This should define a deterministic sheaf, since we know from the theory of partial differential equations that initial data on an open region $U$ uniquely determines a solution to the wave equation on $D^+(U)$~\cite[\S 2.4]{evans2022PartialDifferentialEquations}.
\end{example}

For technical reasons it is not clear how to formulate a precise sheaf condition for the modified definition of Grothendieck topologies in \cref{definition:modified grothendieck topology}. 
A good definition of causal sheaf will also transfer results from causality and relativity theory to concurrency in computer science~\cite{lamport1978TimeClocksOrdering,enriquemoliner2020TensorTopology,soaresbarbosa2023SheafRepresentationMonoidal}.
Some remarks are in \cite[\S 9.6]{schaaf2024TowardsPointFreeSpacetimes}, but we leave this development to future work.

\subsection{Holes in spacetime}\label{section:holes in spacetime}
As briefly mentioned in the \nameref{section:introduction}, there is a philosophical problem of \emph{hole-freeness} of spacetimes \cite{krasnikov2009EvenMinkowskiSpace,manchak2009SpacetimeHolefree}. An accessible overview of this problem is in \cite{robertsNotesOnHoles}.
Adopting the definition on \cite[p3]{krasnikov2009EvenMinkowskiSpace}, a spacetime $M$ is called \emph{hole-free} if for any achronal hypersurface $S\subseteq M$ and any isometric embedding $\pi$ of an open neighbourhood $U$ of the future domain of dependence $\Dchron^+(S)$ into another spacetime $M'$, we have ${\pi(\Dchron^+(S)) = \Dchron^+(\pi(S))}$. It is easy to see that the examples in \cref{figure:differences domains of dependence}(a) and (c) provide spacetimes that are not hole-free, since they can be embedded into Minkowski space, wherein the domains of dependence suddenly~`grow'.
However, the paper \cite{krasnikov2009EvenMinkowskiSpace} shows that, paradoxically, even Minkowski space is not hole-free in this sense. Does using the \emph{localic} domain of dependence $\DLeq^+$ solve this problem?

\printbibliography
\end{document}